\documentclass{article}

\newif\ifdraft
\draftfalse

\usepackage{pdfpages}

\usepackage{todonotes}

\makeatletter
\providecommand{\leftsquigarrow}{%
  \mathrel{\mathpalette\reflect@squig\relax}%
}
\newcommand{\reflect@squig}[2]{%
  \reflectbox{$\m@th#1\rightsquigarrow$}%
}
\makeatother

\usepackage[a4paper,centering,left=2cm,top=2.5cm]{geometry}
\usepackage{xspace}
\usepackage{alltt}

\usepackage{multirow,bigdelim} %% bigdelim for ldelim and rdelim

\usepackage{xcolor}

\usepackage{tikz}
\usetikzlibrary{arrows,decorations,backgrounds}
\usetikzlibrary{shapes,calc,backgrounds,shadows,fit,patterns}
\usetikzlibrary{decorations.markings}

\colorlet{rolecolor}{black}
\colorlet{indcolor}{black}
\colorlet{homcolor}{gray}

\colorlet{tboxcolor}{blue}
\colorlet{aboxcolor}{yellow}

\pgfdeclarelayer{background}
\pgfdeclarelayer{foreground}
\pgfsetlayers{background,main,foreground}

\pgfmathparse{atan2(0,1)}
\ifdim\pgfmathresult pt=0pt % atan2(y, x)
  \tikzset{declare function={atanXY(\x,\y)=atan2(\y,\x);atanYX(\y,\x)=atan2(\y,\x);}}
\else                       % atan2(x, y)
  \tikzset{declare function={atanXY(\x,\y)=atan2(\x,\y);atanYX(\y,\x)=atan2(\x,\y);}}
\fi

\tikzset{ %
  abox/.style={cylinder, shape border rotate=90, aspect=0.2, outer sep=0.05cm,
    draw, fill=aboxcolor},%
  aboxl/.style={abox, draw=aboxcolor!60, fill=aboxcolor!20},%
  aboxr/.style={abox},%
  abox3d/.style={abox, drop shadow, draw=aboxcolor!90!black, %
    left color=aboxcolor!70, right color=aboxcolor!80, middle color=white,%
    append after command={ let%
      \p{cyl@center} = ($(\tikzlastnode.before top)!0.5!(\tikzlastnode.after top)$), %
      \p{cyl@x} = ($(\tikzlastnode.before top)-(\p{cyl@center})$), %
      \p{cyl@y} = ($(\tikzlastnode.top) -(\p{cyl@center})$) in %
      (\p{cyl@center}) edge[draw=none, fill=aboxcolor!50, to path={ ellipse%
        [x radius=veclen(\p{cyl@x})-1\pgflinewidth-0.05cm, %
        y radius=veclen(\p{cyl@y})-1\pgflinewidth-0.05cm, rotate=atanXY(\p{cyl@x})]}]
      () }%
  },%
  tbox/.style={thick, outer sep=0.05cm, inner sep=0.2cm, draw, fill=tboxcolor!50},%
  tboxl/.style={tbox, draw=tboxlcolor, fill=tboxlcolor!40}, %
  tboxr/.style={tbox, draw=tboxrcolor, fill=tboxrcolor!40}, %
  tbox3d/.style={tbox, drop shadow={opacity=0.35},
    thin, draw=tboxcolor!50, top color=tboxcolor!20, bottom color=tboxcolor!50
  },%
  tbox3dl/.style={tbox, draw=tboxlcolor!60, inner color=tboxrcolor!20, outer
    color=tboxrcolor!40},%
  tbox3dr/.style={tbox, draw=tboxrcolor!85, inner color=tboxrcolor!45, outer
    color=tboxrcolor!65},%
  kbbox/.style={draw, thick, outer sep=0cm, inner sep=0.2cm},%
  kbbox3d/.style={kbbox, outer sep=0.05cm, rounded corners},%
  entity/.style={draw, fill=white, thick, inner xsep=0.2cm, outer sep=0.02cm},%
  relation/.style={draw, diamond, aspect=3, fill=white, thick, inner
    ysep=0.02cm, drop shadow},%
  attribute/.style={draw, circle, fill=white, thick, inner sep=0.07cm},%
  umlclass/.style={rectangle split, rectangle split parts=3, rectangle split
    part align={center, left, left}, draw, thick, drop shadow={opacity=0.3},
    inner ysep=0.05cm, %
    rectangle split empty part height=0.2cm, fill=gray!50},%
  umlclassabstract/.style={rectangle split, rectangle split parts=3, rectangle
    split part align={center, left, left}, draw, thick, drop shadow={opacity=0.3},
    inner ysep=0.05cm, %
    rectangle split empty part height=0.0cm, fill=gray!50},%
  individual/.style={draw, circle, outer sep=0.05cm, inner sep=0.08cm, thick,
    color=indcolor, text=indcolor}, %
  constant/.style={fill}, %
  region/.style={rounded corners, color=mygray, fill=mygray!30, text=black}, %
  isa/.style={-open triangle 60, thick},%
  erisa/.style={decoration={markings,mark=at position 1 with {\arrow[line
        width=0.1cm]{stealth}}}, postaction={decorate}},%
  assoc/.style={thick},%
  role/.style={-stealth', thick, color=rolecolor},%
  disj/.style={triangle 60-triangle 60, fill=white, ultra thin, postaction={decorate,
    decoration={markings, mark=at position 0.45 with {\draw[-,solid]
        (-2pt,-2pt) -- (2pt, 2pt); } }}},%
  trans/.style={-stealth', semithick},%
  homomor/.style={-stealth, dashed, thick, color=homcolor},%
  soledge/.style={dashed,-latex},%
  path/.style={decorate, decoration={zigzag, segment length=20pt,
      amplitude=2pt},->, ultra thin},%
  mapping/.style={decorate, decoration={coil, aspect=0, segment length=20pt,
      amplitude=1pt},-latex, ultra thin, mydarkgray},%
  disjmap/.style={decorate, decoration={coil, aspect=0, segment length=20pt,
      amplitude=1pt},-latex, ultra thin, mydarkgray, postaction={decorate,
    decoration={markings, mark=at position 0.55 with {\draw[-,solid]
        (-2pt,-2pt) -- (2pt, 2pt); } }}},%
  crossout/.style={decorate, decoration={markings, mark=at position 0.55 with
      {\draw[-,solid,thick,myred] (-2pt,-2pt) -- (2pt, 2pt); %
        \draw[-,solid,thick,myred] (-2pt,2pt) -- (2pt, -2pt);} }},%
  isometricYXZ/.style={x={(1cm,0cm)}, y={(-1.299cm,-0.75cm)}, z={(0cm,1cm)}},%
}

\def \tikzdots[#1]{
  \begin{scope}[shift={#1}, text=black, inner sep=0.05cm]
    \node at (0,0.15) {\per}; \node {\per}; \node at (0,-0.15) {\per};
  \end{scope}
}

\def\pgfdecoratedcontourdistance{0pt}

\pgfkeys{/pgf/decoration/contour distance/.code={%
    \pgfmathparse{#1}%
    \let\pgfdecoratedcontourdistance=\pgfmathresult}%
}

\pgfdeclaredecoration{contour lineto}{start}
{
    \state{start}[next state=draw, width=0pt]{
        \pgfpathmoveto{\pgfpoint{0pt}{\pgfdecoratedcontourdistance}}%
    }
    \state{draw}[next state=draw, width=\pgfdecoratedinputsegmentlength]{
        \pgfmathparse{-\pgfdecoratedcontourdistance*cot(-\pgfdecoratedangletonextinputsegment/2+90)}%
        \let\shorten=\pgfmathresult%
        \pgfpathlineto{\pgfpoint{\pgfdecoratedinputsegmentlength+\shorten}{\pgfdecoratedcontourdistance}}%
    }
    \state{final}{
        \pgfpathlineto{\pgfpoint{\pgfdecoratedinputsegmentlength}{\pgfdecoratedcontourdistance}}%
    }
}

\def\aboxl[#1,#2,#3,#4,#5]#6{%
  \node[draw, cylinder, alias=cyl, shape border rotate=90, aspect=#3, %
  minimum height=#1, minimum width=#2, outer sep=-0.5\pgflinewidth, %
  color=aboxcolor!60, left color=aboxcolor!30, right color=aboxcolor!40, middle
  color=white] (#4) at #5 {};%
  \node at #5 {#6};%
  \fill [left color=aboxcolor!20, right color=aboxcolor!20] let \p1 =
  ($(cyl.before top)!0.5!(cyl.after top)$), \p2 = (cyl.top), \p3 = (cyl.before
  top), \n1={veclen(\x3-\x1,\y3-\y1)}, \n2={veclen(\x2-\x1,\y2-\y1)} in (\p1)
  ellipse (\n1 and \n2); }

\usepackage{pgfplots} % For creating the plot
\usepackage{nameref}
\usepackage{varioref}

\makeatletter
\newcounter{desc@cont}

\labelformat{desc@cont}{\@currentlabelname}
\newcommand{\myitem}[1]{\item[#1:]\def\@currentlabelname{#1}\refstepcounter{desc@cont}}
\newcommand{\axpart}[1]{#1\def\@currentlabelname{#1}\refstepcounter{desc@cont}}

\makeatother

\usepackage{pifont}

\usepackage{times}
\usepackage{multirow}
\usepackage{hhline}
\usepackage{latexsym}
\usepackage{amssymb,mathtools,amsthm}
\usepackage{url}
\usepackage{booktabs}
\usepackage{rotating}
\usepackage{paralist,enumerate}
\usepackage{stmaryrd}
\usepackage{verbatim}
\usepackage{fancyvrb}

\usepackage{graphicx}
\usepackage{caption}
\usepackage{subcaption}

\usepackage{algorithm}

\usepackage{placeins}

\def\soledge[#1,#2,#3,#4];{
        \draw[dashed,-latex,#4] (#1) -- #3 (#2);
}

\usepackage{stfloats}
\usepackage{float}

\usepackage{listings}
\theoremstyle{example}
\newtheorem{theorem}{Theorem}[section]
\newtheorem{definition}[theorem]{Definition}
\newtheorem{lemma}[theorem]{Lemma}
\newtheorem{corollary}[theorem]{Corollary}

\theoremstyle{definition}
\newtheorem{example}[theorem]{Example}

\newcommand{\qedboxfull}{\vrule height 5pt width 5pt depth 0pt}
\newcommand{\qedfull}{\hfill{\qedboxfull}}

\ifdraft
\usepackage[ulem=normalem]{mychanges}
\else
\usepackage[final,ulem=normalem]{mychanges}
\fi
\setauthormarkup[right]{{\scriptsize[#1]}}
\definechangesauthor{DC}{blue}
\definechangesauthor{DL}{red}
\definechangesauthor{RW}{green}
\definechangesauthor{GX}{purple}

\definecolor{midgrey}{rgb}{0.5,0.5,0.5}
\definecolor{darkred}{rgb}{0.7,0.1,0.1}

\newcommand{\fstyle}[1]{\mathit{#1}}

\newcommand{\scan}{{\fstyle{scan}}}
\newcommand{\hjoin}{{\fstyle{hjoin}}}
\newcommand{\mjoin}{{\fstyle{mjoin}}}
\newcommand{\mat}{{\fstyle{mat}}}
\newcommand{\jucq}{{\fstyle{jucq}}}
\newcommand{\ujucq}{{\fstyle{ujucq}}}

\newcommand{\setNoFancy}[1]{\left\{ #1 \right\}}
\newcommand{\setBNoFancy}[2]{\{ #1 \mathrel{}|\mathrel{} #2\}}
\newcommand{\vecSet}[1]{\mathbf{#1}}
 
\newcommand{\jp}{\fstyle{se}} %% Join-point
\newcommand{\jc}{\fstyle{ea}} %% Join-cluster
\newcommand{\MAux}[1]{\mathcal{#1}^{\fstyle{aux}}}
\newcommand{\qAux}[1]{{#1}^{\fstyle{aux}}}
\newcommand{\sign}{\fstyle{sign}} %% Signature of an atom

\newcommand{\aux}{\fstyle{aux}}

\newcommand{\myPar}[1]{\medskip\noindent\textit{\large{}#1}}

\newcommand{\OWLQL}{\textrm{OWL\,2\,QL}\xspace}

\newcommand{\adom}{\fstyle{adom}}

\newcommand{\A}{\mathcal{A}}

\newcommand{\D}{\mathcal{D}}
\renewcommand{\S}{\mathcal{S}}
\newcommand{\T}{\mathcal{T}}

\renewcommand{\O}{\mathcal{O}}

\newcommand{\M}{\mathcal{M}}

\newcommand{\ontop}{\textsl{Ontop}\xspace}

\newcommand{\ignore}[1]{}

\newcommand{\dlliteR}{\textit{DL-Lite}$_{\mathcal{R}}$\xspace}

\newcommand{\ISA}{\sqsubseteq}

\newcommand{\p}{\textsc{P}\xspace}

\newcommand{\tableRowBreak}{\\}

\newcommand{\tableColSpace}{@{\hspace{1.5em}}}

\newcommand{\C}{\mathcal{C}}

\definecolor{dark-gray}{gray}{0.3}

\newcommand{\textowl}[1]{{\texttt{\small #1}}}
\newcommand{\textsql}[1]{#1}

\newcommand{\unf}{\fstyle{unf}}
\newcommand{\cq}{\fstyle{cq}}
\newcommand{\ucq}{\fstyle{ucq}}

\newcommand{\cover}{\fstyle{cover}}
\newcommand{\unfopt}{\fstyle{unf}\!_{\fstyle{opt}}}
\newcommand{\rew}{\fstyle{rew}}
\newcommand{\Aux}{\fstyle{Aux}}

\renewcommand{\split}{\fstyle{split}}
\newcommand{\wrap}{\fstyle{wrap}}

\newcommand{\atm}{\fstyle{atm}}
\newcommand{\source}{\fstyle{source}}
\newcommand{\target}{\fstyle{target}}

\newcommand{\dist}{\fstyle{dist}}

\newcommand{\head}{\fstyle{head}}

\newcommand{\pred}{\fstyle{pred}}

\newcommand{\setB}[2]{\left\{ #1 \mathrel{}\middle|\mathrel{} #2\right\}}

\newcommand{\set}[1]{\left\{ #1 \right\}}

\renewcommand{\vec}[1]{\mathbf{#1}}

\newcommand{\distSet}[1]{|#1|}
\newcommand{\restBag}[2]{#1|_{#2}} %% \restBag{T}{L} T|_L
\newcommand{\restSet}[2]{\pi_{#2}(#1)} %% \restSet{T}{L} \pi_L(T)
\newcommand{\distEst}[3][\D]{\fstyle{dist}_{#1}(#2,#3)}

\newcommand{\capren}{\mathop{\Cap}}
\newcommand{\occ}[2][\cdot]{#2_{\scriptscriptstyle[#1]}}
\newcommand{\fv}[1][\D]{\fstyle{fv}_{#1}}
\newcommand{\cardEst}[1][\D]{\fstyle{ce}_{#1}}

\newcommand{\cert}{\fstyle{cert}}

\newcommand{\fragu}[1][f]{q_{|_{#1}}^{\fstyle{ucq}(\T)}}
\newcommand{\frag}[1][f]{q_{|_{#1}}}

\newcommand{\cmath}[1]{{\small\[#1\]}}

\newcommand{\eval}{\fstyle{eval}}
\newcommand{\can}[1]{#1}%{\mathcal{\I_{#1}}}

\newcommand{\Ldelim}[1]{\ldelim\{{#1}{3mm}}
\newcommand{\Rdelim}[1]{\rdelim\}{#1}{3mm}}
\newenvironment{mappings}
{
  \begin{array}{l l l l}
  } 
  { \end{array} 
}

\title{Cost-Driven Ontology-Based Data Access\\ (Extended Version)}

\author{Davide Lanti, Guohui Xiao, Diego Calvanese\\
 Faculty of Computer Science\\
 Free University of Bozen-Bolzano, Bolzano, Italy\\
 \texttt{\{dlanti,xiao,calvanese\}@inf.unibz.it}}
\date{}

\begin{document}

\maketitle

\sloppy

\begin{abstract}
  In \emph{ontology-based data access} (OBDA), users are provided with a
  conceptual view of a (relational) data source that abstracts away details
  about data storage.  This conceptual view is realized through an
  \emph{ontology} that is connected to the data source through declarative
  \emph{mappings}, and query answering is carried out by translating the user
  queries over the conceptual view into SQL queries over the data source.
  Standard translation techniques in OBDA try to transform the user query into
  a union of conjunctive queries (UCQ), following the heuristic argument that
  UCQs can be efficiently evaluated by modern relational database engines.  In
  this work, we show that translating to UCQs is not always the best choice,
  and that, under certain conditions on the interplay between
  the
  % components of an OBDA instance, namely
  ontology, the mappings, and the statistics of the data, alternative
  translations can be evaluated much more efficiently.  To find the best
  translation, we devise a cost model together with a novel cardinality
  estimation that takes into account all such OBDA components.  Our experiments
  confirm that
  \begin{inparaenum}[\itshape (i)]
  \item alternatives to the UCQ translation might produce queries
  that are orders of magnitude more efficient, and
\item the cost model we propose is faithful to the actual query evaluation cost, and hence is well suited to select the best translation.
  \end{inparaenum}
\end{abstract}

\section{Introduction}

An important research problem in Big Data is how to provide end-users with
effective access to the data, while relieving them from being aware of details
about how the data is organized and stored. The paradigm of Ontology-based Data
Access (OBDA)~\cite{PLCD*08} addresses this problem by presenting to the
end-users a convenient view of the data stored in a relational database. This
view is in the form of a \emph{virtual} \emph{RDF
 graph}~\cite{W3Crec-RDF-Concepts} that can be queried through
\emph{SPARQL}~\cite{W3Crec-SPARQL-1.1-query}.  Such virtual RDF graph is
realized by means of the \emph{TBox of an \OWLQL ontology}~\cite{W3Crec-OWL2-Profiles} that is connected to the data source
through declarative \emph{mappings}.  Such mappings associate to each
predicate in the TBox (i.e., a concept, role, or attribute) a SQL query over
the data source, which intuitively specifies how to populate the predicate from
the data extracted by the query.

Query answering in this setting is not carried out by actually materialising
the data according to the mappings, but rather by first \emph{rewriting} the
user query with respect to the TBox, and then \emph{translating} the rewritten
query into a SQL query over the data.  In state-of-the-art OBDA
systems~\cite{ontop-system}, such SQL translation is the result of
\emph{structural optimizations}, which aim at obtaining a \emph{union of
 conjunctive queries} (UCQ).  Such an approach is claimed to be effective
because
\begin{inparaenum}[\itshape (i)]
\item joins are over database values, rather than over URIs constructed by
  applying mapping definitions;
\item joins in UCQs are performed by directly accessing (usually, indexed)
  database tables, rather than materialized and non-indexed intermediate views.
\end{inparaenum}
However, the requirement of generating UCQs comes at the cost of an exponential
blow-up in the size of the user query.

A more subtle, \emph{sometimes} critical issue, is that the UCQ structure
accentuates the problem of redundant data, which is particularly severe in OBDA
where the focus is on retrieving \emph{all} the answers implied by the data and
the TBox: each CQ in the UCQ can be seen as a different attempt of enriching
the set of retrieved answers, without any guarantee on whether the attempt will
be successful in retrieving new results.  In fact, it was already observed
in~\cite{BOSX13} that generating UCQs is sometimes counter-beneficial (although
that work was focusing on a substantially different topic).

As for the rewriting step, Bursztyn et al.~\cite{BursztynGM15b,BursztynGM15}
have investigated a space of alternatives to UCQ perfect rewritings, by considering
\emph{joins of UCQs} (JUCQs), and devised a cost-based algorithm to select the
best alternative. However, the scope of their work is limited to the simplified
setting in which there are no mappings and the extension of the predicates in
the ontology is directly stored in the database. Moreover, they use their
algorithm in combination with traditional cost models from the database
literature of query evaluation costs, which, according to their experiments,
provide estimations close to the native ones of the PostgreSQL database engine.

In this work we study the problem of alternative translations in the general
setting of OBDA, where the presence of mappings needs to be taken into account.
To do so, we first study the problem of translating JUCQ perfect rewritings, such as
those from~\cite{BursztynGM15b}, into SQL queries that preserve the JUCQ
structure while maintaining property~\textit{(i)} above, i.e., the ability of
performing joins over database values, rather than over constructed URIs.
% from the application of mappings definitions.
% Such translations in general produce \emph{unions of JUCQs} (UJUCQs).
We also devise a cost model based on a \emph{novel cardinality estimation}, for
estimating the cost of evaluating a translation of a UCQ or JUCQ over the
database. The novelty in our cardinality estimation is that it exploits the
interplay between the components of an OBDA instance, namely ontology,
mappings, and statistics of the data, so as to better estimate the number of
non-duplicate answers.

We carry out extensive and in-depth experiments based on a synthetic scenario
built on top of the \emph{Winsconsin Benchmark}~\cite{DeWitt93}, a widely
adopted benchmark for databases, so as to understand the trade-off between a
translation of UCQs and JUCQs.
% clearly identify which factors, like redundancy or selectivity for query
% operators, influence the choice of the best alternative.
In these experiments we observe that:
\begin{inparaenum}[\itshape (i)]
\item factors such as the number of mapping assertions, also affected by the
  number of axioms in the ontology, and
  the number of redundant answers are the main factors for deciding which
  translation to choose;
\item the cost model we propose is faithful to the actual query evaluation
  cost, and hence is well suited to select the best alternative translation of
  the user query;
\item the cost model implemented by PostgreSQL performs surprisingly poorly in
  the task of estimating the best translation, and is significantly
  outperformed by our cost model.  The main reason for this is that PostgreSQL
  fails at recognizing when different translations are actually equivalent, and
  may provide for them cardinality estimations that differ by several orders of
  magnitude.
\end{inparaenum}

In addition, we carry out an evaluation on a real-world scenario based on the
NPD benchmark for OBDA~\cite{LRXC15}.  Also in these experiments we confirm
that alternative translations to the UCQ one may be more efficient, and that
the same factors already identified in the Winsconsin experiments determine
which choice is best.

The rest of the paper is structured as follows.  Section~\ref{s:preliminaries}
introduces the relevant technical notions underlying OBDA.
Section~\ref{s:background} place our paper with respect to the field
literature.  Section~\ref{c:planning:s:cover-based-trans} provides our
characterization for SQL translations of JUCQs.
Section~\ref{c:planning:s:unf-est} presents our novel model for cardinality
estimation, and Section~\ref{c:planning:s:unf-cost-model} the associated cost
model.  Section~\ref{c:planning:s:eval} provides the evaluation of the cost
model on the Wisconsin and NPD Benchmarks.
Section~\ref{c:planning:s:conclusion} concludes the paper.

\ignore{
We sometimes find convenient to use the $\Rightarrow$ symbol to abbreviate \emph{If, then, else} sentences, as well as logical symbols $\forall$, $\exists$ or $\land$, $\lor$ to abbreviate the text versions of \emph{for all, exists, and, or}. For instance, we sometimes abbreviate sentences such as

\begin{center}
``For all sets $A,B,C$ in $\C$, it holds that if $A=B$ and $B=C$, then $A=C$''
\end{center}
as
\[\forall A,B,C \in \C:A=B \land B=C \Rightarrow A=C.\]
}
\section{Preliminaries}\label{s:preliminaries}

In this section we introduce basic notions and notation necessary for the
technical development of this work.

From now on, we will denote tuples in \textbf{bold} font, e.g., $\vec{x}$ is a
tuple, and when convenient we treat tuples as sets.  Given a predicate symbol
$P$, a tuple $\vec{f}$ of function symbols, and a tuple $\vec{x}$ of variables,
we denote by $P(\vec{f}(\vec{x}))$ an atom where each argument is of the form
$g(\vec{y})$, where $g\in \vec{f}$ and $\vec{y} \subseteq \vec{x}$.
% , and by $P(\vec{f}, \vec{x})$ a generic atom over function symbols $\vec{f}$
% and variables $\vec{x}$.

We rely on the OBDA framework of \cite{PLCD*08}, which we formalize here
through the notion of \emph{OBDA specification}, which is a triple
$\S=(\T,\M,\Sigma)$ where $\T$ is an \emph{ontology TBox}, $\M$ is a set of
\emph{mappings}, and $\Sigma$ is the schema of a relational database. In the
remainder of this section, we formally introduce such elements of an OBDA
specification.

We assume that ontologies are formulated in \dlliteR~\cite{CDLLR07}, which is
the DL providing the formal foundations for OWL~2~QL, the W3C standard ontology
language for OBDA~\cite{W3Crec-OWL2-Profiles}.  A \dlliteR \emph{TBox} $\T$ is
a finite set of axioms of the form $C \ISA D$ or $P \ISA R$, where $C$, $D$ are
\dlliteR \emph{concepts} and $P$, $R$ are \emph{roles}, following the \dlliteR
grammar.  A \dlliteR \emph{ABox} $\A$ is a finite set of assertions of the form
$A(a)$, $P(a,b)$, where $A$ is a concept name, $P$ a role name, and $a$, $b$
\emph{individuals}.  We call the pair $\O=(\T,\A)$ a \dlliteR \emph{ontology}.

% Ontologies in Semantic Web are serialized as \emph{RDF
%  graphs}~\cite{W3Crec-RDF-Concepts}, and they are queried through the W3C
% standard query language SPARQL~\cite{W3Crec-SPARQL-1.1-query}, with queries
% evaluated under the OWL~2~QL entailment regime~\cite{KRRXZ14}.
We consider here \emph{first-order (FO) queries}~\cite{AbHV95}, and we use
$q^{\D}$ to denote the evaluation of a query $q$ over a database $\D$.  We use
the notation $q^{\A}$ also for the evaluation of $q$ over the ABox $\A$, viewed
as a database.  For an ontology $\O$, we use $\cert(q,\O)$ to denote the
\emph{certain answers} of $q$ over $\O$, which are defined as the set of tuples
$\vec{a}$ of individuals such that $\O\models q(\vec{a})$ (where $\models$
denotes the \dlliteR entailment relation).
We consider also various fragments of FO queries, notably \emph{conjunctive
 queries} (CQs), \emph{unions of CQs} (UCQs), and \emph{joins of UCQs}
(JUCQs)~\cite{AbHV95}.
% There is a strong correspondence between SPARQL and conjunctive
% queries. Intuitively, under the OWL~2~QL semantics each \emph{basic graph
% pattern} (BGP) in a SPARQL query can be seen as a \emph{conjunctive query}
% (CQ) without existentially quantified variables. In this work, as commonly
% done in OBDA literature, we allow the use of existentially quantified
% variables in the CQs, so as to exploit the full expressiveness of CQs and
% \dlliteR.

Mappings specify how to populate the concepts and roles of the ontology from
the data in the underlying relational database instance of $\Sigma$. A \emph{mapping} $m$ is an
expression of the form $L(\vec{f}(\vec{x})) \leftsquigarrow q_m(\vec{x})$: the
\emph{target part} $L(\vec{f}(\vec{x}))$ of $m$ (denoted as $\target(m)$) 
is an atom over function symbols\footnote{For conciseness, we use here abstract function symbols in the
 mapping target.  We remind that in concrete mapping languages, such as
 R2RML~\cite{W3Crec-R2RML}, such function symbols correspond to IRI templates
 used to generate object IRIs from database values.} $\vec{f}$ and variables
$\vec{x}$ whose predicate name $L$ is a concept or role name; the \emph{source
 part} $q_m(\vec{x})$ of $m$ (denoted as $\source(m)$) is a FO query over $\Sigma$ with output variables\footnote{In
 general, the output variables of the source query might be a superset of the
 variables in the target, but for our purposes we can assume that they
 coincide.}  $\vec{x}$. We say that the \emph{signature} $\sign(m)$ of $m$ is
the pair $(L,\vec{f})$, and that $m$ \emph{defines} $L$.  We also define
$\sign(\M)=\setB{\sign(m)}{m \in \M}$.

Following \cite{DLLM*13}, we split each mapping $m = L(\vec{f}(\vec{x}))
 \leftsquigarrow q_m(\vec{x})$ in $\M$ into two parts by introducing an
intermediate view name $V_m$ for the FO query $q_m(\vec{x})$.  We
obtain a \emph{low-level} mapping of the form $V_m(\vec{x})
 \leftsquigarrow q_m(\vec{x})$, and a \emph{high-level} mapping of the form
$L(\vec{f}(\vec{x})) \leftsquigarrow V_m(\vec{x})$.
% We denote by $\M_L$ and $\M_H$ respectively the set of low-level
% and high-level mapping assertions obtained from $\M$.
%
In the following, we abstract away the low-level mapping parts, and we consider
$\M$ as consisting directly of the high-level mappings.  In other words,
we directly consider the intermediate view atoms $V_m$ as the source
part, with the semantics $V_m^\D = q_m^\D$, for each database instance
$\D$.  We denote by $\Sigma_\M$ the \emph{virtual schema}
consisting of the relation schemas whose names are the intermediate view
symbols $V_m$, with attributes given by the answer variables of the
corresponding source queries.

From now on we fix an OBDA specification $\S=(\T,\M,\Sigma)$. Given a database
instance $\D$ for $\Sigma$, we call the pair $(\S,\D)$ an \emph{OBDA instance}.
We call the set of assertions $ \A_{(\M,\D)} =
\setB{L(\vec{f}(\vec{a}))}{ L(\vec{f}(\vec{x})) \leftsquigarrow V(\vec{x}) \in
  \M \text{ and } \vec{a} \in V(\vec{x})^\D} $ the \emph{virtual ABox exposed
  by $\D$ through $\M$}.  Intuitively, such an ABox is obtained by
evaluating, for each (high level) mapping $m$, its source view
$V(\vec{x})$ over the database $\D$, and by using the returned
tuples to instantiate the concept or role $L$ in the target part of
$m$.  The \emph{certain answers} $\cert(q,(\S,\D))$ to a query
$q$ over an OBDA instance $(\S,\D)$ are defined as
$\cert(q,(\T,\A_{(\M,\D)}))$.

In the virtual approach to OBDA, such answers are computed without actually
materializing $\A_{(\M,\D)}$, by transforming the query $q$ into a FO query
$q_{\mathit{fo}}$ formulated over the database schema $\Sigma$ such that
$q_{\mathit{fo}}^{\D'}=\cert(q,(\S,\D'))$, for every OBDA instance $(\S,\D')$.
To define the query $q_{\mathit{fo}}$, we introduce the following notions:
\begin{compactitem}
\item A query $q_r$ is a \emph{perfect rewriting} of a query $q'$ with respect
  to a TBox $\T$, if $\cert(q',(\T,\A))=q_r^{\A}$, for every ABox
  $\A$~\cite{CDLLR07}.
\item A query $q_t$ is an \emph{$\M$-translation} of a query
  $q'$, if $q_t^\D = q'^{\A_{(\M,\D)}}$, for every database
  $\D$ for $\Sigma$~\cite{PLCD*08}.
\end{compactitem}
Notice that, by definition, all perfect rewritings (resp., translations) of
$q'$ with respect to $\T$ (resp., $\M$) are equivalent.
Consider now a perfect rewriting $q_{\T}$ of $q$ with respect to $\T$, and then
a translation $q_{\T,\M}$ of $q_{\T}$ with respect to $\M$.  It is possible to
show that such a $q_{\T,\M}$ satisfies the condition stated above for
$q_{\mathit{fo}}$.

Many different algorithms have been proposed for computing perfect rewritings
of UCQs with respect to \dlliteR TBoxes, see, e.g.,~\cite{CDLLR07,KKPZ12b}.
As for the translation, \cite{PLCD*08} proposes an algorithm that is based on
non-recursive \emph{Datalog}~\cite{AbHV95}, extended with function symbols in
the head of rules, with the additional restriction that such rules never
produce nested terms.
We consider Datalog queries of the form $(G,\Pi)$, where $G$ is the
answer atom, and $\Pi$ is a set of Datalog rules following the
restriction above.  We abbreviate a Datalog query of the form
$(q(\vec{x}),\set{q(\vec{x}) \leftarrow B_1, \ldots, B_n})$,
corresponding to a CQ (possibly with function symbols), as $q(\vec{x})
\leftarrow B_1, \ldots, B_n$, and we also call it $q$.
From now on, given a Datalog rule $r = A \leftarrow B_1, \ldots, B_n$, we
denote by $\head(r)$ its head $A$.
% and by $\body(r)$ its body $B_1, \ldots, B_n$.

\begin{definition}[\textbf{Unfolding of a UCQ}~\cite{PLCD*08}]\label{d:unf-cq}
  Let $q(\vec{x}) \leftarrow L_1(\vec{v}_1), \ldots, L_n(\vec{v}_n)$ be a
  CQ.  Then, the \emph{unfolding $\unf(q,\M)$ of $q$ w.r.t.\
   $\M$} is the Datalog query $(q_{\unf}(\vec{x}), \Pi)$, where
  $\Pi$ is a (up to variable renaming) minimal set of rules having the
  following property:
  \begin{compactitem}
  \item[] If $((m_1, \ldots, m_n),\sigma)$ is a pair such that
    $\{m_1,\ldots,m_n\}\subseteq \M$, and
    \begin{compactitem}
    \item $m_i = L_i(\vec{f}_i(\vec{x}_i)) \leftsquigarrow
       V_i(\vec{z}_i)$, for each $1\le i \le n$, and
    \item $\sigma$ is a most general unifier for the set of pairs
      $\{(L_i(\vec{v}_i),L_i(\vec{f}_i(\vec{x}_i))) \mid 1\le i\le n\}$,
    \end{compactitem}
    then the query $q_{\unf}(\sigma(\vec{x})) \leftarrow
     V_1(\sigma(\vec{z}_1)), \ldots, V_n(\sigma(\vec{z}_n))$ belongs to
    $\Pi$.
  \end{compactitem}
  The \emph{unfolding of a UCQ} $q$ is the union of the unfoldings of
  each CQ in $q$.
\end{definition}

It has been proved in~\cite{PLCD*08} that, for a UCQ $q$,
$\unf(q,\M)$ is an $\M$-translation.

\begin{theorem}[\textbf{An Unfolding is a Translation}~\cite{PLCD*08}]\label{t:sound-compl-unf}
  Let $q$ be a UCQ, and $\M$ a set of mappings. Then, 
  \[\unf(q,\M)^\D = q^{\A_{(\M, \D)}}.\]
\end{theorem}

\section{Background on Cost-Based Optimization of Query Rewriting}
\label{s:background}

Bursztyn
et.~al.~\cite{BursztynGM15b,BursztynGM15,DBLP:journals/pvldb/BursztynGM16}
study a space of alternatives to UCQ perfect rewritings, in a setting without
mappings and in which the ABox of the ontology is directly stored in the
database, under suitable assumptions on the form of the chosen relations and
indexes on them.  Their works provides empirical evidence supporting the
hypothesis that UCQ perfect rewritings could be evaluated less efficiently than
alternatives forms of rewritings.  In particular, they explore a space of
alternatives to the traditional UCQ perfect rewriting in the form of \emph{join
 of unions of conjunctive queries (JUCQs)}, and provide an algorithm that,
given a cost model, selects the best choice.

\begin{example}\label{e:obqa-jucq-1}
  Consider the TBox $\T = \set{C \sqsubseteq D}$, and the query
  $q(x,y) \leftarrow D(x), P(x,y)$. A UCQ perfect rewriting of $q$ with respect
  to $\T$ is the Datalog query $(q_{\rew}(x,y),\Pi)$, where
  \[
    \Pi ~=~ \left\{
      \begin{array}{l@{~~}c@{~~}l}
        q_{\rew}(x,y) &\leftarrow& C(x), P(x,y)\\
        q_{\rew}(x,y) &\leftarrow& D(x), P(x,y)
      \end{array}
    \right\}
  \]
  An alternative perfect rewriting can be obtained by applying well-known
  distributive rules for disjunction and conjunction over $q_{\rew}$.  By doing
  so, we obtain the alternative perfect rewriting $(q'_{\rew}(x,y),\Pi')$,
  where
  \[
    \Pi' ~=~ \left\{
      \begin{array}{l@{~~}c@{~~}l}
        q'_{\rew}(x,y) &\leftarrow& q_{{\rew}_D}(x), P(x,y)\\
        q_{{\rew}_D}(x) &\leftarrow& C(x)\\
        q_{{\rew}_D}(x) &\leftarrow& D(x)
      \end{array}
    \right\}
  \]
  Observe that the above query is a JUCQ. In particular, it corresponds to the
  query
  \[q_{\cover}(x,y) \leftarrow q_{{\rew}_D}(x), q_{{\rew}_P}(x,y)\]
  where
  \begin{itemize}
  \item $q_{{\rew}_D}(x)$ is a UCQ perfect rewriting of
    $q_D(x) \leftarrow D(x)$ with respect to $\T$, and
  \item $q_{{\rew}_P}(x,y)$ is a UCQ perfect rewriting of
    $q_P(x,y) \leftarrow P(x,y)$ with respect to $\T$.
    \qedfull
  \end{itemize}
\end{example}

In this work we study the problem of alternative translations in a full-fledged
OBDA scenario with mapppings. A classical pitfall in such scenario is that it
is very easy to lose the ability of performing joins over database values, so
that one has to perform them over URIs constructed by applying mapping
definitions. The next example shows this issue.

\begin{example}\label{e:obda-jucq-2}
  Consider the ontology and queries from Example~\ref{e:obqa-jucq-1} and the
  following set of mappings:
  \[
    \M ~=~ \left\{
      \begin{array}{l@{~~}c@{~~}l}
        C(f(a)) & \leftsquigarrow& V_1(a)\\
        D(f(b)) & \leftsquigarrow& V_2(b)\\
        D(g(c)) & \leftsquigarrow& V_3(c)\\
        P(f(d),h(e)) & \leftsquigarrow& V_4(d,e)
      \end{array}
    \right\}
  \]

  By applying the unfolding procedure to queries $q_{\rew_D}$ and $q_{\rew_P}$,
  and joining the resulting queries, we obtain the following $\M$-translation
  $q_{\fstyle{trans}}$ of $q'_{\rew}$:
  \[q_{\fstyle{trans}}(x,y) \leftarrow q_{\unf_{D}}(x), q_{\unf_{P}}(x,y)\]
  where
  $q_{\unf_{D}}(x)$ is a Datalog query of program
  \[\Pi_{\unf_D} = \hspace{-1.5ex}
    \begin{mappings}
      \Ldelim{3} & q_{\unf_D}(f(a)) & \leftarrow V_1(a) & \Rdelim{3} \\
  & q_{\unf_D}(f(b)) & \leftarrow V_2(b) \\
  & q_{\unf_D}(g(c)) & \leftarrow V_3(c)
\end{mappings}
\]
and $q_{\unf_{P}}(x,y)$ is a Datalog query of program
\[\Pi_{\unf_P} = \set{q_{\unf_P}(f(d),h(e)) \leftarrow V_4(d,e)}\]
The SQL query corresponding to $q_{\fstyle{trans}}$, in algebra notation, is a
query of the form
\[\pi_{x,y}(\pi_{x/f(a)}(V_1) \cup \pi_{x/f(b)}(V_2) \cup \pi_{x/g(c)}(V_3))
  \Join \pi_{x/f(d),y/h(e)}(V_4)
\]
where each expression of the form $x/f(\vec{y})$ appearing in the projections
denotes the application of the function symbol $f$ over database values
instantiating the attributes in $\vec{y}$, and such applications construct
individuals under the answer variable $x$.  Apart from that $\pi$ behaves as
the standard projection operator.

Observe that $q_{\fstyle{trans}}$ is a JUCQ, and that it contains a join over
the result of the application of function symbols to database values. For
instance, in the Ontop OBDA system such application leads to an inefficient
join over columns whose values are constructed from the concatenation of URI
templates and database values. \qedfull
\end{example}

To formalize the intuitions given in the example above, we now introduce some
terminology from~\cite{BursztynGM15b}, which we use in our technical
development.

\begin{definition}[Cover~\cite{BursztynGM15b}]
Let $q$ be a query consisting of atoms $F=\set{L_1,\ldots,L_n}$. A
\emph{cover for $q$} is a collection $C$ of
non-empty subsets of $F$, called \emph{fragments}, such that
\begin{inparaenum}[\itshape (i)]
\item $\bigcup_{f\in C} f = F$ and
\item no fragment is included into another one.
\end{inparaenum}
\end{definition}

\begin{definition}[Fragment Query~\cite{BursztynGM15b}]
  Let $C$ be a cover for a query $q$. The \emph{fragment
    query} $\frag(\vec{x}_f)$, for $f\in C$, is the query whose body
  consists of the atoms in $f$ and whose answer variables $\vec{x}_f$
  are given by the answer variables $\vec{x}$ of $q$ that appear in
  the atoms of $f$, union the existential variables in $f$ that are
  shared with another fragment $f'\in C$, with $f'\neq f$.
\end{definition}

Given a fragment query $\frag(\vec{x}_f)$, and a TBox $\T$, we use the notation
$\fragu(\vec{x}_f)$ to denote a UCQ perfect rewriting of $\frag(\vec{x}_f)$
with respect to $\T$.

\begin{definition}[Cover-based Perfect Rewriting~\cite{BursztynGM15b}]\label{d:c-based-fol-ref}
  Let $C$ be a cover for a query $q$, and $\T$ a TBox.
Consider the query $q_C(\vec{x}) \leftarrow \bigwedge_{f\in C} \fragu(\vec{x}_f)$. Then
\emph{$q_C$ is a cover-based JUCQ perfect rewriting of $q$ with respect to
 $\T$ and $C$} if it is a perfect rewriting of $q$ with respect to
$\T$.
\end{definition}

In \dlliteR, not every cover leads to a cover-based JUCQ perfect
rewriting. Authors in~\cite{BursztynGM15b} gave a sufficient condition, called
``safety'', to guarantee that every cover leads to a cover-based JUCQ perfect
rewriting.

\begin{theorem}[Cover-based query answering~\cite{BursztynGM15b}]\label{t:c-based-qa} Applying Definition~\ref{d:c-based-fol-ref} on a safe cover $C$ for $q$ with respect to a TBox $\T$, and any rewriting technique producing a UCQ, yields a cover-based JUCQ perfect rewriting $q_{r}$ of $q$ with respect to $\T$.
\end{theorem}

\section{Cover-based Translation in OBDA}\label{c:planning:s:cover-based-trans}

In this section, we propose a cover-based translation of UCQ perfect
rewritings. To do so, we rely on our fixed OBDA specification
$\S=(\T,\M,\Sigma)$.  In addition, we also fix a query $q(\vec{x})$ and a
(safe) cover $C$ for it, as well as its cover-based JUCQ perfect rewriting
$q_C(\vec{x}) \leftarrow \bigwedge_{f\in C} \fragu$ with respect to $\T$ and
$C$. We introduce two different characterizations of unfoldings of $q_C$, which
produce $\M$-translations of $q$. The first characterization relies on the
intuition of joining the unfoldings of each fragment query in $q_C$.

\begin{definition}[Unfolding of a JUCQ -- Type~1]
  \label{d:unf-jucq}
  Let $C$ be a cover for a query $q$.
  For each $f\in C$, let $Aux_f$ be an auxiliary predicate for
  $\fragu(\vec{x}_f)$, and $U_f$ a view symbol for the
  unfolding $\unf(\fragu(\vec{x}_f), \M)$.
  Consider the set $\MAux{\M} = \setB{\Aux_f(\vec{x}_f)
    \leftsquigarrow U_f(\vec{x}_f)}{f\in C}$ of mappings
  associating the auxiliary predicates to the auxiliary view names. Then, we
  define the \emph{unfolding $\unf(q_C,\M)$ of $q_C$ with respect
   to\ $\M$} as
  $\unf(\qAux{q}_C(\vec{x}) \leftarrow \bigwedge_{f\in C}
   \Aux_f(\vec{x}_f), \MAux{\M})$.
\end{definition}

\begin{theorem}[Translation~1]\label{t:sound-compl-unf-jucq}
  Let $q_C(\vec{x}) \leftarrow \bigwedge_{i=1}^n q^{\ucq}_{|f_i}$ be a JUCQ
  cover-based perfect rewriting. Then, the query $\unf(q_C, \M)$ is an
  $\M$-translation of $q_C$.
  \begin{proof}
    By definition of $\M$-translation, we need to prove that, for each OBDA instance $\D$ of $\S$, the following equality holds:

    \[\unf(q_C, \M)^\D = q_C^{\can{{\A_{(\M, \D)}}}}.\]

    By Definition~\ref{d:unf-jucq}, we know that:

    \begin{equation}\label{eq:planning:unfC-def}
      \unf(q_C, \M) = \unf(\qAux{q}_C(\vec{x}) \leftarrow \bigwedge_{f \in C} \Aux_f(\vec{x}_f), \MAux{M}),
    \end{equation}
    with $\MAux{M}$ and $\Aux_f$ as in Definition~\ref{d:unf-jucq}. Let $\D$ be an arbitrary database instance for $\Sigma$. Then, it is not hard to verify that $q_C(\vec{x})^{\can{\A_{(\M,\D)}}} = \qAux{q}_C(\vec{x})^{\can{\A_{(\MAux{M},\D)}}}$, by applying Definition~\ref{d:unf-cq} and using the distributive property of the join operation over the union operation. By using Theorem~\ref{t:sound-compl-unf}, it holds that $\qAux{q}_C(\vec{x})^{\can{\A_{(\MAux{M},\D)}}} = \unf(\qAux{q}_C, \MAux{M})^\D$. By applying the transitive property of $=$, we obtain that $q_C(\vec{x})^{\can{\A_{(\M,\D)}}} = \unf(\qAux{q}_C(\vec{x}), \MAux{M})^\D$. The thesis follows by applying Equation~\eqref{eq:planning:unfC-def}, and from the fact that the choice of $\D$ was arbitrary.
\end{proof}
\end{theorem}

The above unfolding characterization for JUCQs corresponds to a translation containing SQL joins over URIs resulting from the application of function symbols to database values, in a similar fashion as in Example~\ref{e:obda-jucq-2}, rather than over (indexed) database values directly. In general, such joins cannot be evaluated efficiently by RDBMSs~\cite{RoRe15}.

In the remainder of this section we propose a solution to this issue by
defining an alternative characterization for the unfolding of a JUCQ, that
ensures that joins are always performed on database values. To do so, we first
need to introduce a number of auxiliary notions and lemmas.

\begin{definition}[Equivalent Sets of Mappings]\label{c:prel:d:equiv-mappings}
  Two sets of mapping assertions $\M$ and $\M'$ are said to be
  \emph{equivalent}, denoted as $\M \equiv \M'$, if
  \[
    \text{For each database instance $\D$ of $\Sigma$, it holds that } \A_{(\M,
     \D)} = \A_{(\M', \D)}
  \]
\end{definition}

A direct consequence of
Theorem~\ref{t:sound-compl-unf} is that unfoldings with respect to equivalent
mappings must be equivalent.
\begin{lemma}{\label{t:rev-sound-compl-unf}}
  Let $\M'$ be a set of mappings such that $\M' \equiv \M$. Then, for every UCQ $q$, it holds
  \[\unf(q, \M) \equiv \unf(q, \M')\]
\begin{proof}
    \noindent We prove the contrapositive of the statement.

    \noindent Without loss of generality, let $a$ be such that $a \in \unf(q, \M)^\D$, and
    $a \notin \unf(q, \M')^\D$, for some data instance $\D$ of $\Sigma$.
    \noindent For Theorem~\ref{t:sound-compl-unf}, the above implies that
    \[a \in q^{\A_{(\M, \D)}}, a \notin q^{\A_{(\M', \D)}}.\]
    \noindent Hence, $\A_{(\M, \D)} \neq A_{(\M', \D)}$, and therefore $\M \not \equiv \M'$.
  \end{proof}
\end{lemma}

\begin{definition}[Restriction of a Mapping to a Signature]\label{d:restriction-mapping-to-template}

  Let $(L,\vec{f}) \in \sign(\M)$ be a signature in $\M$. Then, the
  \emph{restriction $\M|_{(L,\vec{f})}$ of $\M$ with respect to the
   signature $(L,\vec{f})$} is the set $\setB{m \in \M}{\sign(m) = (L,\vec{f})}$ of mappings.
\end{definition}

  \begin{definition}[Wrap of a Mapping]\label{d:wrap-mapping}
  Let the set $\setB{L(\vec{f}(\vec{v}_i)) \leftsquigarrow V_i(\vec{v}_i)}{1\le i\le n}$ be the restriction $\M|_{(L, \vec{f})}$ of $\M$ with respect to
  the signature $(L, \vec{f})$, and $\vec{f}(\vec{v})$ a tuple of terms over fresh variables $\vec{v}$. Then, the
  \emph{wrap of $\M|_{(L, \vec{f})}$} is the (singleton) set
  $\wrap(\M|_{(L, \vec{f})}) = \set{L(\vec{f}(\vec{v})) \leftsquigarrow
    W(\vec{v})}$ of mappings, where $W$ is a fresh view name for the Datalog query $(W(\vec{v}), \set{W(\vec{v}_i) \leftarrow V_i(\vec{v}_i) \mid 1\le i\le n})$
  with answer variables $\vec{v}$.

  The \emph{wrap of $\M$} is the set $\wrap(\M) = \bigcup_{(L,
    \vec{f}) \in \sign(\M)} \wrap(\M|_{(L, \vec{f})})$ of mappings.
\end{definition}

\begin{example}
  Consider the following set of mapping assertions
  \[\M|_{(A, f)} = \hspace{-1.5ex}
  \begin{mappings}
    \Ldelim{2} & A(f(x,y)) & \leftsquigarrow V_1(x,y) & \Rdelim{2} \\
    & A(f(a,b)) & \leftsquigarrow V_2(a,b)
  \end{mappings}
  \]
Let $v,v'$ be fresh variables. Then, $\wrap(\M|_{(A, f)}) = \set{A(f(v,v'))
 \leftsquigarrow W(v,v')}$, where $W(v,v')$ is the Datalog query
\[(W(v,v'), \set{W(x,y) \leftarrow V_1(x, y), W(a,b) \leftarrow V_2(a, b)})\]
\end{example}

\begin{lemma}\label{l:wrap-template}
Given a signature $(L, \vec{f}) \in \sign(\M)$, the following equivalence holds:

  \[\M|_{(L, \vec{f})} \equiv \wrap(\M|_{(L, \vec{f})})\]
  \begin{proof}
    Let $\M|_{(L, \vec{f})}$ be the set $\setB{L(\vec{f}(\vec{v}_i)) \leftsquigarrow V_i(\vec{v}_i)}{i = 1, \ldots, n}$. Then, for every data instance $\D$,
    the virtual ABox $\A_{(\M|_{(L, \vec{f})}, \D)}$ can be written as:
    \begin{equation}\label{eq:wrap-template-1:plan}
    \setB{L(\vec{f}(\vec{o}))}{\exists i\in\{1,\ldots,n\}: \vec{o} \in V_i(\vec{v}_i)^\D}.
    \end{equation}

    The Set~\eqref{eq:wrap-template-1:plan} can equivalently be rewritten as:

    \begin{equation}\label{eq:wrap-template-2:plan}
      \setB{L(\vec{f}(\vec{o}))}{\vec{o} \in (W(\vec{v}), \setB{W(\vec{v}_i) \leftarrow V_i(\vec{v}_i)}{i = 1, \ldots, n})^\D},
    \end{equation}
    for some fresh view symbol $W$. The thesis follows by observing that the Set~\eqref{eq:wrap-template-2:plan} corresponds to the definition of the virtual ABox $\A_{(\wrap(\M|_{(L, \vec{f})}), \D)}$.
\end{proof}
\end{lemma}

\begin{lemma}\label{l:wrap-mapping}
  Let $\M$ be a set of mappings. Then $\M \equiv \wrap(\M)$.
\begin{proof}
  It follows directly from Lemma~\ref{l:wrap-template}, and from the fact that fresh view symbols are introduced for each signature.
\end{proof}
\end{lemma}

The $\operatorname{wrap}$ operation produces a set of mappings in which there do not exist two mappings sharing the same signature. Hence, it groups the mappings for a signature into a single mapping.The next lemma formalizes this property.

\begin{lemma}\label{l:wrap-template-sharing}
  Consider the wrap $\wrap(\M)$ of the set $\M$ of mappings. Then, for every mapping $m, m' \in \wrap(\M)$, it holds:
  \[\text{If }\sign(m) = \sign(m')\text{, then } m = m'\]
\end{lemma}

\begin{proof}
It follows directly from the definition of the $\wrap$ operation.
\end{proof}

 We now introduce an operation that \emph{splits} a mapping according to the function symbols adopted on its source part.

\begin{definition}[Split]\label{d:splitted-mapping}
  Let $m =L(\vec{x}) \leftsquigarrow U(\vec{x})$ be a mapping where $U$ is a
  view name for a NRLP query
  $(U(\vec{x}), \set{U(\vec{f}_i(\vec{x}_i)) \leftarrow V_i(\vec{x}_i) \mid
   1\le i\le n})$. Then, the \emph{split of $m$} is the set
  $\split(m) = \setB{L(\vec{f}_i(\vec{x}_i)) \leftsquigarrow
   V_i(\vec{x}_i)}{1\le i\le n}$ of mappings.  For a mapping $m'$ not of the
  form above, the split of $m'$ is defined as the set $\set{m'}$, that is,
  $\split(m') = \set{m'}$.  We denote by $\split(\M)$ the \emph{split} of the
  set $\M$ of mappings, which is defined as the set
  $\bigcup_{m \in \M} \split(m)$ of mappings.
\end{definition}

\begin{example}{\label{e:splitted-mapping}}
  Consider a mapping $m = A(x) \leftsquigarrow U(x)$, where $U(x) = (U(x), \Pi)$  and
  \[\Pi =
  \begin{mappings}
    \Ldelim{3} & U(h(a,b)) & \leftarrow V_1(a,b) & \Rdelim{3}\\
    & U(h(c,d)) & \leftarrow V_2(c,d)\\
    & U(i(e,f,g)) & \leftarrow V_3(e,f,g).
  \end{mappings}
  \]
  \noindent Then,
  \[\split(\M) =
  \begin{array}{l l l l}
    \ldelim\{{3}{3mm} & A(h(a,b)) & \leftsquigarrow V_1(a,b) & \rdelim\}{3}{3mm} \\
    & A(h(c,d)) & \leftsquigarrow V_2(c,d)\\
    & A(i(e,f,g)) & \leftsquigarrow V_3(e,f,g).
  \end{array}
  \]
\end{example}

\begin{lemma}\label{l:split-mapping}
  Let $m$ be a mapping. Then $\set{m} \equiv \split(\set{m})$.
\end{lemma}

\begin{proof}
    The thesis trivially holds if $m$ is of a form such that $\set{m} = \split(m)$.
    If $\set{m} \neq \split(m) $, then $m$ must be of the form $L(\vec{x}) \leftsquigarrow U(\vec{x})$, in which $U$ is a view name for a NRLP query $(U(\vec{x}), \Pi)$, where $\Pi = \set{U(\vec{f}_i(\vec{x}_i)) \leftarrow V_i(\vec{x}_i) \mid 1\le i\le n}$. By Definition~\ref{c:prel:d:equiv-mappings}, we need to prove that, for each instance $\D$ of $\S$, $\A_{(\set{m}, \D)} = \A_{(split(\set{m}), \D)}$. Let $\D$ be an arbitrary database instance. By definition of virtual ABox:

    \[A_{(\set{m}, \D)} = \setB{L(\vec{f}_i(\vec{a}_i))}{\vec{f}_i(\vec{a}_i) \in U(\vec{x})^\D}\]

    The set above is equal to

    \begin{align*}
      & \setB{L(\vec{f}_i(\vec{a}_i))}{\exists (q(\vec{f}_i(\vec{x}_i)) \leftarrow V_i(\vec{x}_i)) \in \Pi \text{ s.t.\ } \vec{a}_i \in V_i(\vec{x}_i)^\D} \\
      = & \setB{L(\vec{f}_i(\vec{a}_i))}{\exists (L(\vec{f}_i(\vec{x}_i)) \leftsquigarrow V_i(\vec{x}_i)) \in \operatorname{split}(\set{m}) \text{ s.t.\ } \vec{a}_i \in V_i(\vec{x}_i)^\D}
      = \A_{(split(\set{m}), \D)}.
    \end{align*}
  \end{proof}

\begin{corollary}
  \label{cor:split-mapping}
  Let $\M$ be a set of mappings.  Then $\M \equiv \split(\M)$.
\end{corollary}

\begin{proof}
By Definition~\ref{c:prel:d:equiv-mappings}, we need to prove that, for each
instance $\D$ of $\S$, $\A_{(\M, \D)} = \A_{(split(\M), \D)}$. Let $\D$ be a
data instance for $\S$. The following trivially holds:
\[
  \A_{(\M, \D)} = \bigcup_{m \in \M} \A_{(\set{m}, \D)}.
\]

By Lemma~\ref{cor:split-mapping},
\[
  \bigcup_{m \in \M} \A_{(\set{m}, \D)} = \bigcup_{m \in \M}
  \A_{(\split(\set{m}), \D)}.
\]

By Definition~\ref{d:splitted-mapping}, we obtain
\[
  \bigcup_{m \in \M} \A_{(\split(\set{m}), \D)} = \A_{(\split(\M), \D)}.
\]
\end{proof}

\begin{definition}[Unfolding of a JUCQ -- Type~2]\label{d:unf-opt-jucq}
  Let $\qAux{q}_C$ be a query and $\MAux{\M}$ a set of mappings as in
  Definition~\ref{d:unf-jucq}. Then, the \emph{optimized unfolding}
  $\unfopt(q_C(\vec{x}), \M)$ of $q_C$ with respect to $\M$ is defined as
  $\unf(\qAux{q}_C(\vec{x}), \wrap(\split(\MAux{\M})))$.
\end{definition}

Observe that the optimized unfolding of a JUCQ is a \emph{union of JUCQs} (UJUCQ). Moreover, where each JUCQ produces answers built from a \emph{single} tuple of function symbols, if all the attributes are kept in the answer. The next example, aimed at clarifying the notions introduced so far, illustrates this.

\begin{example}
  Let $q(x,y,z) \leftarrow P_1(x,y), C(x), P_2(x,z)$, and consider a
  cover $\set{f_1, f_2}$ generating fragment queries $q|_{f_1} =
   q(x,y) \leftarrow P_1(x,y), C(x)$ and $q|_{f_2} = q(x,z) \leftarrow
   P_2(x,z)$. Consider the set of mappings
  \[
    \M=
  \left\{\begin{array}{l@{~}l@{~}l@{\qquad\qquad}llll}
      P_1(f(a),g(b)) &\leftsquigarrow& V_1(a,b) &  P_1(f(a), g(b)) & \leftsquigarrow V_2(a,b)\\
     P_1(h(a), i(b)) & \leftsquigarrow& V_3(a,b) & C(f(a)) & \leftsquigarrow V_4(a)\\
     P_2(f(a), k(b)) & \leftsquigarrow& V_5(a,b) & P_2(f(a), h(b)) & \leftsquigarrow V_6(a,b)\\
  \end{array}\right\}
  \]
\emph{Translation I.}  According to Definition~\ref{d:unf-jucq}, the JUCQ $q(x,y,z) \leftarrow q|_{f_1}(x,y), q|_{f_2}(x,z)$ can be rewritten as the auxiliary query $q^{\aux}(x,y,z) = \Aux_1(x,y), \Aux_2(x,z)$ over mappings
  \[
    \MAux{\M} =
    \begin{tabular}{l l \tableColSpace l l }
      \ldelim\{{1}{2mm} & $\Aux_1(x,y) \leftsquigarrow U_1(x,y)$ & $\Aux_2(x,z) \leftsquigarrow U_2(x,z)$ & \rdelim\}{1}{2mm}
    \end{tabular}
  \]
where $U_1$ is a view name for $\unf(q|_{f_1}(x,y), \M) = (U_1(x,y), \Pi_1)$, and $U_2$ is a view name for $\unf(q|_{f_2}(x,z), \M) = (U_2(x,z), \Pi_2)$, such that

\noindent\begin{tabular}{l l}
  \begin{minipage}{.55\linewidth}
    \cmath{
      \Pi_1 =
  \begin{tabular}{l l  l l  }
      \ldelim\{{2}{2mm} &  $U_1(f(a),g(b))$ &  $\leftarrow V_1(a,b), V_4(a)$ & \rdelim\}{2}{2mm}\\
&  $U_1(f(a),g(b))$ &  $\leftarrow V_2(a,b), V_4(a)$
  \end{tabular}
}
\end{minipage}
  \begin{minipage}{.45\linewidth}
    \cmath{\Pi_2 =
  \begin{tabular}{l l  l l }
\ldelim\{{2}{2mm} &  $U_2(f(a),k(b))$ & $\leftarrow V_5(a,b)$ & \rdelim\}{2}{2mm}\\
 & $U_2(f(a),h(b))$ &  $\leftarrow V_6(a,b)$\\
  \end{tabular}
}
  \end{minipage}
\end{tabular}

\noindent \emph{Translation II.} By Definition~\ref{d:splitted-mapping}, we compute the split of $\MAux{\M}$:
\cmath{
\split(\MAux{\M}) =\!\!\!
\begin{tabular}{l l \tableColSpace l l}
  \ldelim\{{2}{2mm} & $\Aux_1(f(a),g(b)) \leftsquigarrow V_1(a,b), V_4(a)$ & $\Aux_2(f(a),k(b)) \leftsquigarrow V_5(a,b)$ &   \rdelim\}{2}{2mm}\\
  & $\Aux_1(f(a),g(b)) \leftsquigarrow V_2(a,b), V_4(a)$ & $\Aux_2(f(a),h(b)) \leftsquigarrow V_6(a,b)$\\
\end{tabular}
}%
By Definition~\ref{d:wrap-mapping}, we compute the wrap of
$\split(\MAux{\M})$:
\cmath{\wrap(\split(\MAux{\M})) =
\begin{tabular}{l l \tableColSpace ll}
  \ldelim\{{2}{2mm} & $\Aux_1(f(a),g(b)) \leftsquigarrow W_{3}(a,b)$ & $\Aux_2(f(a),k(b)) \leftsquigarrow W_{4}(a,b)$ & \rdelim\}{2}{2mm} \\
 &  & $\Aux_2(f(a),h(b)) \leftsquigarrow W_{5}(a,b)$ &
\end{tabular}
}%
where $W_{3}(a,b)$, $W_{4}(a,b)$, $W_{5}(a,b)$
are Datalog queries whose programs are respectively
\\[3mm]
\noindent\begin{tabular}{l l}
  \begin{minipage}{.53\linewidth}
    \cmath{\Pi_3 =
  \begin{tabular}{l l l l}
   \ldelim\{{2}{2mm} &      $W_{3}(a,b)$ & $\leftarrow V_1(a,b), V_4(a)$ & \rdelim\}{2}{2mm}\\
 &     $W_{3}(a,b)$ & $ \leftarrow V_2(a,b), V_4(a) $
  \end{tabular}
}
  \end{minipage}
  \begin{minipage}{.45\linewidth}
\vspace{-3mm}
  \begin{tabular}{l l l l}
    $\Pi_4$ =  \ldelim\{{1}{2mm}   & $W_{4}(a,b)$ & $\leftarrow V_5(a,b)$ & \rdelim\}{1}{2mm}\\[1mm]
    $\Pi_5$ =  \ldelim\{{1}{2mm}   & $W_{5}(a,b)$ & $\leftarrow V_6(a,b)$ & \rdelim\}{1}{2mm}\\
  \end{tabular}
  \end{minipage}
\end{tabular}

\noindent Finally, by Definition~\ref{d:unf-opt-jucq} we compute the optimized unfolding of $q_C$ with respect to $\M$:
\cmath{\unfopt(q_C(x,y,z), \M) =\unf(\qAux{q}(x,y,z), \wrap(\split(\MAux{\M}))) = (\qAux{q}_{\unf}(x,y,z), \Pi_{\unf}) }%
where
\cmath{
  \Pi_{\unf} =
  \begin{tabular}{l l  l}
%    \ldelim\{{1}{2mm} & $q^{aux}_{unf}(f(a),g(b)) \leftarrow W_3(a,b)$ & $q^{aux}_{unf}(f(a), k(b)) \leftarrow W_4(a,b)$ &     \rdelim\}{1}{2mm}
     \ldelim\{{2}{2mm} & $\qAux{q}_{\unf}(f(a),g(b),k(b')) \leftarrow
      W_3(a,b), W_4(a,b')$ &     \rdelim\}{2}{2mm} \\
 & $\qAux{q}_{\unf}(f(a),g(b),h(b')) \leftarrow W_3(a,b),
  W_5(a,b')$ &
 \end{tabular}
}%
Observe that $\unfopt(q_C(x,y,z), \M)$ is a UJUCQ. Moreover, each of the
two JUCQs in $\qAux{q}_{\unf}$ contributes with answers built out of a
specific tuple of function symbols.
%  Hence, for each JUCQ $q$, $q'$ in the
% unfolding, and for every database instance $\D$, it holds that
% $q^\D \cap q'^\D = \emptyset$.
\qedfull
\end{example}

\begin{theorem}[Translation~2]
  The query $\unfopt(q_C, \M)$ is an $\M$-translation of $q_C$.
\begin{proof}

  By Definition~\ref{d:unf-opt-jucq},
  \[\unf_{opt}(q_C(\vec{x}), \M) = \unf(\qAux{q}_C(\vec{x}), \wrap(\split(\MAux{M}))).\]
\noindent By Lemma~\ref{l:wrap-mapping} and Corollary~\ref{cor:split-mapping},
  $\wrap(\split(\MAux{M})) \equiv \MAux{M}$. Hence, by
  Lemma~\ref{t:rev-sound-compl-unf},

  \[\unf(\qAux{q}_C(\vec{x}), \wrap(\split(\MAux{M}))) \equiv \unf(\qAux{q}_C(\vec{x}), \MAux{M}).\]

  \noindent The thesis is proved by observing that, according to
  Theorem~\ref{t:sound-compl-unf-jucq}, $\unf(\qAux{q}_C(\vec{x}), \MAux{M})$ is an $\M$-translation of $q_C$.
\end{proof}
\end{theorem}

\section{Unfolding Cardinality Estimation}\label{c:planning:s:unf-est}

For convenience, in this section, we use relational algebra notation~\cite{AbHV95} for CQs.
To deal with multiple occurrences of the same predicate in a CQ, the
corresponding algebra expression would contain renaming operators.  However, in
our cardinality estimations we need to understand when two attributes actually
refer to the same relation, and this information is lost in the presence of
renaming.  Instead of introducing renaming, we first explicitly replace
multiple occurrences of the same predicate name in the CQ by aliases (under the
assumption that aliases for the same predicate name are interpreted as the same
relation). Specifically, we use alias $\occ[i]{V}$ to represent the $i$-th
occurrence of predicate name $V$ in the CQ.  Then, when translating the aliased
CQ to algebra, we use \emph{fully qualified attribute names} (i.e., each
attribute name is prefixed with the (aliased) predicate name).  So, to
reconstruct the relation name $V$ to which an attribute $\occ[i]{V}.x$ refers,
it suffices to remove the occurrence information $\occ[i]{}$ from the prefix
$\occ[i]{V}$.  When the actual occurrence of $V$ is not relevant, we use
$\occ{V}$ to denote the alias.

We first consider the restricted case of CQs, which we call \emph{basic CQs},
whose algebra expression is of the form
\[E = \occ{V}^0 \Join_{\theta_1} \occ{V}^1
  \Join_{\theta_2}\cdots\Join_{\theta_n} \occ{V}^n,\]%
where, the $V^i$s denote predicate names, and for each $i\in\{1,\ldots,n\}$,
the join condition $\theta_i$ is a single equality, i.e., of the form
$\occ{V}^j.\vec{x}=\occ{V}^i.\vec{y}$, for some $j<i$.  We then extend our
cardinality estimation to arbitrary CQs, allowing for projections and arbitrary
joins.

Given a basic CQ $E$ as above, we denote by $E^{(m)}$, for $1\le m\le n$, the
sub-expression of $E$ up to the $m$-th join operator, namely
$E^{(m)}= \occ{V}^0 \Join_{\theta_1} \occ{V}^1
\Join_{\theta_2}\cdots\Join_{\theta_m} \occ{V}^m$.

\smallskip

In the following, in addition to an OBDA specification, we also fix a database
instance $\D$ for $\Sigma$.
% , and for a relation symbol or an algebra expression $E$, we denote its
% extension $E^{\D}$ under $\D$ simply as $E$, when the
% meaning is clear from the context.
We use $V$ and $W$ to denote relation names (with an
associated relation schema) in the virtual schema $\M_\Sigma$, whose
associated relations consist of (multi)sets of labeled tuples (see the
\emph{named perspective} in \cite{AbHV95}). Given a relation $S$, we
denote by $\distSet{S}$ the number of (distinct) tuples in $S$, by
$\restSet{S}{L}$ the \emph{projection} of $S$ over attributes
$L$ (under \emph{set-semantics}), by
$\restBag{S}{L}$ the \emph{restriction} of $S$ over attributes
$L$ (under \emph{bag-semantics}), and by
$\restSet{S_1}{L_1}\capren\restSet{S_2}{L_2}$ intersection of relations
disregarding attribute names, i.e.,
$\restSet{S_1}{L_1}\cap \rho_{L_2\mapsto L_1}(\restSet{S_2}{L_2})$.
We also use the classical notation $P(\alpha)$ to denote the probability
that an event $\alpha$ happens.

\subsection{Background on Cardinality Estimation}

In this subsection we list a number of assumptions that are commonly made by models of cardinality estimations proposed in the database literature (e.g., see~\cite{SiKS05}). Some of these assumptions will be maintained also in our cardinality estimator, while others will be relaxed or dropped due to the additional information given by the structure of the mappings and the ontology, which is not available in a traditional database setting.

\myPar{Uniform distribution in the interval.} This assumption holds when values are uniformly distributed across one interval. Formally:
\begin{equation}\label{eq:min-max} \tag{a1}
  P(C < v) = (v - min(C)) / (max(C) - min(C))
\end{equation}
  for each value $v$ in a column $C$.

  \myPar{Uniform distribution across distinct values.} This assumption holds when values in a column are uniformly distributed across the range of values. Formally:
  \begin{equation}\label{eq:uniformity} \tag{a2}
    \forall v_1, v_2 \in C.\ P(C = v_1) = P(C = v_2)
  \end{equation}
  for all values $v_1$ and $v_2$ in a column $C$.

  A consequence of this assumption is that, for each relation name $T$ and set of attributes $\vecSet{x}$,
  \[P(\restBag{T}{\vecSet{x}} = \vec{v}) = P(\restSet{T}{\vecSet{x}} = \vec{v}) = 1/\distSet{\restSet{T}{\vecSet{x}}}\text{, }\forall \vec{v} \in \pi_\vecSet{x}(\restSet{T}{\vecSet{x}}).\]
  Hence, \emph{number of repetitions per tuple of elements} can be calculated as
  \begin{equation}{\label{eq:num-rep-per-tuple}}
    \distSet{T} \cdot P(\restBag{T}{\vecSet{x}} = \vec{v}) = \frac{\distSet{T}}{\distSet{\restSet{T}{\vecSet{x}}}}.
  \end{equation}
  From Equation~\eqref{eq:num-rep-per-tuple}, we directly compute the \emph{ratio of distinct tuples over $\vec{x}$ in $T$} as:
  \[\frac{\distSet{\restSet{T}{\vecSet{x}}}}{\distSet{T}}\]
\myPar{Independent Distributions.} This assumption holds when the distributions between different attributes are independent. Formally:
\begin{equation}\label{eq:independence} \tag{a3}
  P(C_1=v_1 \mid C_2 = v2) = P(C_1 = v_1)
\end{equation}
\myPar{Facing Values.} This assumption holds when attributes join ``as much as possible''. That is, given a join $T \Join_{\vec{x} = \vec{y}} S$, it is assumed that
\begin{equation}\label{eq:facing-assumption} \tag{a4}
  \distSet{\restSet{T}{\vecSet{x}} \cap \restSet{S}{\vecSet{y}}} = min(\distSet{\restSet{T}{\vecSet{x}}},\distSet{\restSet{S}{\vecSet{y}}}).
\end{equation}

Under these assumptions, the cardinality of a join
$V \Join_{\vec{x}=\vec{y}} W$ can be estimated~\cite{SwSc94} to be:
\begin{equation}
  \label{eq:std-estimator}
  \small
  \fv(V \Join_{\vec{x}=\vec{y}} W) \cdot
  \frac{\distSet{V^{\D}}}{\distEst{V}{\vec{x}}} \cdot
  \frac{\distSet{W^{\D}}}{\distEst{W}{\vec{y}}}
\end{equation}
% \distSet{\restSet{V_1}{L}} \cdot \distSet{\restSet{V_2}{Y}}
where $\fv(V \Join_{\vec{x}=\vec{y}} W)$ is an estimation of the number of distinct values satisfying the
join condition $V \Join_{\vec{x}=\vec{y}} W$ (i.e., $\fv$ estimates
$\distSet{\restSet{V^{\D}}{\vecSet{x}} \capren
  \restSet{W^{\D}}{\vecSet{y}}}$,
% corresponding to the cardinality of the query
% $\distSet{\pi_{V.\vecSet{x}} (V \Join_{V.\vec{x} =
% W.\vec{y}} W)}$)
and $\distEst{V}{\vec{x}}$ (resp.,
$\distEst{W}{\vec{y}}$) corresponds to the estimation of
$\distSet{\restSet{V^{\D}}{\vec{x}}}$ (resp.,
$\distSet{\restSet{W^{\D}}{\vec{y}}}$), both calculated according to the
aforementioned assumptions.  Note that the fractions such as
$\frac{\distSet{V^{\D}}}{\distEst{V}{\vec{x}}}$ estimate the number of
tuples associated to each value that satisfies the join condition, and derive directly from Assumption~\eqref{eq:uniformity}.

\subsection{Cardinality Estimation of CQs}

\paragraph{Cardinality Estimator.}

Given a basic CQ $E'$, $\cardEst(E')$ estimates the number $\distSet{{E'}^\D}$
of distinct results in the evaluation of $E'$ over $\D$. We define it as
\begin{equation}\label{eq:f}
  \small
  \cardEst(E \Join_{\occ[p]{V}.\vec{x}=\occ[q]{W}.\vec{y}} \occ[q]{W}) =
  \begin{cases}
    \displaystyle
    \left\lceil\frac{\fv(\occ[p]{V} \Join_{\occ[p]{V}.\vec{x}=\occ[q]{W}.\vec{y}} \occ[q]{W}) \cdot \distSet{V^\D} \cdot \distSet{W^\D}}{\distEst{V}{\occ[p]{V}.\vec{x}} \cdot \distEst{W}{\occ[q]{W}.\vec{y}}}\right\rceil, & \text{if $E = V$}\\[1em]
    \displaystyle \left\lceil\frac{\fv(E \Join_{\occ[p]{V}.\vec{x}=\occ[q]{W}.\vec{y}} \occ[q]{W})
        \cdot \cardEst(E) \cdot \distSet{W^\D}}{\distEst{E}{\occ[p]{V}.\vec{x}} \cdot
        \distEst{W}{\occ[p]{V}.\vec{y}}}\right\rceil, & \text{otherwise.}
  \end{cases}
\end{equation}
Our cardinality estimator exploits assumptions~\eqref{eq:uniformity} and~\eqref{eq:independence}
above, and relies on our definitions of the \emph{facing values estimator}
$\fv$ and of the \emph{distinct values estimator} $\operatorname{dist}_\D$,
which are based on additional statistics collected with the help of the
mappings, instead of being based on assumptions~\eqref{eq:min-max} and~\eqref{eq:facing-assumption},
as in Formula~\eqref{eq:std-estimator}.

\paragraph{Facing Values Estimator.}
Given a basic CQ $E' = E \Join_{\occ[p]{V}.\vec{x}=\occ[q]{W}.\vec{y}} \occ[q]{W}$, the
estimation $\fv(E')$ of the cardinality
$\distSet{\restSet{E^\D}{V.\vec{x}} \capren \restSet{W^\D}{W.\vec{y}}}$ is
defined as
\begin{equation}\label{eq:k}
  \small
  \fv(E \Join_{\occ[p]{V}.\vec{x}=\occ[q]{W}.\vec{y}} \occ[q]{W}) =
  \begin{cases}
    \distSet{\restSet{V^\D}{\vecSet{x}} \capren \restSet{W^\D}{\vecSet{y}}}, & \text{if } E = V \\[1em]
    \displaystyle
    \left\lceil \distSet{\restSet{V^\D}{\vecSet{x}} \capren \restSet{W^\D}{\vecSet{y}}} \cdot
      \frac{\distEst{E}{\occ[p]{V}.\vec{x}}}{\distEst{V}{\occ[p]{V}.\vec{x}}}\right\rceil, & \text{otherwise,}
  \end{cases}
\end{equation}
where $\distSet{\restSet{V^\D}{\vecSet{x}} \capren \restSet{W^\D}{\vecSet{y}}}$
is assumed to be a statistic available after having analyzed the mappings
together with the data instance. The fraction
$\frac{\distEst{E}{\occ[p]{V}.\vec{x}}}{\distEst{V}{\occ[p]{V}.\vec{x}}}$ is a scaling factor
relying on assumption~\eqref{eq:uniformity}.

\paragraph{Distinct Values Estimator.}

\begin{definition}[Equivalent Attributes]
  Let $Q$ be a set of qualified attributes, and $E$ be basic CQ. We define the set $\jc(E, Q)$ of equivalent attributes of $Q$ in $E$ as $\bigcup_{i>0} C_i$, where
  \begin{itemize}
  \item $C_1 := \set{Q}$
  \item {\small $C_{n+1} := C_{n} \cup \setB{Q'}{\exists Q'' \in C_{n} \text{ s.t.\ } Q'=Q'' \text{ or } Q''=Q' \text{ is a join condition in }E}$}, $n \ge 1$
  \end{itemize}
\end{definition}

\begin{definition}[Join Sub-Expression]
  Given a basic CQ $E$ and a set $\occ[p]{V}.\vec{x}$ of qualified attributes, the
  expression $\jp(E,\occ[p]{V}.\vec{x})$ denotes the longest sub-expression $E^{(n)}$ in
  $E$, for some $n>1$, such that
  $E^{(n)}= E^{(n-1)} \Join_{\occ[q]{W}.\vec{y} = \occ[r]{U}.\vec{z}} \occ[r]{U}$, for some
  relation name $W$, tuples of attributes $\vec{y}$ and $\vec{z}$ such that $\occ[r]{U}.\vec{z} \in \jc(E, \occ[p]{V}.\vec{x})$,
  if $E^{(n)}$ exists, and $\bot$ otherwise.
\end{definition}
For $E$ and $\occ[p]{V}.\vec{x}$, the estimation $\distEst{E}{\occ[p]{V}.\vec{x}}$ of the
cardinality $\distSet{\restSet{E^\D}{\occ[p]{V}.\vec{x}}}$ is defined as
\begin{equation}\label{eq:dist}
  \small
  \distEst{E}{\occ[p]{V}.\vec{x}}=
  \begin{cases}
    \distSet{\restSet{V^\D}{\vecSet{x}}}, & \text{if } E = V \\[0.5em]
    \displaystyle
    \min\left\{\left\lceil \fv(E')\cdot
        \frac{\cardEst(E)}{\cardEst(E')}\right\rceil\!, \fv(E')\right\},
    & \text{if } \jp(E, \occ[p]{V}.\vec{x}) = E'\neq \bot\\[1em]
    \displaystyle
    \min\left\{\left\lceil\distSet{\restSet{V^\D}{\vecSet{x}}}
        \cdot \frac{\cardEst(E)}{\distSet{V^\D}}\right\rceil\!,
      \distSet{\restSet{V^\D}{\vecSet{x}}}\right\},\!\!\!\!
    & \text{otherwise.}
  \end{cases}
\end{equation}%
where $\distSet{\restSet{V^\D}{\vecSet{x}}}$ is assumed to be a statistic
available after having analyzed the mappings together with the data
instance. Observe that the fractions $\frac{\cardEst(E)}{\cardEst(E')}$ and
$\frac{\cardEst(E)}{\distSet{V^\D}}$ are again scaling factors relying on
assumption~\eqref{eq:uniformity}.  Also, $\distEst{E}{V.\vec{x}}$ must not increase
when the number of joins in $E$ increases, which explains the use of $\min$ for
the case where the number of distinct results in $E$ increases with the number
of joins.
\begin{example}\label{e:f-estimator}
  \begin{figure}[t]
  \scriptsize
\centering
\begin{tabular}{c c c}
  \begin{minipage}{.25\textwidth}
    \vspace{-1.9cm}
  \includegraphics[width=\textwidth]{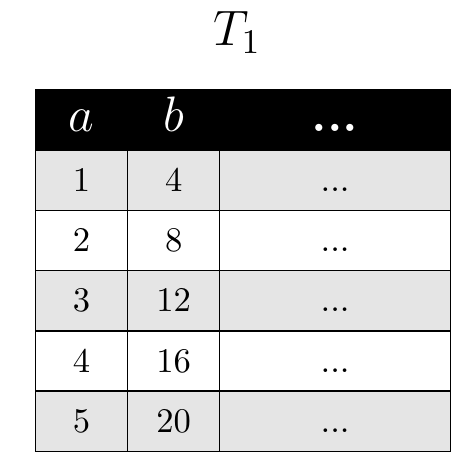}
  \end{minipage}
  &
  \begin{minipage}{.25\textwidth}
    \includegraphics[width=\textwidth]{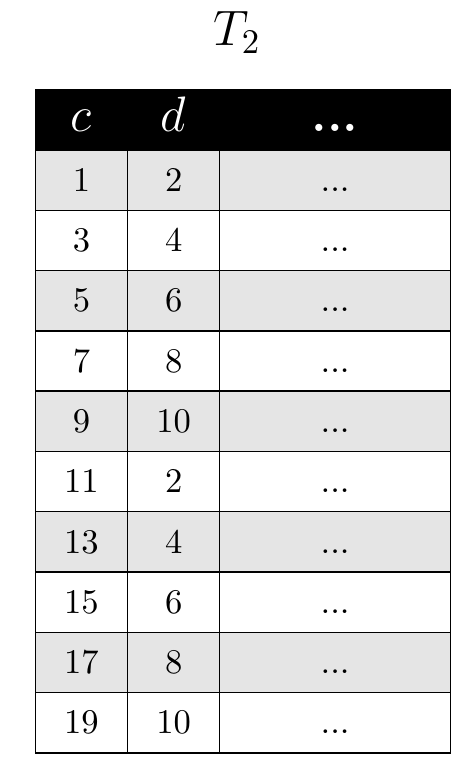}
  \end{minipage}
  &
  \begin{minipage}{.25\textwidth}
    \includegraphics[width=\textwidth]{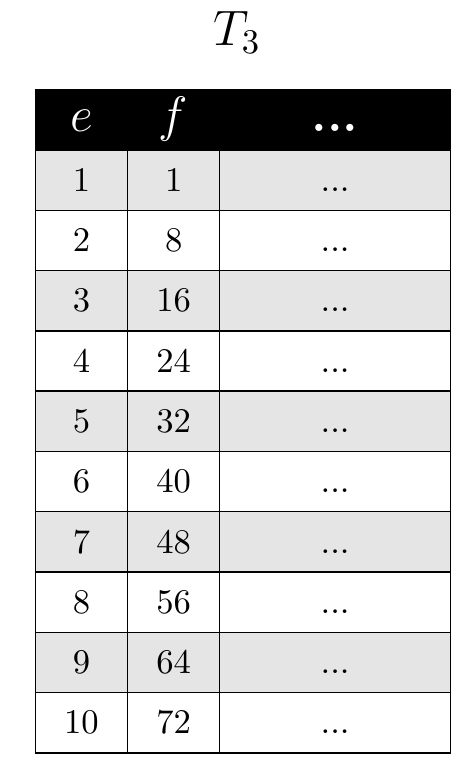}
  \end{minipage}
\end{tabular}
\caption{Data instance $\D$.\label{f:f-estimator-example}}
\end{figure}

Consider the data instance $\D$ from Figure~\ref{f:f-estimator-example}. Relevant statistics are:

\begin{itemize}
\item $\distSet{T_1^\D} = 5$,\quad $\distSet{T_2^\D} = \distSet{T_3^\D} = 10$
\item
  $\distSet{\restSet{T_1^\D}{\textsql{a}}} =
  \distSet{\restSet{T_2^\D}{\textsql{d}}} = 5$,\quad
  $\distSet{\restSet{T_2^\D}{\textsql{c}}} =
  \distSet{\restSet{T_3^\D}{\textsql{f}}} =
  \distSet{\restSet{T_3^\D}{\textsql{e}}} = 10$,
\item
  $\distSet{\restSet{T_1^\D}{\textsql{a}} \capren
   \restSet{T_2^\D}{\textsql{c}}} = 3$,\quad
  $\distSet{\restSet{T_2^\D}{\textsql{d}} \capren
   \restSet{T_3^\D}{\textsql{e}}} = 5$,\quad
  $\distSet{\restSet{T_1^\D}{\textsql{a}} \capren
   \restSet{T_3^\D}{\textsql{f}}} = 1$.
\end{itemize}
We calculate $\cardEst(E)$ for the basic CQ
$E = T_1 \Join_{T_1.\textsql{a}=T_2.\textsql{c}} T_2
\Join_{T_2.\textsql{d}=T_3.\textsql{e}} T_3
\Join_{T_1.{\textsql{a}}=T_3'.\textsql{f}} T_3'$, where $T_3'$ is an alias
(written in this way for notational convenience) for the table $T_3$. To do so,
we first need to calculate the estimations $\cardEst(E^{(1)})$ and
$\cardEst(E^{(2)})$.
{\small
\begin{align}
  \cardEst(E^{(1)}) \enskip = \enskip & \cardEst(T_1
  \Join_{T_1.\textsql{a}=T_2.\textsql{c}} T_2) \enskip = \enskip \left\lceil\frac{\fv(T_1
   \Join_{T_1.\textsql{a}=T_2.\textsql{c}} T_2) \cdot \distSet{T_1^\D} \cdot
   \distSet{T_2^\D}}{\distEst{T_1}{\textsql{a}} \cdot
   \distEst{T_2}{\textsql{c}}}\right\rceil \nonumber\\
  \enskip = \enskip & \left\lceil\frac{\distSet{\restSet{T_1^\D}{\textsql{a}} \capren
    \restSet{T_2^\D}{\textsql{c}}} \cdot \distSet{T_1^\D} \cdot
   \distSet{T_2^\D}}{\distSet{\restSet{T_1^\D}{\textsql{a}}} \cdot
   \distSet{\restSet{T_2^\D}{\textsql{c}}}}\right\rceil \enskip = \enskip \left\lceil(3 \cdot 5 \cdot 10) / (5 \cdot 10)\right\rceil \enskip = \enskip 3 \nonumber\\
  \label{e:cost-application:it2}
  \cardEst(E^{(2)}) \enskip  = \enskip & \cardEst(E^{(1)}
  \Join_{T_2.\textsql{d}=T3.\textsql{e}} T_3)
   = \!\left\lceil\!\frac{\fv(E^{(1)} \Join_{T_2.\textsql{d}=T3.\textsql{e}}
   T_3) \cdot \cardEst(E^{(1)}) \cdot
   \distSet{T_3^\D}}{\distEst{E^{(1)}}{T_2.\textsql{d}} \cdot
   \distEst{T_3}{\textsql{e}}}\right\rceil
\end{align}
}%
By Formula~\eqref{eq:dist}, $\distEst{E^{(1)}}{T_2.\textsql{d}}$ in
Formula~\eqref{e:cost-application:it2} can be calculated as
{\small
\begin{align*}
  \distEst{E^{(1)}}{T_2.\textsql{d}} \enskip = \enskip &
  \min\left\{\left\lceil\frac{\distSet{\restSet{T_2^\D}{\textsql{d}}}}{\distSet{T_2^\D}}
    \cdot
    \cardEst(E^{(1)})\right\rceil, \distSet{\restSet{T_2^\D}{\textsql{d}}} \right\} \\
  \enskip = \enskip & \min\left\{\left\lceil\frac{5}{10} \cdot 3\right\rceil,
    5\right\}
  \enskip = \enskip {\left\lceil\frac{3}{2}\right\rceil} \enskip = \enskip 2
\end{align*}
}%
By Formula~\eqref{eq:k},
$\fv(E^{(1)} \Join_{T_2.\textsql{d}=T3.\textsql{e}} T_3)$ in
Formula~\eqref{e:cost-application:it2} can be calculated as
{\small
\begin{align*}
  \fv(E^{(1)} \Join_{T_2.\textsql{d}=T3.\textsql{e}} T_3) \enskip = \enskip & \left\lceil\frac{\fv(T_2 \Join_{T_2.\textsql{d}=T3.\textsql{e}} T_3)}{\distEst{T_2}{\textsql{d}}} \cdot \distEst{E^{(1)}}{T_2.\textsql{d}}\right\rceil\\
  \enskip = \enskip & \left\lceil\frac{\distSet{\restSet{T_2^\D}{\textsql{d}} \capren
    \restSet{T_3^\D}{\textsql{e}}}}{\distSet{\restSet{T_2^\D}{\textsql{d}}}}
  \cdot \distEst{E^{(1)}}{T_2.\textsql{d}}\right\rceil \enskip = \enskip
  \left\lceil\frac{5}{5} \cdot 2\right\rceil \enskip = \enskip 2
\end{align*}
}%
By plugging the values for $\fv$ and $dist_\D$ in
Formula~\eqref{e:cost-application:it2}, we obtain
{\small
\begin{align*}
  \cardEst(E^{(2)}) \enskip = \enskip & \left\lceil(2 \cdot 3 \cdot 10)/(2 \cdot 10)\right\rceil = 3
\end{align*}
}%
We are now ready to calculate the cardinality of $E$, which is given by the
formula
{\small
\begin{align}\label{e:cost-application:it3}
  \cardEst(E) \enskip & = \enskip \cardEst(E^{(2)} \Join_{T_1.{\textsql{a}}=T_3'.\textsql{f}} T_3') = \left\lceil\frac{\fv(E^{(2)} \Join_{T_1.{\textsql{a}}=T_3'.\textsql{f}} T_3') \cdot \cardEst(E^{(2)}) \cdot \distSet{T_3^\D}}{\distEst{E^{(2)}}{T_1.\textsql{a}} \cdot \distEst{T_3}{\textsql{f}}}\right\rceil
\end{align}
}%
By Formula~\eqref{eq:dist}, $\distEst{E^{(2)}}{T_1.\textsql{a}}$ in
Formula~\eqref{e:cost-application:it3} can be computed as
{\small
\begin{align*}
  \distEst{E^{(2)}}{T_1.\textsql{a}} \enskip = \enskip & \min\left\{\left\lceil\frac{\fv(E^{(1)})}{\cardEst(E^{(1)})} \cdot \cardEst(E^{(2)})\right\rceil, \fv(E^{(1)})\right\}  = \min\left\{\left\lceil\frac{3}{3} \cdot 3\right\rceil, 3\right\} =  3
\end{align*}
}%
Then, by Formula~\eqref{eq:k}, $\fv(E^{(2)} \Join_{T_1.\textsql{a} = T_3'.\textsql{f}} T_3')$ in Formula~\eqref{e:cost-application:it3} can be computed as
{\small
\begin{align*}
  \fv(E^{(2)} \Join_{T_1.\textsql{a} = T_3'.\textsql{f}} T_3') \enskip = \enskip & \left\lceil\frac{\fv(T_1 \Join_{T_1.\textsql{a} = T_3'.\textsql{f}} T_3')}{\distEst{T_1}{\textsql{a}}} \cdot \distEst{E^{(2)}}{T_1.\textsql{a}}\right\rceil = \left\lceil\frac{3}{5}\right\rceil = 1
\end{align*}
}%
By plugging the values for $\fv$ and $\dist_{\D}$ in (\ref{e:cost-application:it3}), we finally obtain
{\small
\begin{align*}
  \cardEst(E) \enskip = \enskip & \left\lceil(1 \cdot 3 \cdot 10)/(3 \cdot 10)\right\rceil = 1
\end{align*}
}%
Observe that, in this example, our estimation is exact, that is, $\cardEst(E) =
\distSet{E^\D}$.
\qedfull
\end{example}

\subsection{Extending Cardinality Estimation to Non-basic CQs}

We study now how to extend the cardinality estimation developed above to
arbitrary join conditions and to CQs with existential variables.

\subsubsection{Extending Cardinality Estimation to Arbitrary Join Conditions}

To extend the cardinality estimation to the case of arbitrary join conditions,
we exploit the following property of theta-joins:
\[\sigma_{\theta_2}(T_1 \Join_{\theta_1} T_2) = T_1 \Join_{\theta_1 \land \theta_2} T_2\]
The cardinality of a CQ without existential variables, and with $n$ atoms and
$m$ join conditions ($m\geq n$) can be estimated in two steps. In the first
step, we estimate through the equations in the previous section the cardinality
of a basic CQ of $n$ atoms over $n-1$ join conditions. In the second step, we
multiply the number obtained in the first step by the probability that the
additional $\theta_n, \ldots, \theta_{m}$ conditions are all satisfied. Such
probability can be easily calculated by relying on the uniformity assumptions.
We illustrate this on the case of three atoms and three join conditions. The
generalization to an arbitrary number of atoms and join conditions is
straightforward.

\begin{example}
  \[
    \cardEst((T_1 \Join_{T_1.\textsql{a}=T_2.\textsql{d}} T_2) \Join_{T_1.\textsql{h}=T_3.\textsql{i} \land T_2.\textsql{d}=T_3.\textsql{e}} T_3) =
    \cardEst((T_1 \Join_{T_1.\textsql{a}=T_2.\textsql{d}} T_2) \Join_{T_1.\textsql{h}=T_3.\textsql{i}} T_3) \cdot \frac{\fv(T_2 \Join_{T_2.\textsql{d}=T_3.\textsql{e}} T_3)}{\distEst{T_2}{T_2.\textsql{d}}}
  \]
  where $\frac{\fv(T_2 \Join_{T_2.\textsql{d}=T_3.\textsql{e}} T_3)}{\distEst{T_2}{T_2.\textsql{d}}}$ is an estimation for $P(T_2.\textsql{d}=T_3.\textsql{e})$ under our assumptions.
\end{example}

\subsubsection{Extending Cardinality Estimation to CQs with Existential
 Variables (Projection)}

\paragraph{Consequences of Cardinality Assumptions}

In this paragraph we discuss some important consequences of our assumptions of uniform distribution across distinct values (Assumption~\eqref{eq:uniformity}) and of independent distributions (Assumption~\eqref{eq:independence}).

Consider a set of columns $C_1, \ldots, C_n$ over a table $T$. Then, Assumption~\eqref{eq:uniformity} can be written as
\begin{equation}\label{eq:c:planning:extend-uniform}
  \forall \vec{v} \in \restBag{T}{C_1,\ldots,C_n} :
  P(\restBag{T}{C_1,\ldots,C_n} = \vec{v}) = \frac{1}{\distSet{\pi_{C_1, \ldots, C_n}T}}.
\end{equation}
A consequence of Assumption~\eqref{eq:independence} is that
\begin{equation}\label{eq:c:planning:product-probs}
  \forall (v_1, \ldots, v_n) \in C_1 \times \cdots \times C_n : P(\restBag{T}{C_1} = v_1 \cap \cdots \cap \restBag{T}{C_n} = v_n) = P(\restBag{T}{C_1} = v_1) \cdot \cdots \cdot P(\restBag{T}{C_n} = v_n).
\end{equation}
By Assumption~\eqref{eq:uniformity}, we know that
\[
\forall i \in \set{1,\ldots,n} : P(\restBag{T}{C_i} = v_i) = \frac{1}{\distSet{\pi_{C_i}T}}.
\]
Hence, Equation~\eqref{eq:c:planning:product-probs} becomes
\begin{equation}\label{eq:c:planning:to-rewrite}
  \forall (v_1, \ldots, v_n) \in C_1 \times \cdots \times C_n : P(\restBag{T}{C_1} = v_1 \cap \cdots \cap \restBag{T}{C_n} = v_n) = \Pi_{i=1}^n \frac{1}{\distSet{\pi_{C_i}T}}. \end{equation}
We observe that Equation~\eqref{eq:c:planning:to-rewrite} can equivalently be
rewritten as
\begin{equation}
  \forall \vec{v} \in \pi_{C_1,\ldots,C_n}T :
  P(\pi_{C_1,\ldots,C_n}T = \vec{v}) = \frac{1}{\Pi_{i=1}^n \frac{1}{\distSet{\pi_{C_i}T}}}.
\end{equation}
Therefore, by Equation~\eqref{eq:c:planning:extend-uniform}, we conclude that
\begin{equation}\label{eq:estimate-projection}
  \distSet{\pi_{C_1, \ldots, C_n}T} = \Pi_{i=1}^n \distSet{\pi_{C_i}T}.
\end{equation}

Let $\pi_{\occ{V^1}.\vec{x}_1, \ldots, \occ{V^n}.\vec{x}_n}(E)$ be a CQ, where
$E$ is a CQ.  Then, by Equation~\eqref{eq:estimate-projection}, we estimate the
cardinality
$\distSet{\pi_{\occ{V^1}.\vec{x}_1, \ldots, \occ{V^n}.\vec{x}_n}(E)}$ through
the formula\footnote{See also
 \url{http://adrem.ua.ac.be/sites/adrem.ua.ac.be/files/db2-selectivity.pdf},
 page~13.}
\[
  \min\{\cardEst(E), \Pi_{i\in \set{1, \ldots, n}} \distEst{E}{\occ{V^i}.\vec{x}_i}\}
\]

% \begin{align}
%   f(E) - \left[f(E) \cdot \left(1 - \frac{\dist(\occ{V^1}.\vec{x}_1,
%   E)}{f(E)}\right) \cdot \cdots \cdot \left(1 -
%   \frac{\dist(\occ{V^n}.\vec{x}_n, E)}{f(E)}\right)\right]
% \end{align}

\subsection{Collecting the Necessary Base Statistics}
\label{s:mappings-analysis}

The estimators introduced above assume a number of statistics to be
available. We now show how to compute such statistics on a data instance by
analyzing the mappings. Consider a set
$\M=\setB{L_i(\vec{f}_i(\vec{v}_i)) \leftsquigarrow V_i(\vec{v}_i)}{1 \le i \le
 n}$ of mappings and a data instance $\D$. We store the statistics:

\begin{enumerate}[(S$_1$)]
\item $\distSet{V_i^\D}$, for each $i \in \{1, \ldots, n\}$;
\item $|\pi_{\vecSet{x}} (V_i^\D)|$, if $f(\vec{x})$ is a term in
  $\vec{f}_i(\vec{v}_i)$, for some function symbol $f$ and
  $i \in \{1, \ldots, n\}$;
\item
  % $\distSet{\pi_{V_i.\vecSet{x}} (V_i \Join_{V_i.\vec{x} = V_j.\vec{y}}
  % V_j)}$
  $\distSet{\restSet{V_i^\D}{\vec{x}} \capren \restSet{V_j^\D}{\vec{y}}}$, if
  $f(\vec{x})$ is a term in $\vec{f}_i(\vec{v}_i)$, and $f(\vec{y})$ is a term
  in $\vec{f}_j(\vec{v}_j)$, for some function symbol $f$ and
  $i,j \in \{1, \ldots, n\}$, $i\neq j$.
\end{enumerate}

Statistics~(S$_1$) and~(S$_2$) are required by all three estimators that we have
introduced, and can be measured directly by evaluating source queries on
$\D$. Statistics~(S$_3$) can be collected by first iterating over the function
symbols in the mappings, and then calculating the cardinalities for joins over
pairs of source queries whose corresponding mapping targets have a function
symbol in common.  It is easy to check that Statistics~(S$_1$)--(S$_3$) suffice for
our estimation, since all joins in a CQ are between source queries, and
moreover, every translation calculated according to Definition~\ref{d:unf-cq}
contains only joins between pairs of source queries considered by
Statistics~(S$_3$). The same intuitions apply for the possible projections
applied to a basic CQ in an unfolding.

\subsection{Cardinality  Estimation of an Unfolding}

We now show how to estimate the cardinality of an unfolding\footnote{Regardless of whether the unfolding is of a UCQ or a JUCQ, as they \emph{must} be equivalent.}
  by using the
formulae~\eqref{eq:f},~\eqref{eq:k}, and~\eqref{eq:dist} introduced for
cardinality estimation. The next theorem shows that such estimation can be
calculated by summing-up the estimated cardinalities for each CQ in the
unfolding of the input query, provided that
\begin{inparaenum}[\itshape (i)]
\item the unfolding is being calculated over the \emph{wrap} of the set of mappings, and
\item the query to unfold is a CQ.
\end{inparaenum}

We start our discussion with an auxiliary lemma about rules in an unfolding.

\begin{lemma}\label{l:basic-Datalog-fact}
  Consider an unfolding $(q(\vec{v}), \Pi)$ and two rules
  $r_1, r_2 \in \Pi$. If $\head(r_1)$ and $\head(r_2)$ are not
  unifiable, then for each database instance $\D$ for $\Sigma$, we
  have that

  \[(q(\vec{v}), \setNoFancy{r_1})^\D \cap (q(\vec{v}), \setNoFancy{r_2})^\D = \emptyset.\]
\end{lemma}

\begin{proof}
  Let $A_1 = (q(\vec{v}), \setNoFancy{r_1})^\D$, and
  $A_2 = (q(\vec{v}), \setNoFancy{r_2})^\D$.  W.l.o.g., let
  $\head(r_1) = q(\vec{f}(\vec{x}))$, and
  $\head(r_2) = q(\vec{g}(\vec{y}))$. Since $\vec{f}(\vec{x})$ and
  $\vec{g}(\vec{y})$ are not unifiable, it must be that
  $\vec{f} \neq \vec{g}$. Since no constants occur in either
  $\head(r_1)$ or $\head(r_2)$, it directly follows from the
  definition of answer of a Datalog query over an instance $\D$ that
  $A_1 \subseteq \setBNoFancy{\vec{f}(\vec{a})}{\vec{a} \in
    \adom(\D)}$, and
  $A_2 \subseteq \setBNoFancy{\vec{g}(\vec{a})}{\vec{a} \in
    \adom(\D)}$, where $\adom(\D)$ denotes the set of constants occurring in $\D$\footnote{Called \emph{active domain}, see~\cite{AbHV95}}.
  The thesis follows by observing that
  $\setBNoFancy{\vec{f}(\vec{a})}{\vec{a} \in \adom(\D)} \cap
  \setBNoFancy{\vec{g}(\vec{a})}{\vec{a} \in \adom(\D)} = \emptyset$.
    % The thesis follows by observing that, for every
    % database instance $\D$,
    % $\setBNoFancy{\vec{f}^\D(\vec{a})}{\vec{a} \in \adom(\D)} \cap
    % \setBNoFancy{\vec{g}^\D(\vec{a})}{\vec{a} \in \adom(\D)} =
    % \emptyset$, since for the standard name assumption
    % $\vec{f}^\D = \vec{f} \neq \vec{g} = \vec{g}^\D$.
  \end{proof}

\begin{theorem}\label{t:est-unf}
  Consider a CQ
  $q(\vec{x}) \leftarrow L_1(\vec{v}_1), \ldots, L_n(\vec{v}_n)$ such
  that $\vecSet{x} = \bigcup_{i=1}^n \vecSet{v_i}$. Then,

  \[|\unf(q(\vec{x}), \M)^\D| = \sum_{q_{u} \in \unf(q, \wrap(\M))} |q_{u}(\vec{x})^\D| \]

\begin{proof}
  Since by Lemma~\ref{l:wrap-mapping} $\M \equiv \wrap(\M)$, it
  suffices to prove that for each pair of different queries $q_u, q_u'$
  in $\unf(q(\vec{x}), \wrap(\M))$ it holds that
  $q_u^\D \cap q_u'^\D = \emptyset$. Consider two arbitrary queries
  $q_u, q_u'$ in $\unf(q(\vec{x}), \wrap(\M))$. By
  Definition~\ref{d:unf-cq}, there must exist two different lists of
  mappings $(m_1, \ldots, m_n)$ and $(m_1', \ldots, m_n')$ in
  $\wrap(\M)$ and two substitutions $\sigma$, $\sigma'$ such that:
  \begin{enumerate}[\itshape (i)]
  \item $\head(q_u) = q_u(\sigma(\vec{x}))$, $\head(q_u') = q_u'(\sigma'(\vec{x}))$;
  \item $\target(m_i) = L_i(\sigma(\vec{v}_i))$, $\target(m_i') = L_i(\sigma'(\vec{v}_i))$, for each $i \in \set{1, \ldots, n}$.
  \end{enumerate}
%and $\pred(\target(m_i)) = L_i = \pred(\target(m_i'))$, for each $i \in \set{1, \ldots, n}$.
%W.l.o.g., let $\sigma(\vec{x}) = \vec{f}(\vec{y})$, and $\sigma'(\vec{x}) = \vec{f}'(\vec{y}'))$.

  W.l.o.g., let $m_i \neq m_i'$, for some $i \in \set{1, \ldots,
    n}$. Moreover, assume that
  $\sigma(\vec{v}_i) = \vec{f}_i(\vec{x}_i)$ and
  $\sigma'(\vec{v}_i) = \vec{f}_i'(\vec{x}_i')$, for some tuples of
  terms $\vec{f}_i(\vec{x}_i)$ and $\vec{f}_i'(\vec{x}_i')$. By~\emph{(ii)} and 
  Lemma~\ref{l:wrap-template-sharing}, it must be
  $\vec{f}_i \neq \vec{f}_i'$. Hence, $\vec{f}_i(\vec{x}_i)$ and
  $\vec{f}_i'(\vec{x}_i')$ do not unify. Therefore, also
  $\sigma(\vec{v}_i)$ and $\sigma'(\vec{v}_i)$ do not unify. Since, by
  assumption, $\vecSet{x} \supseteq \vecSet{v_i}$, we conclude that
  $\sigma(\vec{x})$ does not unify with $\sigma'(\vec{x})$. Then, by
  Lemma~\ref{l:basic-Datalog-fact}, it follows that
\[q_{u}(\sigma(\vec{x}))^\D \cap q_{u}'(\sigma'(\vec{x}))^\D = \emptyset\]

By~\emph{(i)}, it holds that $q_u^\D \cap q_u'^\D = \emptyset$. Since the choice of $q_u, q_u'$ was arbitrary, we conclude the thesis.
%\[\forall 1 \le i \neq i' \le m\ :\ \forall \D\ :\  q_{i}(\sigma_i(\vec{x}))^\D \cap q_{i'}(\sigma_{i'}(\vec{x}))^\D = \emptyset\]
%from which the theorem directly follows.
\end{proof}
\end{theorem}

The previous theorem states that the cardinality of the unfolding of a query
over the wrap of a set of mappings corresponds to the sum of the cardinalities of each CQ
in the unfolding, under the assumption that all the attributes are kept in the
answer. Intuitively, the proof relies on the fact that, when the wrap of a set of mappings is used, each
CQ in the unfolding returns answer variables built using a specific combination
of function names. Hence, to calculate the cardinality of a CQ $q$, it
suffices to collect statistics as described in the previous paragraph, but over
$\wrap(\M)$ rather than $\M$, and sum up the estimations for each CQ in
$\unf(q,\wrap(\M))$.

The method above might overestimate the actual cardinality if the input CQ contains non-answer variables.
In Section~\ref{ss:deal-with-proj} we show how to address this limitation
by storing, for each property in the mappings, the probability of having duplicate answers if the projection operation
is applied to one of the (two) arguments of that property.

Also, the method above assumes a CQ as input to the unfolding, whereas a
perfect rewriting is in general a UCQ. This is usually not a critical aspect,
especially in practical applications of OBDA.  By using a saturated set of
mappings (called T-mapping~\cite{RoCa12}) $\M_\T$ in place of $\M$, in fact,
the perfect rewriting of an input CQ $q$ \emph{almost always}~\cite{KiKZ12}
coincides with $q$ itself\footnote{\emph{Always}, if the CQ is interpreted as a
 SPARQL query and evaluated according to the OWL~2~QL entailment regime, or if
 the CQ does not contain existentially quantified variables.}.  Hence, in most
cases we can directly use in Theorem~\ref{t:est-unf} the input query $q$, if we
use $\wrap(\M_\T)$ instead of $\wrap(\M)$. In the following subsection we
provide a fully detailed example of how this is done, as well as more details
on T-mappings and related techniques.

\subsection{A Note on Rewriting}\label{c:prel:s:rewriting-modern-obda}

Modern OBDA systems such as \ontop or Ultrawrap~\cite{ultrawrap} use advanced
rewriting techniques for the rewriting phase. In particular, \ontop uses a
combination of the \emph{T-mapping} and \emph{tree-witness
 rewriting}~\cite{ontop-system} techniques. We now give basic notions about
these two techniques that will be useful in the rest of this paper.

\begin{definition}
  The \emph{saturation $\A_\T$ of an ABox $\A$ with respect to a TBox $\T$} is the ABox
  \[\setB{L(\vec{a})}{(\T, \A) \models L(\vec{a}), \vec{a}\text{ is a tuple of individuals} }\]

\end{definition}

In OBDA, saturated virtual ABoxes can obtained by compiling the ontology in the mappings so as to obtain a set of mappings called \emph{T-mapping}~\cite{RoKZ13}. Intuitively, a T-mapping exposes a saturated ABox.

\begin{definition}
  Given an OBDA specification $\mathcal{S} = (\T, \M, \Sigma)$, the mapping set of mappings $\M_\T$ is a T-mapping for $\S$ if, for every OBDA instance $\D$ of $\mathcal{S}$, $\A_{(\M,\D)_{\T}} = \A_{(\M_\T,\D)}$.
\end{definition}

The saturation of an ABox is also known as a \emph{forward chaining} technique, as opposed to \emph{backward chaining} techniques such as rewriting algorithms. It is easy to see that, for each saturated ABox $\A_\T$ and \emph{single-atom query} $q(\vec{x}) \leftarrow X(\vec{x})$, it holds that $q(\vec{x})^{\A_{\T}} = \cert(q(\vec{x}), (\T, \A))$. However, such equality does not hold in case $q$ is a general CQ.

Authors from~\cite{GKKP*14} gave a characterization of queries for which query
answering with respect to an ontology corresponds to query evaluation over
saturated ABoxes. The \emph{tree-witness rewriting} technique~\cite{GKKP*14}
exploits this characterization so as to produce perfect rewritings that are
minimal with respect to a saturated ABox. In short, queries for which query
answering with respect to an ontology \emph{does not} correspond to query
evaluation over saturated ABoxes are the ones whose answers are influenced by
existentials in the right-hand-side of \dlliteR TBox axioms. Such queries
are called queries with \emph{tree-witnesses}.

\begin{example}\label{e:long-example}
Consider the data instance $\D$ from Figure~\ref{f:long-example}.
\begin{figure}[t]
  \scriptsize
  \centering
  \begin{tabular}{c c c}
    \begin{minipage}{.12\textwidth}
  \includegraphics[width=\textwidth]{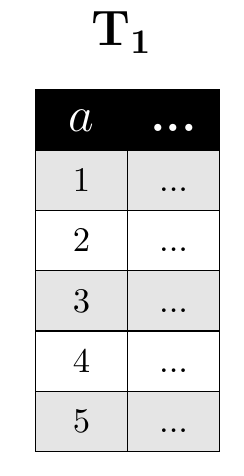}
    \end{minipage}
    &
    \begin{minipage}{.12\textwidth}
      \includegraphics[width=\textwidth]{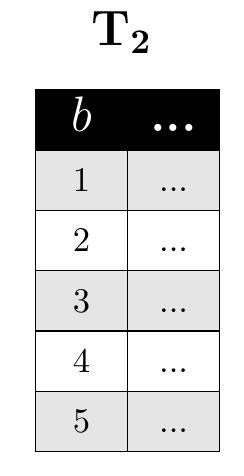}
    \end{minipage}
    &
    \begin{minipage}{.25\textwidth}
      \includegraphics[width=\textwidth]{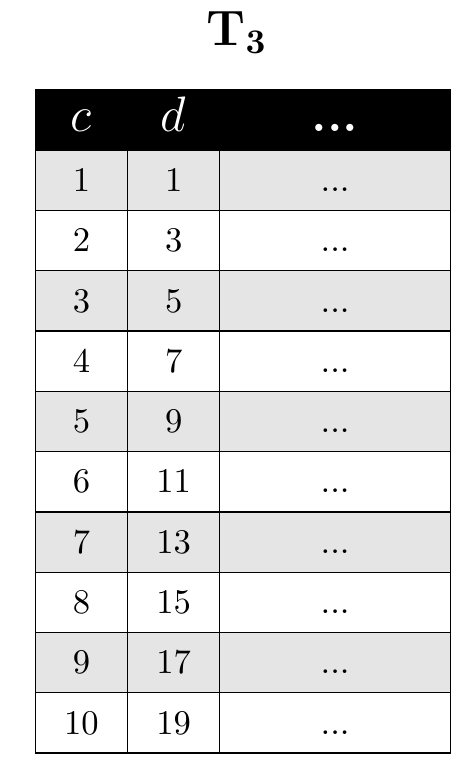}
    \end{minipage}\\
    \begin{minipage}{.25\textwidth}
      \includegraphics[width=\textwidth]{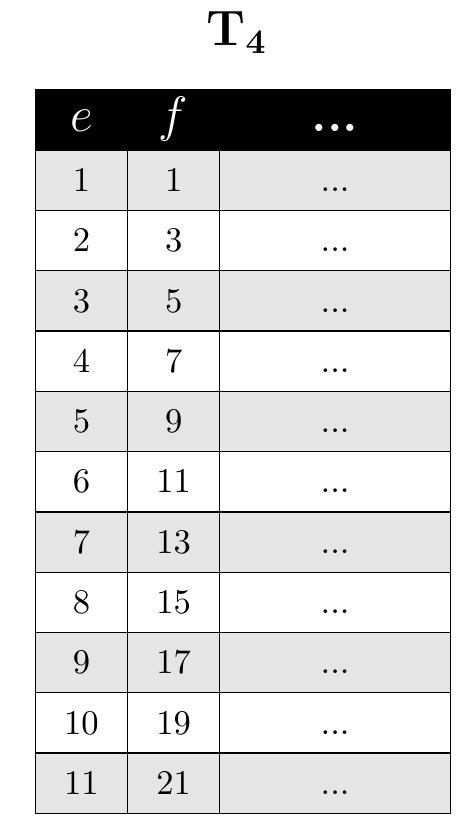}
    \end{minipage}
    &
    \begin{minipage}{.25\textwidth}
      \includegraphics[width=\textwidth]{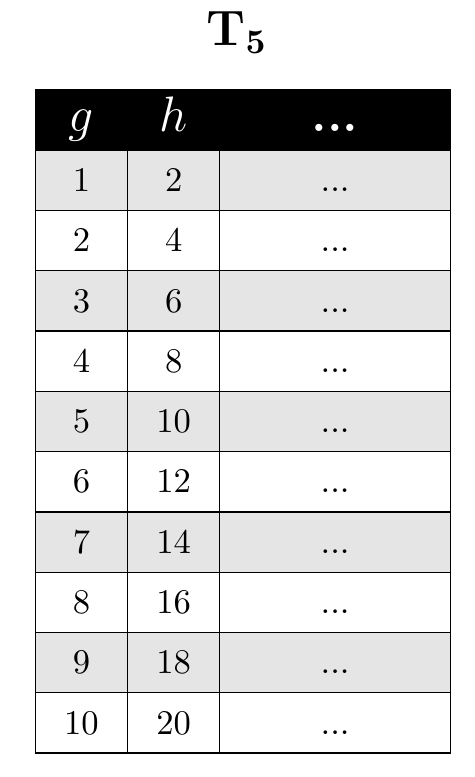}
    \end{minipage}
    &
    \begin{minipage}{.25\textwidth}
      \includegraphics[width=\textwidth]{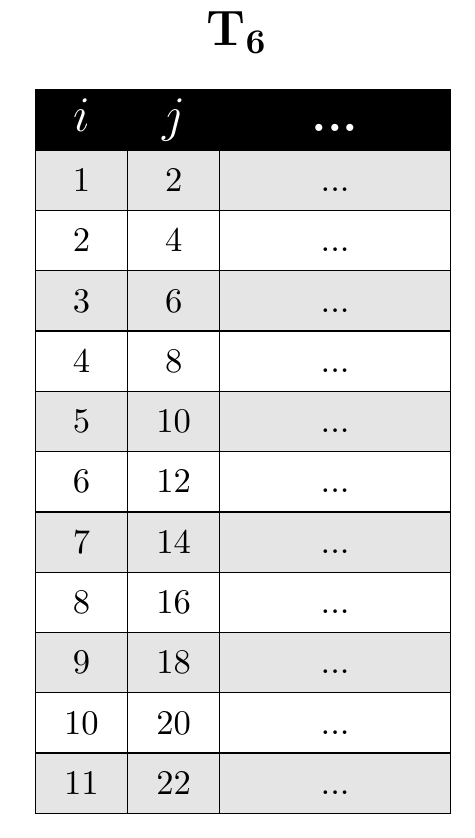}
    \end{minipage}
  \end{tabular}
\caption{\label{f:long-example} Data Instance $\D$.}
\end{figure}
Consider the query
\[q(x,y,z) \leftarrow C(x), R_1(x,y), R_2(x,z)\]
\noindent and the set of mappings
\[\M = \hspace{-1.5ex}
\begin{tabular}{l l l @{\qquad\qquad} l l l}
  \ldelim\{{4}{3mm} & $A(l(a))$ & $\leftsquigarrow T_1(a)$ & $P_2(l(g), m(h))$ & $\leftsquigarrow T_5(g,h)$ & \rdelim\}{4}{3mm} \tableRowBreak
  & $C(l(b))$ & $\leftsquigarrow T_2(b)$ & $P_2(l(g),n(h))$ & $\leftsquigarrow T_5(g,h)$ \tableRowBreak
  & $P_1(l(c),m(d))$ & $\leftsquigarrow T_3(c,d)$ & $R_2(l(i),m(j))$ & $\leftsquigarrow T_6(i,j)$ \tableRowBreak
  & $R_1(l(e), m(f))$ & $\leftsquigarrow T_4(e,f)$ &
\end{tabular}
\]

Let $\O = \set{A \sqsubseteq C, P_1 \sqsubseteq R_1, P_2 \sqsubseteq R_2}$. Since all the variables in $q$ are answer variables, to compute $\cert(q, (\T, \A_{(\M, \D)})$ it suffices to compute the evaluation $q^{\can{\A_{(\M_{\T}, \D)}}}$, where $\M_\T$ is a T-mapping of $\M$ with respect to $\T$. Such T-mapping is:
\[\M_\T = \M \cup \hspace{-1.5ex}
\begin{tabular}{l l l @{\qquad} l l l}
  \ldelim\{{2}{3mm} & $C(l(a))$ & $\leftsquigarrow T_1(a)$ & $R_2(l(c), m(h))$ & $\leftsquigarrow T_5(g,h)$ & \rdelim\}{2}{4mm} \tableRowBreak
  & $R_1(l(c),m(d))$ & $\leftsquigarrow T_3(c,d)$ & $R_2(l(g),n(h))$ & $\leftsquigarrow T_5(g,h)$ \tableRowBreak
\end{tabular}
\]

The wrap $\wrap(\M_\T)$ of $\M_\T$ is the set
\[
\begin{tabular}{l l l @{\qquad\qquad} l l l}
  \ldelim\{{4}{3mm} & $A(l(a))$ & $\leftsquigarrow T_1(a)$ & $P_2(l(g), m(h))$ & $\leftsquigarrow T_5(g,h)$ & \rdelim\}{4}{3mm} \tableRowBreak
  & $C(l(v))$ & $\leftsquigarrow V_1(v)$ & $P_2(l(g),n(h))$ & $\leftsquigarrow T_5(g,h)$ \tableRowBreak
  & $P_1(l(c),m(d))$ & $\leftsquigarrow T_3(c,d)$ & $R_2(l(z_1),m(z_2))$ & $\leftsquigarrow V_3(z_1,z_2)$ \tableRowBreak
  & $R_1(l(w_1), m(w_2))$ & $\leftsquigarrow V_2(w_1,w_2)$ & $R_2(l(g),n(h))$ & $\leftsquigarrow T_5(g,h)$
\end{tabular}
\]
%%%%%%%%%%%%%%%%%%%%%%%%%%%
where $V_1, \ldots, V_3$ are view names for Datalog queries of programs respectively
\begin{center}
  \begin{tabular}{l l l}
    $\Pi_1$ & = & $\set{V_1(a) \leftarrow T_1(a), V_1(b) \leftarrow T_2(b)}$ \tableRowBreak
    $\Pi_2$ & = & $\set{V_2(c,d) \leftarrow T_3(c,d), V_2(e,f) \leftarrow T_4(e,f)}$ \tableRowBreak
    $\Pi_3$ & = & $\set{V_3(g,h) \leftarrow T_5(g,h), V_3(i,j) \leftarrow T_6(i,j)}$. \tableRowBreak
  \end{tabular}
\end{center}
To estimate the cardinality of $q^{\can{\A_{(\M_{\T}, \D)}}}$ by using our cardinality estimator, we need the following statistics:
\begin{itemize}
\item $\distSet{\pi_{v}(V_1)} = 5$
\item $\distSet{\pi_{w_1,w_2}(V_2)} = 11$
\item $\distSet{\pi_{z_1,z_2}(V_3)} = 11$
\item $\distSet{\pi_{g,h}(T_5)} = 10$
\item $\distSet{\restSet{V_2}{w_1}} = \distSet{\restSet{V_2}{w_2}} = \distSet{\restSet{V_3}{z_1}} = \distSet{\restSet{V_3}{z_2}} = 11$
\item $\distSet{\restSet{T_6}{i}} = \distSet{\restSet{T_6}{j}} = 10$
\item $\distSet{\restSet{V_1}{v} \capren \restSet{V_2}{w_1}} = \distSet{\restSet{V_1}{v} \capren \restSet{V_3}{z_1}} = 5$
\item $\distSet{\restSet{V_1}{v} \capren \restSet{T_5}{g}} = 5$
\end{itemize}

By applying Definition~\ref{d:unf-cq}, the unfolding $\unf(q, \wrap(\M_\T))$ is the NRLP query $(q(x,y,z), \Pi_{\unf})$, where:
\[\Pi_{\unf} = \hspace{-1.5ex}
\begin{tabular}{l l l l l}
  \ldelim\{{2}{3mm} & $q(l(v),m(w_2),m(z_2))$ & $\leftarrow$ & $V_1(v), V_2(v, w_2), V_3(v,z_2)$ & \rdelim\}{2}{0mm} \tableRowBreak
  & $q(l(v),m(w_2),n(h))$ & $\leftarrow$ & $V_1(v), V_2(v, w_2), T_5(v,h)$ \tableRowBreak
\end{tabular}
\]

In relational algebra:

\[E_1 = \pi_{x/l(v), y/m(w_2), z/m(z_2)}(V_1 \Join_{V_1.v=V_2.w_2} V_2 \Join_{V_1.v = V_3.z_2} V_3)\]
  \[\cup\]
\[E_2 = \pi_{x/l(v), y/m(w_2), z/n(h)}(V_1 \Join_{V_1.v=V_2.w_2} V_2 \Join_{V_1.v = T_5.a} T_5)\]

According to Theorem~\ref{t:est-unf}, the cardinality of $q_{\unf}$ in $\D$ can be calculated as $|E_1| + |E_2|$. Such cardinalities
can be computed in the same fashion as in Example~\ref{e:f-estimator}.
\end{example}

\subsection{Extending Cardinality Estimation to Input Queries with
  Projection}\label{ss:deal-with-proj}

In the presence of projection in the input query, our estimation for
the cardinality of an unfolding is an upper-bound of the actual
cardinality, if the cardinalities for each single CQ are estimated
correctly. In this section we discuss a lower-bound of the actual
cardinality of an unfolding.

We first recall how to compute the cardinality of a union of sets.
Given sets $A, B$, the following equation holds:
\begin{equation}\label{eq:card-set-union}
  |A \cup B| = |A| + |B| - |A \cap B|.
\end{equation}
  In general, given $n$ sets $A_1, \ldots,
  A_n$, %$|\bigcup_i A_i| = \sum_{i=1}^n (|A_i| + (-1)^{n-1} |\bigcap_{i=1}^n A_i|)$.
  it holds that
\begin{equation}\label{eq:card-set-union-general}
    \Big|\bigcup_i A_i\Big| ~=~ \sum_{i=1}^n \Bigg( \sum_{\substack{ J\subseteq \{1,\ldots,n\} \\|J|=i }} ((-1)^{i-1} \cdot \Big|\bigcap_{j\in J} A_{j}\Big|)\Bigg).
\end{equation}

 Recall that our estimator stores the cardinalities of the
 intersections of \emph{pairs of joinable columns}. An idea could be
 to estimate the cardinality of a union between $n$ sets by exploiting
 such cardinalities. Observe that, according to the equations above,
 the information about the cardinality of the intersection of pairs of
 sets is sufficient to calculate the cardinality of the union of two
 sets. However, it is not sufficient to calculate the cardinality of
 the union of more than two sets. Still, one can estimate a lower
 bound for such cardinality, by assuming that for three arbitrary sets
 $A$, $B$, $C$, we have that
\begin{equation}\label{eq:a4} \tag{a4}
  |A \cap B \cap C| ~=~ \min\left\{|A \cap B|, |A \cap C|, |B \cap C|\right\}.
\end{equation}
Intuitively, we are assuming that elements are shared \emph{as much as
  possible} between the three sets $A$, $B$, $C$.

Following this idea, in the next example we show how to calculate the
cardinality of a union of $n$ sets by using a set of linear
equations, and by relying only on the information about the
cardinalities of the intersections of pairs of sets.

\begin{example}\label{e:linear-equations}
  Consider the sets $A = \set{1,2,3,4}$, $B = \set{1,2,5,6}$, $C = \set{1,2,3,4}$, and $D = \set{5,6,7,8}$. We consider the following cardinalities to be available:
  \begin{itemize}
  \item $|A| = |B| = |C| = |D| = 4$
  \item $|A \cap B| = 2$, $|A \cap C| = 4$, $|A \cap D| = 0$, $|B \cap C| = 2$, $|B \cap D| = 2$, $|C \cap D| = 0$.
  \end{itemize}
  For each expression of the form $|S_1 \cap \cdots \cap S_k|$, we introduce
  a variable $x_{S_1\cdots S_k}$.

  In order to calculate $|A \cup B \cup C\cup D|$, by applying
  Equation~\eqref{eq:card-set-union-general} we first need to calculate
  $|A \cap B \cap C|$, $|A \cap B \cap D|$, $|B \cap C \cap D|$,
  $|A \cap C \cap D|$, and $|A \cap B \cap C \cap D|$. Boundaries for
  such values can be found by using the following system of linear
  inequalities:
\[
  \begin{array}{l}
      x_{ABC} \le \min\{x_{AB},x_{AC},x_{BC}\}\\
      x_{ABD} \le \min\{x_{AB}, x_{AD}, x_{BD}\}\\
      x_{BCD} \le \min\{x_{BC}, x_{CD}, x_{BD}\}\\
      x_{ACD} \le \min\{x_{AC}, x_{CD}, x_{AD}\} \\
      x_{ABCD} \le \min\{x_{ABC}, x_{ABD}, x_{BCD}, x_{ACD}\}\\
      x_{AB} = 2, x_{AC} = 4, x_{AD} = 0, x_{BC} = 2, x_{BD} = 2, x_{CD} = 0
  \end{array}
\]  %
By  Assumption~\eqref{eq:a4}, we obtain:
\[
\begin{array}{l}
      x_{ABC} = \min\{x_{AB},x_{AC},x_{BC}\} = \min\{2, 4, 2\} = 2\\
      x_{ABD} = \min\{x_{AB}, x_{AD}, x_{BD}\} = \min\{2, 0, 2\} = 0\\
      x_{BCD} = \min\{x_{BC}, x_{CD}, x_{BD}\} = \min\{2, 0, 2\} = 0\\
      x_{ACD} = \min\{x_{AC}, x_{CD}, x_{AD}\} = \min\{4, 0, 0\} = 0\\
      x_{ABCD} = \min\{x_{ABC}, x_{ABD}, x_{BCD}, x_{ACD}\} = \min\{2, 0, 0, 0\} = 0\\
      x_{AB} = 2, x_{AC} = 4, x_{AD} = 0, x_{BC} = 2, x_{BD} = 2, x_{CD} = 0
\end{array}
\]
  By applying the results of the system in Equation~\eqref{eq:card-set-union-general}, we obtain:
  \begin{align*}
    4\cdot 4 - (2 + 4 + 0 + 2 + 2 + 0) + (2 + 0 + 0 + 0) - 0 & = 16 - 10 + 2  =  8
  \end{align*}
  Observe that, in this example, the obtained lower bound $8$ corresponds to the actual cardinality of $A \cup B \cup C \cup D$. \qedfull
\end{example}

Observe that the previous approach for calculating the cardinality of
a union of sets comes at the cost of calculating an expression of
length exponential in the size of the union. Thus, this approach is
feasible only for unions of a small number of sets.

We now show how the approach above, or any other approach able to
estimate the cardinality for the union of sets, can be used to
calculate the cardinality of an unfolding in the presence of
projection. To do so, we first need to introduce the auxiliary notion
of \emph{answer template matrix}.

\begin{definition}
  Let $q(\vec{x}) \leftarrow L_1(\vec{v}_1), \ldots, L_n(\vec{v}_n)$ be a CQ,
  and consider the unfolding $\unf(q,\M) = (q_{\unf}(\vec{x}), \Pi)$ of $q$
  with respect to $\M$. W.l.o.g., we assume $\Pi$ to be of the form
  \[\Pi = \setB{q_{\unf}(\vec{f}_j(\vec{y}_j)) \leftarrow V_1^j \land \cdots
     \land V_n^j}{1 \le j \le m}.
  \]
  Then, the \emph{answer template matrix (ATM)} of $q_{\unf}$ is the matrix
\[
  \atm(q_{\unf}) =
  \begin{bmatrix}
    \vec{f}_1\\
    \vdots \\
    \vec{f}_m
  \end{bmatrix}
\]
of tuples of function templates $\vec{f}_1,\ldots,\vec{f}_m$.
\end{definition}

It is possible to use the notion of ATM to improve our estimation for the
cardinality of an unfolding in the presence of the projection operator.  The
estimation proceeds differently based on whether
\begin{inparaenum}[\itshape (i)]
\item the ATM does not contain duplicate rows, or
\item the ATM contains duplicate rows.
\end{inparaenum}

We start our discussion with case~\emph{(i)}.

\paragraph{Estimation without Duplicate Rows in the ATM.}

By arguing as in the proof of Theorem~\ref{t:est-unf}, it is immediate to see
that the ATM of an unfolding of a CQ $q$ with respect to $\wrap(\M)$ \emph{does not}
contain repeated rows if $q$ does not contain existentially quantified
variables. By exploiting the notion of ATM, the next result extends
Theorem~\ref{t:est-unf} to a more general case.

\begin{theorem}\label{t:est-unf-1}
  Consider a CQ $q(\vec{x}) \leftarrow L_1(\vec{v}_1), \ldots,
  L_n(\vec{v}_n)$. Then, if $\atm(\unf(q(\vec{x}), \wrap(\M)))$ does not
  contain repeated rows, the following holds:
  \[
    |\unf(q(\vec{x}), \M)^\D| = \sum_{q_{u} \in \unf(q, \wrap(\M))}
    |q_{u}(\vec{x})^\D|
  \]
\end{theorem}
\begin{proof}
  Since by Lemma~\ref{l:wrap-mapping} $\M \equiv \wrap(\M)$, it
  suffices to prove that for each pair of different queries
  $q_u, q_u'$ in $\unf(q(\vec{x}), \wrap(\M))$ it holds that
  $q_u^\D \cap q_u'^\D = \emptyset$. Consider two arbitrary queries
  $q_u, q_u'$ in $\unf(q(\vec{x}), \wrap(\M))$.  By
  Definition~\ref{d:unf-cq}, there must exist two different lists of
  mappings $(m_1, \ldots, m_n)$ and $(m_1', \ldots, m_n')$ in
  $\wrap(\M)$ and two substitutions $\sigma$, $\sigma'$ such that
  \begin{equation}
    \label{eq:head-unification}
    \head(q_u) = q_u(\sigma(\vec{x})) \qquad\text{and}\qquad
    \head(q_u') = q_u'(\sigma'(\vec{x})).
  \end{equation}
  Since the ATM of the unfolding does not contain repeated rows, then
  necessarily $\sigma(\vec{x})$ and $\sigma'(\vec{x})$ are not unifiable.
  Then, by Lemma~\ref{l:basic-Datalog-fact}, it follows that
  \[
    q_{u}(\sigma(\vec{x}))^\D \cap q_{u}'(\sigma'(\vec{x}))^\D = \emptyset
  \]
  By Equations~\eqref{eq:head-unification}, it holds that
  $q_u^\D \cap q_u'^\D = \emptyset$. Since the choice of $q_u, q_u'$ was
  arbitrary, we conclude the thesis.
\end{proof}

Theorem~\ref{t:est-unf-1} gives us a less restrictive condition for calculating
the cardinality of an unfolding, as it essentially says that projecting out
variables does not change the number of results as long as duplicate rows do
not appear in the ATM of the unfolding.

\paragraph{Estimation with Duplicate Rows in the ATM.}

We now study the problem of estimating the number of results of an unfolding
for which the ATM contains duplicate rows. First, observe that any unfolding
can be partitioned into several different UCQs where either the ATM \emph{does
 not} contain duplicate rows, or all rows are equal.  Summing up the
cardinalities of such partitions would give the cardinality of the original
unfolding.  Hence, we can calculate the cardinality of a general unfolding by
handling such two cases.  In the previous paragraph we have already seen how to
deal with the case in which there are no duplicate rows in the ATM. In this
paragraph we study the case in which all rows of the ATM are equal.
For this we exploit another useful property of the unfolding with respect to
the wrap of a set of mappings.

\begin{theorem}\label{t:single-view-theorem}
  Let $q(\vec{x}) \leftarrow L_1(\vec{v}_1), \ldots, L_n(\vec{v}_n)$ be a CQ
  for which the unfolding $\unf(q,\wrap(\M))$ of $q$
  with respect to $\wrap(\M)$ is such that all rows in the ATM are equal to
  some tuple $\vec{f}$ of templates.  Let $L_i(\vec{v}_i)$ be an atom in $q$
  such that $\vecSet{x} \supseteq \vecSet{v}_i$. Then, for each pair of rules
  \[
    \begin{array}{rrcl}
      r_1 = & q_{u}(\vec{f}(\vec{y})) &\leftarrow& V_1(\vec{y}_1),
      \ldots, V_i(\vec{y}_i), \ldots, V_n(\vec{y}_n),\\
      r_2 = & q_{u}'(\vec{f}(\vec{y'})) &\leftarrow& V_1'(\vec{y}_1'), \ldots,
      V_i'(\vec{y}_i'), \ldots, V_n'(\vec{y}_n')
    \end{array}
  \]
  in the program of $\unf(q,\wrap(\M))$, it holds that $V_i = V_i'$.
\end{theorem}
\begin{proof}
  By  Definition~\ref{d:unf-cq}, there must exist two different lists of
  mappings $(m_1, \ldots, m_n)$ and $(m_1', \ldots, m_n')$ in
  $\wrap(\M)$ and two substitutions $\sigma$, $\sigma'$ such that:
  \begin{enumerate}[\itshape (i)]
  \item $q_u(\vec{f}(\vec{y})) = q_{u}(\sigma(\vec{x}))$,
    $q'_u(\vec{f}(\vec{y'})) = q'_u(\sigma'(\vec{x}))$;
  \item $\target(m_j) = L_j(\sigma(\vec{v}_j))$, $\target(m_j') =
    L_j(\sigma'(\vec{v}_j))$, for each $j \in \set{1, \ldots, n}$;
  \item $\sigma(\source(m_j)) = V_j(\vec{y}_j)$, $\sigma'(\source(m_j')) =
    V'_j(\vec{y}'_j)$, for each $j \in \set{1, \ldots, n}$.
  \end{enumerate}

  By~\emph{(iii)}, in order to show that $V_i = V_i'$, it suffices to prove
  that $m_i = m_i'$.  Assume by contradiction that $m_i \neq m_i'$.  Let
  $\sigma(\vec{v}_i) = \vec{f}_i(\vec{x}_i)$ and
  $\sigma'(\vec{v}_i) = \vec{f}_i'(\vec{x}_i')$, where $\vec{f}_i(\vec{x}_i)$
  and $\vec{f}_i'(\vec{x}_i')$ are tuples of terms.  By~\emph{(ii)} and
  Lemma~\ref{l:wrap-template-sharing}, it must be $\vec{f}_i \neq
  \vec{f}_i'$. Hence, $\vec{f}_i(\vec{x}_i)$ and $\vec{f}_i'(\vec{x}_i')$ do
  not unify. Therefore, also $\sigma(\vec{v}_i)$ and $\sigma'(\vec{v}_i)$ do
  not unify. Since, by assumption, $\vecSet{x} \supseteq \vecSet{v}_i$, we
  conclude that $\sigma(\vec{x})$ does not unify with
  $\sigma'(\vec{x})$. However, this contradicts with~\emph{(i)}.  Hence,
  $m_i=m_i'$.
\end{proof}

In the next example we show how to exploit Therem~\ref{t:single-view-theorem}
for the estimation of the cardinality of an unfolding.

\begin{example}\label{e:blue-cols-madness}
  Consider a database instance $\D$ and the CQ
  $q(x,y) \leftarrow P_1(x,y), P_2(y,z), P_3(z,w)$, for which all rows in the
  ATM of $\unf(q, \wrap(\M)) = (q_{unf}(x,y), \Pi)$ are equal.  Then, by
  Theorem~\ref{t:single-view-theorem} above, $\Pi$ is of the following shape:
  \[
    \begin{array}{l l l l}
      cq_1: & q_{\unf}(f(x),g(y)) & \leftarrow
      & V_{P_1}(x,y), V^1_{P_2}(\textcolor{blue}{y},z_1), V^1_{P_3}(z_1,w_1)\\
      & & & ~~~~~~~~~\vdots\\
      cq_n: & q_{\unf}(f(x),g(y)) & \leftarrow
      & V_{P_1}(x,y), V^n_{P_2}(\textcolor{blue}{y},z_n), V^n_{P_3}(z_n,w_n)
    \end{array}
  \]
  That is, the relation $V_{P_1}$ appears in each $cq_1, \ldots, cq_n$. Due to
  this fact, together with the fact that atoms
  $V^1_{P_3}(z_1,w_1), \ldots, V^n_{P_3}(z_n,w_n)$ do not contain answer
  variables, we have the interesting property that any duplicate \emph{across
   different CQs} in the occurrences of \textcolor{blue}{$y$} in $V^j_{P_2}$
  (marked in \textcolor{blue}{blue}), for $j\in\set{1,\ldots,n}$, would produce
  a duplicate result in $q^\D$.  Under Assumption~\eqref{eq:uniformity} of
  uniformity, we get
  \[
    |q(x,\textcolor{blue}{y})^{\D}| ~=~
    \frac{|(\pi_{\textcolor{blue}{y}}(\cq_1^*) \cup \cdots \cup
     \pi_{\textcolor{blue}{y}}(\cq_n^*))^{\D}|}{\sum_{i=1}^n
     |(\pi_{\textcolor{blue}{y}}(\cq_i^*))^{\D}|}
    \cdot \sum_{j=1}^n |\cq_j^{\D}|
    \text{,}
  \]
  where $\cq_i^*$ represents the query $\cq_i$ in which we have replaced in the
  head $f(x)$ with $x$ and $g(y)$ with $y$, and the fraction
  \[
    \frac{|(\pi_{\textcolor{blue}{y}}(\cq_1^*) \cup \cdots \cup
     \pi_{\textcolor{blue}{y}}(\cq_n^*))^{\D}|}{\sum_{i=1}^n
     |(\pi_{\textcolor{blue}y}(\cq_i^*))^{\D}|}
  \]
  denotes the ratio of distinct values across the different CQs over the answer
  variable $\textcolor{blue}{y}$, calculated by assuming uniformity, i.e., that
  each value is repeated the same amount of times.
  \qedfull
\end{example}

For estimating the cardinality of the unfolding from the previous example, we
need to be able estimate the quantity
$|(\pi_{\textcolor{blue}{y}}(cq_1^*) \cup \cdots \cup
\pi_{\textcolor{blue}{y}}(cq_n^*))^{\D}|$.  For this, we can resort to
Equation~\eqref{eq:card-set-union-general} as in
Example~\ref{e:linear-equations}, by using the available statistics on
intersections between pairs of attributes sets that are arguments for some
function symbol.
We also observe that, in the case where $V_{P_2}^j = V_{P_2}^k$ for some $j$
and $k$, then
$|(\pi_{\textcolor{blue}{y}}(cq_j^*) \cup
\pi_{\textcolor{blue}{y}}(cq_k^*))^{\D}| =
\max\{|(\pi_{\textcolor{blue}{y}}(cq_j^*))^{\D}|,
|(\pi_{\textcolor{blue}{y}}(cq_k^*))^{\D}|\}$, and the computation of the
cardinality of the union can be simplified.

\begin{example}\label{e:blue-cols-madness-2}
  Recall the scenario from Example~\ref{e:blue-cols-madness}. Assume that
  $n = 6$, and that
  \[
    \begin{array}{l}
      V_{P_2}^1 = V_{P_2}^2\\
      V_{P_2}^3 = V_{P_2}^4\\
      V_{P_2}^5 = V_{P_2}^6\\
    \end{array}
  \]
  Then, let
  \[
    \begin{array}{l}
      |(\pi_{\textcolor{blue}{y}}(cq_1^*))^\D| ~=~
      \max\{|(\pi_{\textcolor{blue}{y}}(cq_1^*))^{\D}|,
      |(\pi_{\textcolor{blue}{y}}(cq_2^*))^{\D}|\} ~=~ |(\pi_{\textcolor{blue}{y}}(cq_1^*) \cup
      \pi_{\textcolor{blue}{y}}(cq_2^*))^{\D}| \\
      |(\pi_{\textcolor{blue}{y}}(cq_3^*))^\D| ~=~
      \max\{|(\pi_{\textcolor{blue}{y}}(cq_3^*))^{\D}|,
      |(\pi_{\textcolor{blue}{y}}(cq_4^*))^{\D}|\}  ~=~ |(\pi_{\textcolor{blue}{y}}(cq_3^*) \cup
      \pi_{\textcolor{blue}{y}}(cq_4^*))^{\D}| \\
      |(\pi_{\textcolor{blue}{y}}(cq_5^*))^\D| ~=~
      \max\{|(\pi_{\textcolor{blue}{y}}(cq_5^*))^{\D}|,
      |(\pi_{\textcolor{blue}{y}}(cq_6^*))^{\D}|\} ~=~ |(\pi_{\textcolor{blue}{y}}(cq_5^*) \cup
      \pi_{\textcolor{blue}{y}}(cq_6^*))^{\D}| .
    \end{array}
  \]
  Hence, 
  \[
  |(\pi_{\textcolor{blue}{y}}(cq_1^*) \cup \cdots \cup \pi_{\textcolor{blue}{y}}(cq_6^*))^{\D}| = |(\pi_{\textcolor{blue}{y}}(cq_1^*) \cup \pi_{\textcolor{blue}{y}}(cq_3^*) \cup \pi_{\textcolor{blue}{y}}(cq_5^*))^\D|.
  \]
  By using Equation~\eqref{eq:card-set-union-general}, we get the formula:
  \begin{equation}\label{eq:c:planning:est-union-of-proj}
    |(\pi_{\textcolor{blue}{y}}(cq_1^*) \cup \pi_{\textcolor{blue}{y}}(cq_3^*) \cup \pi_{\textcolor{blue}{y}}(cq_5^*))^\D| = \\
    \begin{array}[t]{l}
      |(\pi_{\textcolor{blue}{y}}(cq_1^*))^{\D}| + |(\pi_{\textcolor{blue}{y}}(cq_3^*))^{\D}| + |(\pi_{\textcolor{blue}{y}}(cq_5^*))^{\D}| \\
      {}- |(\pi_{\textcolor{blue}{y}}(cq_1^*) \cap \pi_{\textcolor{blue}{y}}(cq_3^*))^{\D}| - |(\pi_{\textcolor{blue}{y}}(cq_1^*) \cap \pi_{\textcolor{blue}{y}}(cq_5^*))^{\D}|
      - |(\pi_{\textcolor{blue}{y}}(cq_3^*) \cap \pi_{\textcolor{blue}{y}}(cq_5^*))^{\D}|\\
      {}+ |(\pi_{\textcolor{blue}{y}}(cq_1^*) \cap \pi_{\textcolor{blue}{y}}(cq_3^*) \cap \pi_{\textcolor{blue}{y}}(cq_5^*))^{\D}|.
    \end{array}
  \end{equation}
    Assume that 
    \[
    \max\{|(\pi_{\textcolor{blue}{y}}(cq_1^*) \cap \pi_{\textcolor{blue}{y}}(cq_3^*))^{\D}|, |(\pi_{\textcolor{blue}{y}}(cq_1^*) \cap \pi_{\textcolor{blue}{y}}(cq_5^*))^{\D}|, |(\pi_{\textcolor{blue}{y}}(cq_3^*) \cap \pi_{\textcolor{blue}{y}}(cq_5^*))^{\D}|\} = |(\pi_{\textcolor{blue}{y}}(cq_1^*) \cap \pi_{\textcolor{blue}{y}}(cq_3^*))^{\D}|.
    \]
    Then, by relying on Assumption~\eqref{eq:a4}, Equation~\eqref{eq:c:planning:est-union-of-proj} can be rewritten as:
    \[
    |(\pi_{\textcolor{blue}{y}}(cq_1^*) \cup \pi_{\textcolor{blue}{y}}(cq_3^*) \cup \pi_{\textcolor{blue}{y}}(cq_5^*))^\D| = \\
    \begin{array}[t]{l}
      |(\pi_{\textcolor{blue}{y}}(cq_1^*))^{\D}| + |(\pi_{\textcolor{blue}{y}}(cq_3^*))^{\D}| + |(\pi_{\textcolor{blue}{y}}(cq_5^*))^{\D}| \\
      {}- 2 \cdot |(\pi_{\textcolor{blue}{y}}(cq_1^*) \cap \pi_{\textcolor{blue}{y}}(cq_3^*))^{\D}| - |(\pi_{\textcolor{blue}{y}}(cq_1^*) \cap \pi_{\textcolor{blue}{y}}(cq_5^*))^{\D}|
      - |(\pi_{\textcolor{blue}{y}}(cq_3^*) \cap \pi_{\textcolor{blue}{y}}(cq_5^*))^{\D}|\text{.}\\
    \end{array}
    \]
    We use such formula, containing only binary intersections, to compute the ratio of distinct values
    \[
    r ~=~ \frac{|(\pi_{\textcolor{blue}{y}}(cq_1^*) \cup \pi_{\textcolor{blue}{y}}(cq_3^*) \cup \pi_{\textcolor{blue}{y}}(cq_5^*))^{\D}|}{|(\pi_{\textcolor{blue}{y}}(cq_1^*))^{\D}| +
      |(\pi_{\textcolor{blue}{y}}(cq_3^*))^{\D}| + |(\pi_{\textcolor{blue}{y}}(cq_5^*))^{\D}|}\text{.}
    \]
    Hence, we estimate
    $|q(x,\textcolor{blue}{y})^\D|$ to be
    \begin{equation}\label{eq:c:planning:final-eq-planning}
      r \cdot \sum_{j=1}^6 |cq_j^{\D}|
    \end{equation}
\qedfull
\end{example}

The previous example makes use of Equation~\eqref{eq:card-set-union-general}, which is practical only if the union is on a small number of sets. Additionally, for the sake of clarity, we have used the actual cardinalities, rather than estimates. We observe, however, that all the cardinalities in Formula~\eqref{eq:c:planning:final-eq-planning} can be estimated by using our estimators $\cardEst$, $\fv$, and $\dist_{\D}$. % An alternative way is to store additional statistics over the projection on one of the two arguments of each property in the wrap of the T-mapping, and to use this value and the uniformity assumption to estimate $|W_1 \cup W_2 \cup W_3|$.

\section{Unfolding Cost Model}\label{c:planning:s:unf-cost-model}
\label{c:planning:s:unf-cost-model}

We are now ready to estimate the actual costs of evaluating UJUCQ and UCQ
unfoldings, by exploiting the cardinality estimations from the previous
section. Our cost model is based on traditional textbook-formulae for query
cost estimation~\cite{SiKS05}, and it assumes source parts in the mappings
to be CQs.

\subsection{Cost for the Unfolding of a UCQ.} Recall from Section~\ref{c:planning:s:cover-based-trans} that the unfolding of a UCQ produces a UCQ translation $q^{\ucq} = \bigvee_i q_i^{\cq}$. We estimate the cost of evaluating $q^{\ucq}$ as
\cmath{
 \label{eq:ucq-cost}
% \textstyle
 c_{\eval}(q^{\ucq}) \enskip = \enskip \sum_i c_{\eval}(q_i^{\cq}) + c_{u}(q^{\ucq})}%
where
\begin{compactitem}
\item $c_{\eval}(q_i^{\cq})$ is the cost of evaluating each $q_i^{\cq}$ in $q^{\ucq}$;
\item $c_{u}(q^{\ucq})$ is the cost of removing duplicate results.
\end{compactitem}

In the remainder of this paragraph we explain how we calculate each of these addends.

\myPar{Value of ${c_{\eval}(q^{\cq})}$.} Let $q^{\cq} = V_1 \Join \cdots \Join V_n$, where $V_i = T_{i1} \Join \cdots \Join T_{i{m_i}}$, where $m_i \in \mathbb{N}$ and $1 \le i \le n$.
\noindent Then, we define $c_{\eval}(q^{\cq})$ as
\begin{align*}
  c_{\eval}(q^{\cq}) = \sum_{1 \le i \le n}\sum_{1 \le j \le m_i} c_{\scan}(T_{ij}) + c_{\hjoin}(q^{\cq})
\end{align*}
where
\begin{itemize}
\item $c_{\scan}(T_{ij}) = \distSet{T_{ij}^\D} \times c_t$
  \begin{itemize}
  \item where $c_t$ is the fixed cost of retrieving one tuple from the database. Such constant, as well as other constants in this section, can be
    found through \emph{calibration techniques}~\cite{DBLP:conf/vldb/GardarinST96}.
  \end{itemize}
\item $c_{\hjoin}(q_i^{\cq}) = \left(\sum_{i=1}^n m_i\right) \cdot \cardEst(q_i^{\cq}) \cdot c_j$
  \begin{itemize}
  \item where $c_j$ is the fixed cost of joining one tuple.
  \end{itemize}
\end{itemize}

Summing up, $c_{\eval}(q_i^{\cq})$ is the cost of scanning each table in the conjunctive query and performing a \emph{hash join}.
We assume the statistics $\distSet{T_{ij}^\D}$, for $1 \le i \le n$, $1 \le j \le m_i$, to be available after having analyzed the database instance according to the database schema.

\myPar{Value of ${c_{u}(q^{\ucq})}$.}
We assume the removal of duplicates at this level to be carried out by a sort-based strategy. Under this assumption, the cost of removal is
\[c_{u}(q^{\ucq}) = \cardEst(q^{ucq}) \cdot log(\cardEst(q^{\ucq})) \cdot c_u\]
where $\cardEst(q^{\ucq})$ is the cardinality estimation of the unfolding,
calculated as described in Section~\ref{c:planning:s:unf-est}, and $c_u$ is a constant denoting the cost of eliminating a duplicate tuple.

\myPar{Cost for the Unfolding of a  JUCQ.} Recall from Section~\ref{c:planning:s:cover-based-trans} that the optimized unfolding of a JUCQ produces a UJUCQ. We estimate the cost of a single JUCQ $q^{jucq} = \bigwedge_{i} q_{i}^{\ucq}$ in the unfolding as
 \cmath{
 \label{eq:jucq-cost}
 \textstyle
 c(q^{\jucq}) \enskip=\enskip \sum_i c_{\eval}(q_i^{\ucq}) + \sum_{i \neq k}
 c_{\mat}(q_i^{\ucq}) + c_{\mjoin}(q^{\jucq})}%
where
\begin{compactitem}
\item $c_{\eval}(q_i^{\ucq})$ is the cost of evaluating each UCQ component
  $q_i^{\ucq}$;
  % in the input JUCQ;
\item $\sum_{i \neq k}c_{\mat}(q_i^{\ucq})$ is the cost of materializing the intermediate
  results from $q_i^{\ucq}$, where the $k$-th UCQ is assumed to be \emph{pipelined}~\cite{SiKS05} and not materialized; % from a UCQ component in the input JUCQ;
\item $c_{\mjoin}(q^{\jucq})$ is the cost of a merge join over the materialized intermediate results.
\end{compactitem}

The cost for a UJUCQ $q^{\ujucq}=\bigvee_i q^{jucq}_i$, if all the attributes are kept in the answer, is simply the sum $\sum_i c(q^{\jucq}_i)$, since the results of all JUCQs are disjoint (c.f., Section~\ref{c:planning:s:cover-based-trans}). Otherwise, we need to consider the cost of eliminating duplicate results.

\myPar{Value of ${c_{\mat}(q_i^{\ucq})}$.} Let $q_i^{\ucq} = q_1^{\cq} \Join \cdots \Join q_m^{\cq}$. Then
\[c_{\mat}(q_i)^{\ucq} = \cardEst(q_i^{\ucq}) \cdot c_m\]
\noindent where $c_m$ is the fixed cost of materializing a tuple.

\myPar{Value of ${c_{\mjoin}(q^{\jucq})}$.} Let $q^{\jucq} = q_1^{\ucq} \Join \cdots \Join q_m^{\ucq}$. Then
\[c_{\mjoin}(q^{\jucq}) = m \cdot \cardEst(q^{\jucq}) \cdot c_j \]
\noindent where $\cardEst(q^{\jucq})$ is the cardinality estimation of the unfolding, calculated as described in Section~\ref{c:planning:s:unf-est}. Observe that
we do not consider the cost of sorting each $q_i^{\ucq}$, as this cost is already included in $c_{\eval}(q_i^{\ucq})$.

\section{Experimental Results}\label{c:planning:s:eval}

In this section, we provide an empirical evaluation that compares unfoldings of UCQs and (optimized) unfoldings of JUCQs, as well as the estimated costs and the actual time needed to evaluate the unfoldings.
We ran the experiments on an HP Proliant server with
%24 Intel Xeon CPUs (24 cores @3.47GHz),
2 Intel Xeon X5690 Processors (each with 12 logical cores at 3.47GHz),
 106GB of RAM and five 1TB 15K RPM HDs.
%
%The OS is Ubuntu 16.04 64-bit edition.
 As RDBMS we have used PostgreSQL 9.6. The material to replicate our  experiments is available  online~\footnote{\url{https://github.com/ontop/ontop-examples/tree/master/iswc-2017-cost}}.

\begin{figure}[t]
  \scriptsize
  \centering
\begin{subfigure}[b]{0.49\textwidth}
  \includegraphics[width=\textwidth]{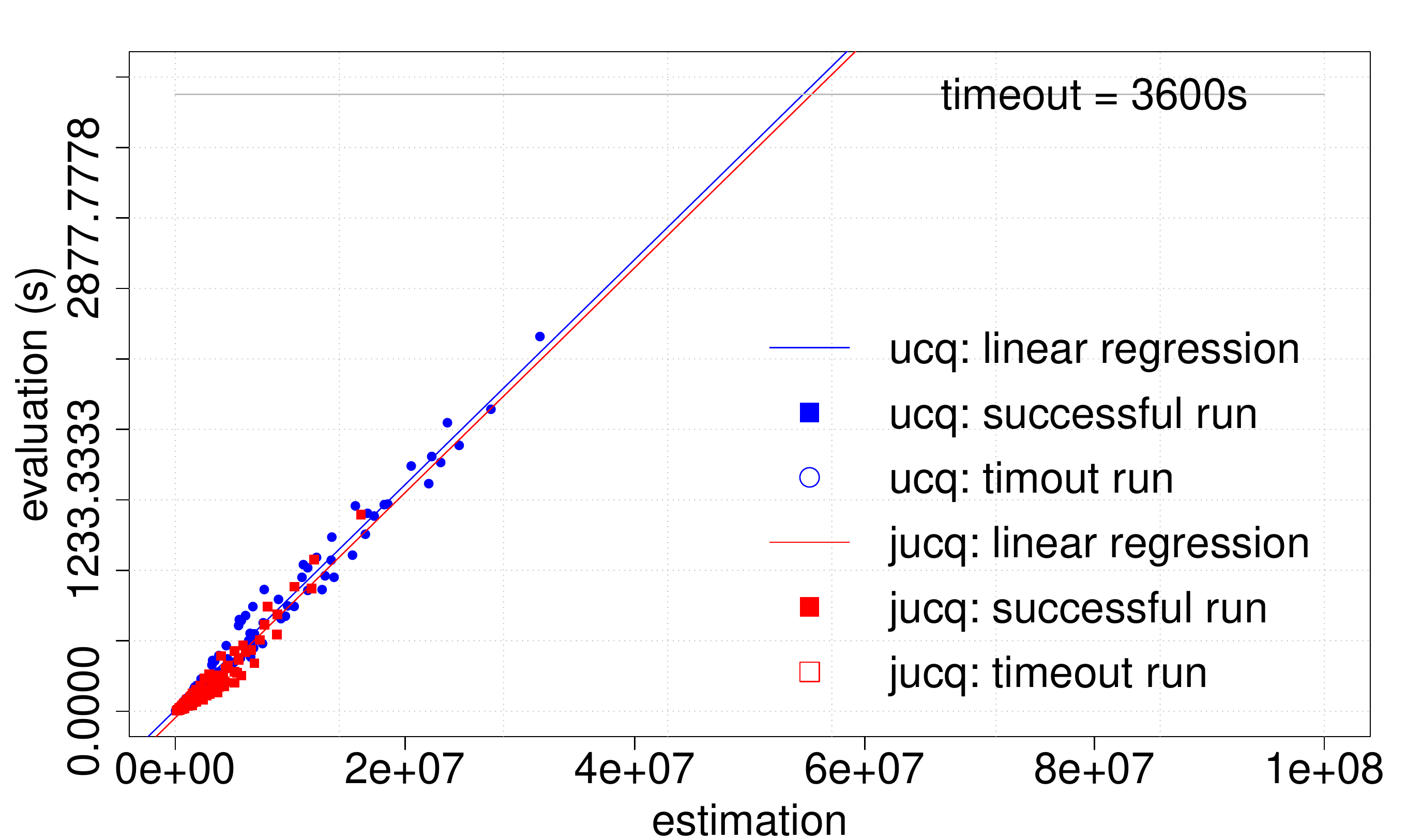}
\caption{Our cost model vs. evaluation}
\label{fig:est-vs-eval:ours}
\end{subfigure}
\begin{subfigure}[b]{0.49\textwidth}
\includegraphics[width=\textwidth]{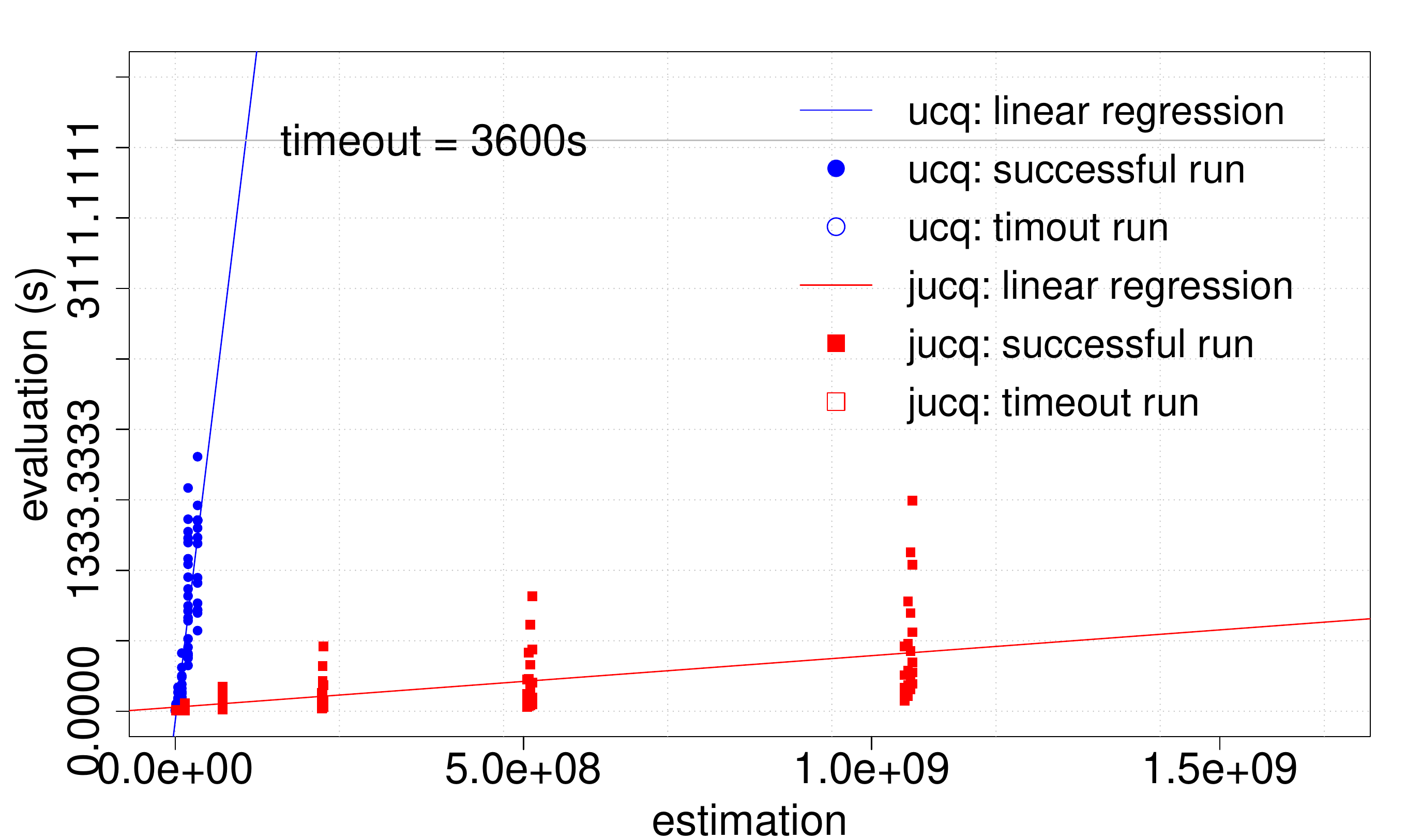}
\caption{PostgreSQL cost model vs. evaluation}
\label{fig:est-vs-eval:postgres}
\end{subfigure}

\begin{subfigure}[b]{0.9\textwidth}
  \begin{tabular}{ccc@{\extracolsep{4pt}}c@{~}c@{~}c@{~}c@{~}c@{~}c@{~}c@{~}c}
    \toprule
    & \multicolumn{2}{c}{queries} & \multicolumn{2}{c}{linear regression of our cost} & \multicolumn{2}{c}{linear regression of PostgreSQL cost} \\ \cline{2-3}\cline{4-5}\cline{6-7}
    & succ. & time out & $b_0$ & $b_1$  &  $b_0$ & $b_1$ \\ \midrule
    UCQ & 68 & 16 & 1.40e+01 & 6.17e-05  & 6.16e+01 & 3.61e+05 \\
    JUCQ & 83 & 1 & -3.76e+01 & 6.33e-05  & 2.56e+01 & 3.25e-07\\
    \bottomrule
  \end{tabular}
  \caption{Results of linear regressions  ($ \mathit{evaluation} = b_0 + b_1 \times \mathit{estimation} $)}
\label{fig:est-vs-eval:lm}
\end{subfigure}
  \caption{Cost estimations vs evaluation running times}
  \label{fig:est-vs-eval}
\end{figure}

\subsection{Wisconsin Experiment}

We devised an OBDA benchmark based on the \emph{Wisconsin Benchmark}~\cite{DeWitt93}. In the
following subsections we discuss each element of the benchmark, and their rationale.

\paragraph{Mappings.}

The mappings set consists of $1364$ mapping assertions.
Each mapping defines a property of the form \textowl{:J20O$xx$M$y$R$z$Prop$i$}, where

\begin{itemize}
\item \textowl{J20}, read ``join selectivity $20\%$'', denotes that each mapping for that property is on a query retrieving the $20\%$ of the total tuples;
\item O$xx$, read ``offset $xx$'', where $xx \in \set{0,5,10,15}$, denotes the offset for the filter in the mappings based on column \textsql{onePercent};
\item M$y$, $y \in \set{1, \ldots, 6}$ denotes the number of mapping assertions defining the property;
\item R$z$, $z \le y$ denotes the number of redundant mapping assertions, that is, the number of mapping assertions whose source query does not produces new individuals for the property in the virtual RDF ABox;
\item Prop$i$, $i \in \set{1,\ldots,3}$. The tested BGP is made of three properties.
\end{itemize}

For instance, we the following listing provides the mapping definition for the property \textowl{:S20O0M2R0Prop1}.

\noindent
\begin{minipage}{1.0\linewidth}
\begin{lstlisting}
target       :number/{unique2} :S20O0M2R0Prop1
                :name/{evenOnePercent}/{stringu1}/{stringu2} .
source       SELECT "unique2", evenOnePercent, stringu1, stringu2
             FROM t1_1m
             WHERE "onepercent" >= 0 AND "onepercent" < 20
\end{lstlisting}
\end{minipage}

\paragraph{SPARQL Queries.}

Our test is on $84$ queries, instantiations of the following template:

\noindent \begin{minipage}{\linewidth}
\begin{lstlisting}[mathescape]
SELECT DISTINCT * WHERE {
   ?x :M${m}$R${r}$Prop1 ?y1; :J$j$M$m$R$r$Prop2 ?y2; :J$j$M$m$R$r$Prop3 ?y3
}
\end{lstlisting}

\end{minipage}
where $j \in \set{5,10,15,20}$ denotes the selectivity of the join between the first property and each of the remaining two, expressed as a percentage of the number of retrieved rows for each mapping defining the property (each mapping retrieves $200$k tuples); $m \in \set{1, \ldots ,6}$ denotes the number of mappings
defining the property (all such mappings have the same signature), and
$r \in \{0, \ldots ,m-1\}$ denotes the number of
\emph{redundant} mappings, that is, the number of mappings assertions
retrieving the same results of another mapping definining the
property, minus one.
\ignore{Queries are numbered as $q_{n}$, where $m \cdot n$. Hence,  query $q_3$ is the query $q_{m=2 \cdot 1=r}$, with $2$ mappings one of which is redundant.}

For each query, we have tested a correspondent cover query of two fragments $f_1, f_2$, where each fragment is an instantiation of the following templates:

\noindent \begin{minipage}{1.0\linewidth}
\begin{lstlisting}[mathescape]
$f_1$: SELECT DISTINCT ?x ?y1 ?y2 WHERE { ?x :M$m$R$r$Prop1 ?y1; :J$j$M$m$R$r$Prop2 ?y2. }
$f_2$: SELECT DISTINCT ?x ?y3 WHERE { ?x :M$m$R$r$Prop1 ?y2; ?x :J$j$M$m$R$r$Prop3 ?y3. }
\end{lstlisting}
\end{minipage}

\paragraph{Ontology.}

We have carried out our tests over an empty TBox. Observe that this is not a limitation, as the case with a TBox can be reduced to the case without a TBox by making use of a T-mapping as in Example~\ref{e:long-example}.

\paragraph{Data.}

We have created several copies of the wisconsin table, and populated each copy with ten million tuples.

\paragraph{Evaluation.}
\label{sec:exper-eval}

In Figure~\ref{fig:est-vs-eval}, we present the cost estimation and
the actual running time for each query.
%We use red squares for JUCQs and blue dots for UCQs.
%We denote the successful runs with filled symbols and the timeout ones with non-filled ones.
We have the following observations:

\begin{compactitem}
\item In this experiment, for the considered cover, JUCQs are generally faster
  than UCQs. In fact, out of the 84 SPARQL queries, only one JUCQ was timed
  out, while 16 UCQs were timed out. The mean running time of
  successful UCQs and JUCQs are respectively $160$ seconds and $350$ seconds.

\item In Figure~\ref{fig:est-vs-eval:ours}, where the fitted lines are obtained by applying linear regression over successful UCQ and JUCQ evaluations, we observe a strong linear
  correlation between our estimated costs and real running
  times. Moreover, the coefficients ($b_1$ and $b_0$) for UCQs and JUCQs
  are rather close. This empirically shows that our cost model can
  estimate the real running time well.

\item Figure~\ref{fig:est-vs-eval:postgres} shows that the PostgreSQL
  cost model assigns the same estimation to many queries having different running times. Moreover, the linear regressions for UCQs and JUCQs are rather
  different, which suggests that PostgreSQL is not able to recognize when two translations are semantically equivalent. Hence, PostgreSQL is not able to estimate the cost of these queries properly.

\end{compactitem}

\begin{figure}[tb]
  \centering
  \includegraphics[width=.75\textwidth]{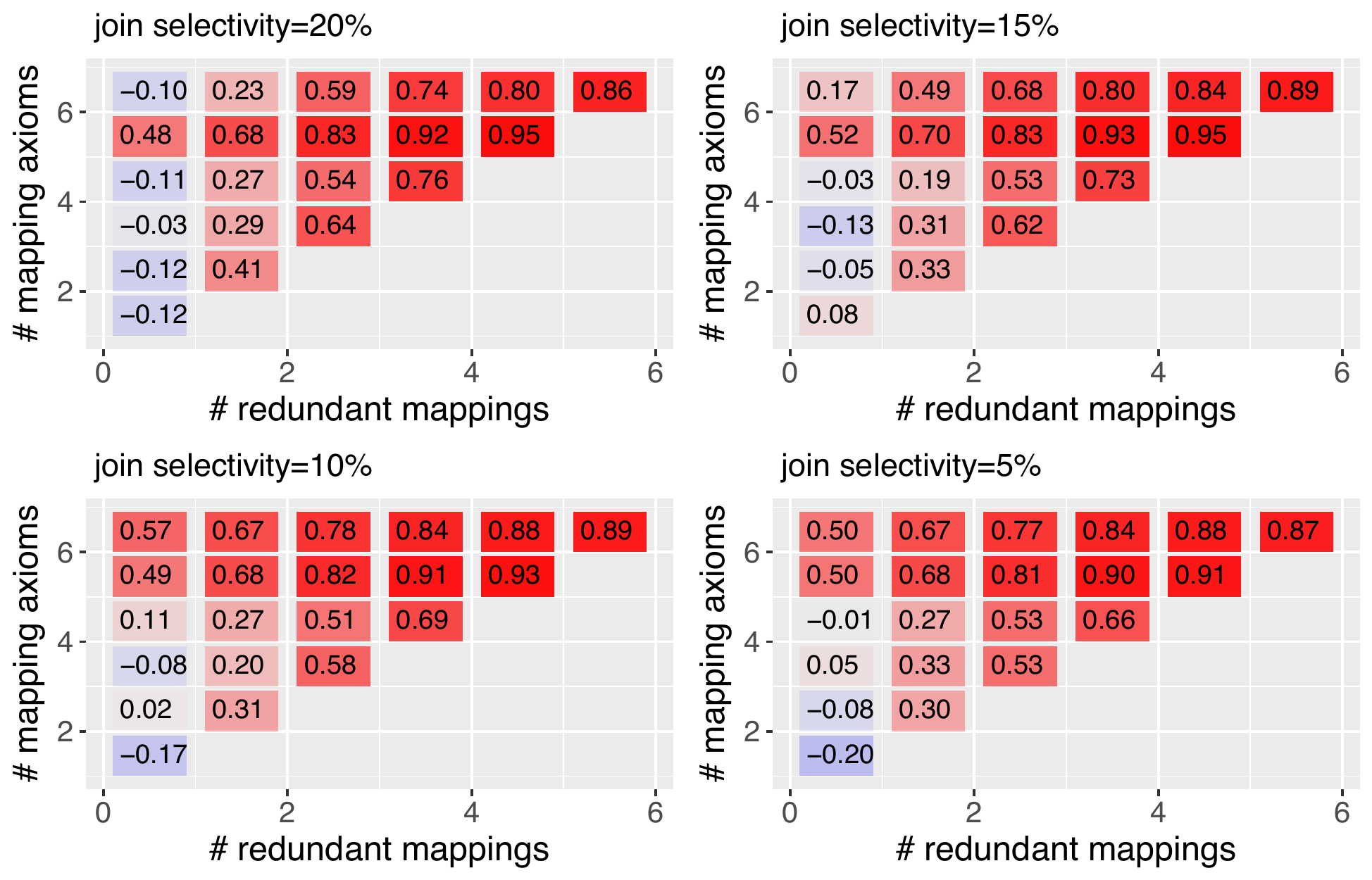}
  \caption{Performance gain of JUCQ compared with UCQ}
  \label{fig:ucq-vs-jucq}
\end{figure}

In Figure~\ref{fig:ucq-vs-jucq}, we visualize the performance gain of
JUCQs compared with UCQs. The four subgraphs correspond to four
different settings join selectivities. Each subgraph is a matrix in
which each cell shows the value of the performance gain
$g = 1 - \text{jucq\_time}/\text{ucq\_time}$. When $g > 0$, we apply
the red color; otherwise blue.
% We also use the absolute values of $g$
% as the alpha values of each cell.
These graphs clearly show that when there is a large number of mappings
and there is high redundancy, we have better performance gains.
When the redundancy is low (0 or 1), and the number of mapping axioms
is large, the join selectivity plays an important role in the
performance gain, as discussed in~\cite{BursztynGM15b}; in other
cases, the impacts are non-significant.

Figures~\ref{fig:postgres-ucq-est-vs-real} and
\ref{fig:postgres-jucq-est-vs-real} report the cardinalities estimated by PostgreSQL divided by the actual sizes of the query answers for all
UCQ and JUCQ queries. For UCQs, it shows that PostgreSQL normally
underestimates the cardinalities, but it overestimates them when the
redundancies are high. As for JUCQS, PostgreSQL always overestimates
the cardinalities, ranging from 40 to 200K times. These numbers partially
explains why PostgreSQL estimate the costs of both UCQs and JUCQs so badly in
Figure~\ref{fig:est-vs-eval:postgres}.

We obtained similar conclusions for a query with four atoms, and a
cover of three fragments.
%For more details, refer to the extended version~\cite{tech} of this work.

\begin{figure}[t]
  \centering
  \includegraphics[width=.9\textwidth]{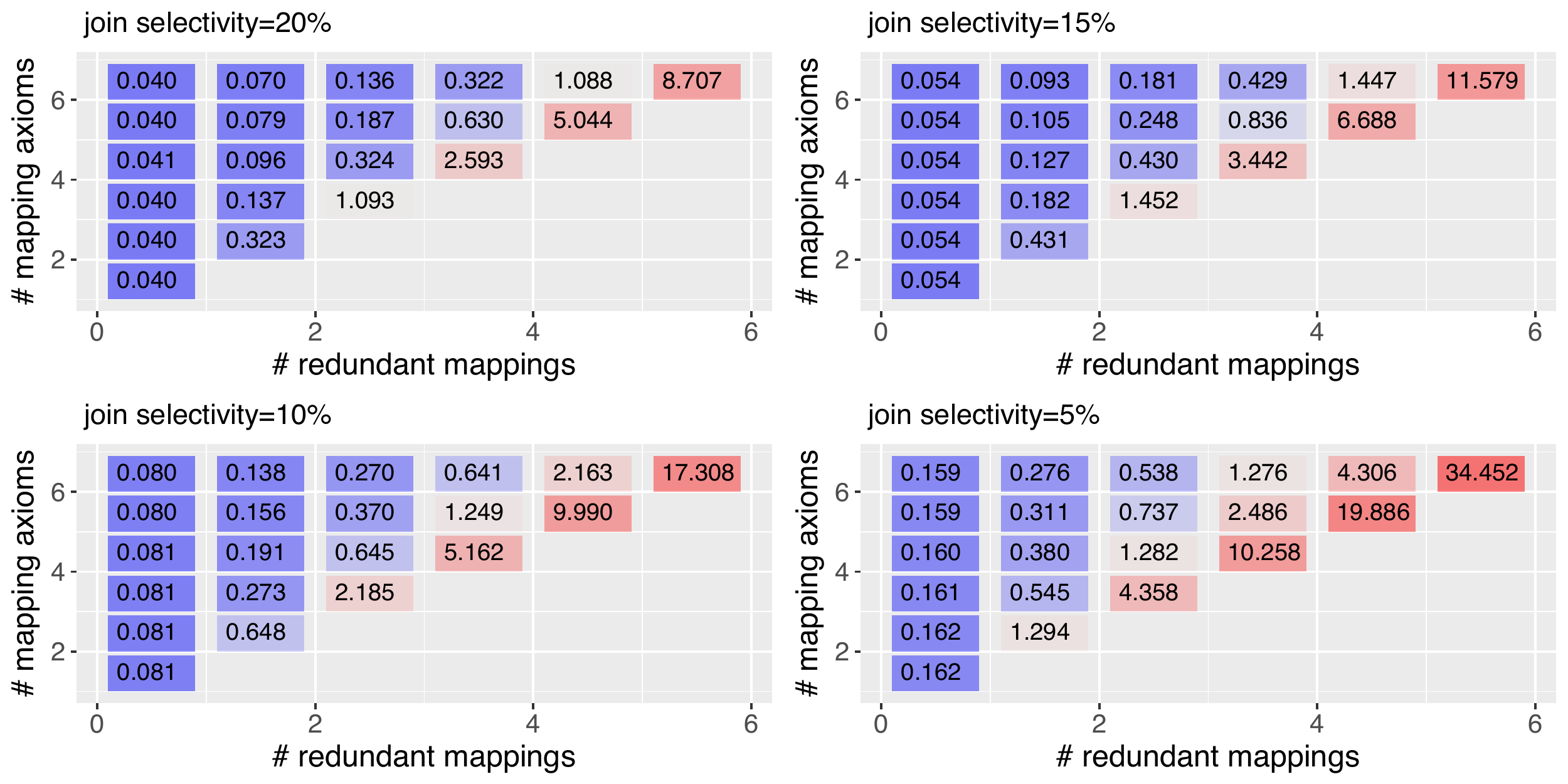}
  \caption{UCQs: (PostgreSQL estimated cardinality) / (real cardinality)}
  \label{fig:postgres-ucq-est-vs-real}
% \end{figure}

% \begin{figure}[t]
  \centering
  \includegraphics[width=.9\textwidth]{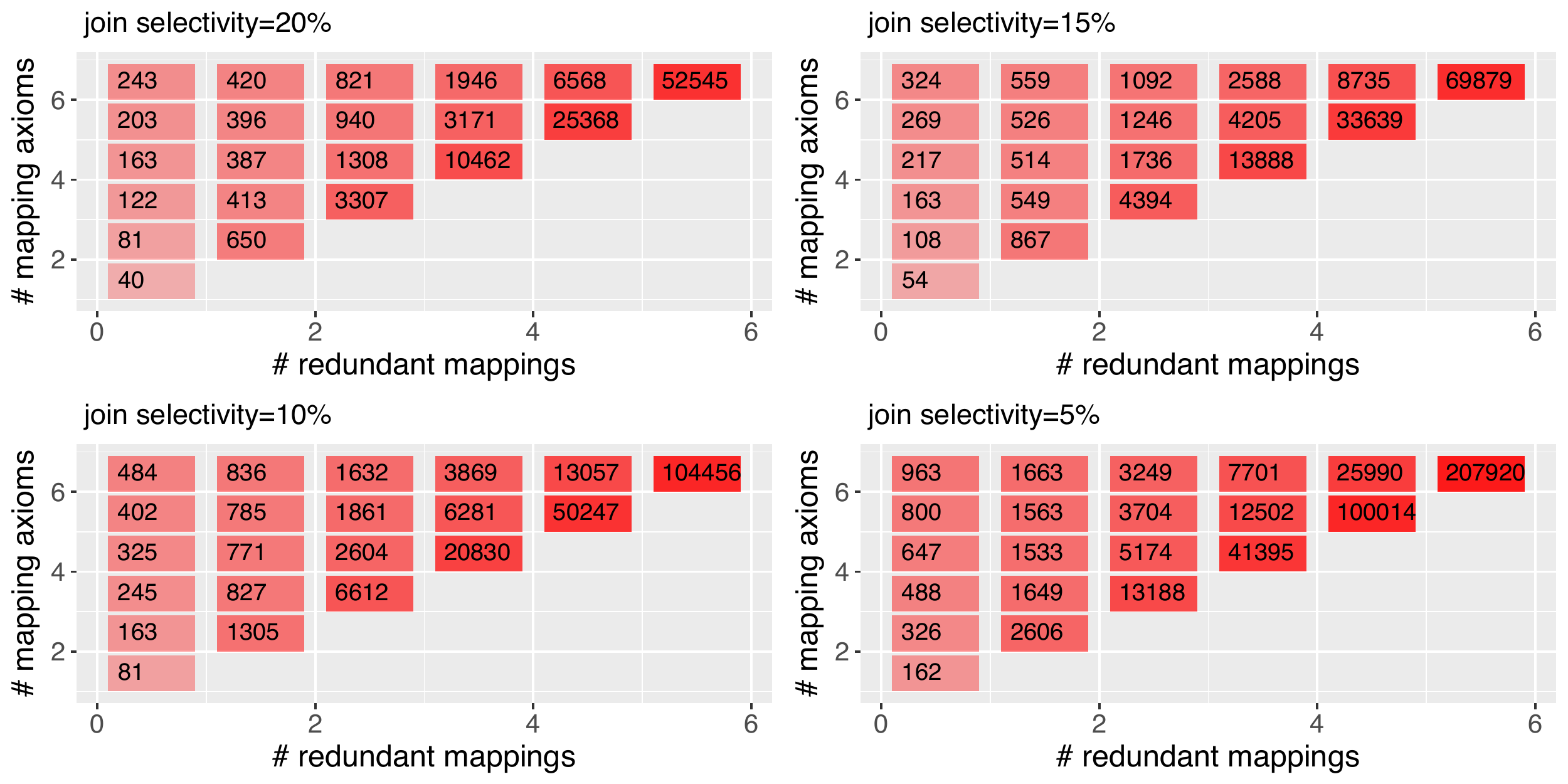}
  \caption{JUCQs: (PostgreSQL estimated cardinality) / (real cardinality)}
  \label{fig:postgres-jucq-est-vs-real}
\end{figure}

\FloatBarrier

\subsubsection{Wisconsin Experiment II: Four Atoms}

In this experiment we consider 84 queries, instances of the following template

\noindent \begin{minipage}{\linewidth}
\begin{lstlisting}[mathescape]
SELECT DISTINCT * WHERE {?x :M${m}$R${r}$Prop1 ?y1; :J$j$M$m$R$r$Prop2 ?y2; :J$j$M$m$R$r$Prop3 ?y3; :J$j$M$m$R$r$Prop3 ?y4}
\end{lstlisting}
\end{minipage}

where $j \in$, $m$, and $r$ are defined as for the case with three atoms.

For each query, we have tested a correspondent cover query of three fragments $f_1, f_2, f_3$, where each fragment is an instantiation of the following templates:

\noindent \begin{minipage}{1.0\linewidth}
\begin{lstlisting}[mathescape]
$f_1$: SELECT DISTINCT ?x ?y1 ?y2 WHERE { ?x :M$m$R$r$Prop1 ?y1; :J$j$M$m$R$r$Prop2 ?y2. }
$f_2$: SELECT DISTINCT ?x ?y3 WHERE { ?x a :M$m$R$r$Prop1; ?x :J$j$M$m$R$r$Prop3 ?y3. }
$f_3$: SELECT DISTINCT ?x ?y4 WHERE { ?x a :M$m$R$r$Prop1; ?x :J$j$M$m$R$r$Prop3 ?y4. }
\end{lstlisting}
\end{minipage}

\paragraph{Evaluation.}

Figures~\ref{fig:est-vs-eval-4},~\ref{fig:postgres-ucq-est-vs-real-4}, and~\ref{fig:postgres-jucq-est-vs-real-4} are the counterparts for the Figures of the three atoms case. Observe that the results look extremely similar, thus confirming our comments from the previous Subsection.

\begin{figure}[t]
  \scriptsize
  \centering
\begin{subfigure}[b]{0.49\textwidth}
  \includegraphics[width=\textwidth]{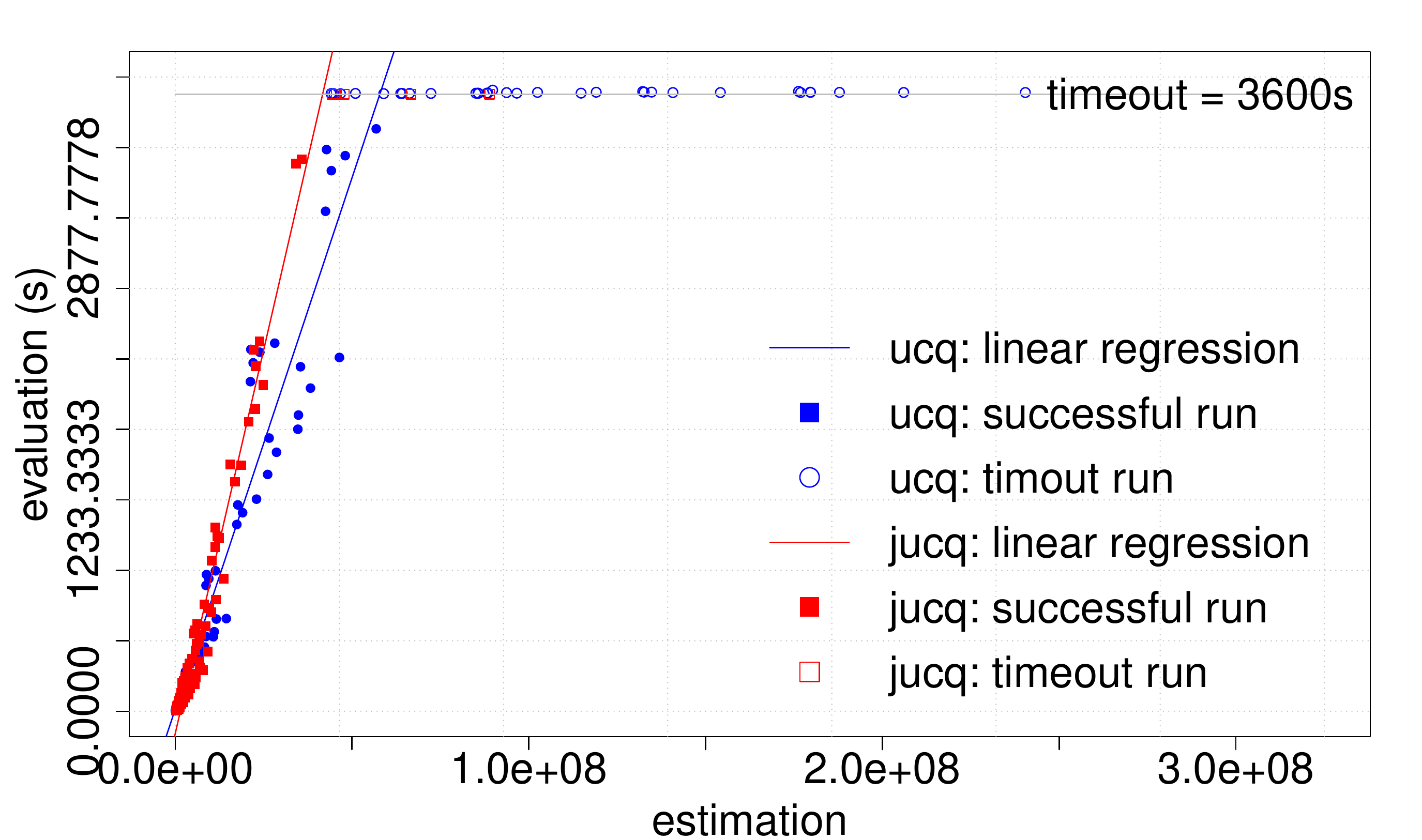}
\caption{Our cost model vs. evaluation}
\label{fig:est-vs-eval:ours-4}
\end{subfigure}
\begin{subfigure}[b]{0.49\textwidth}
\includegraphics[width=\textwidth]{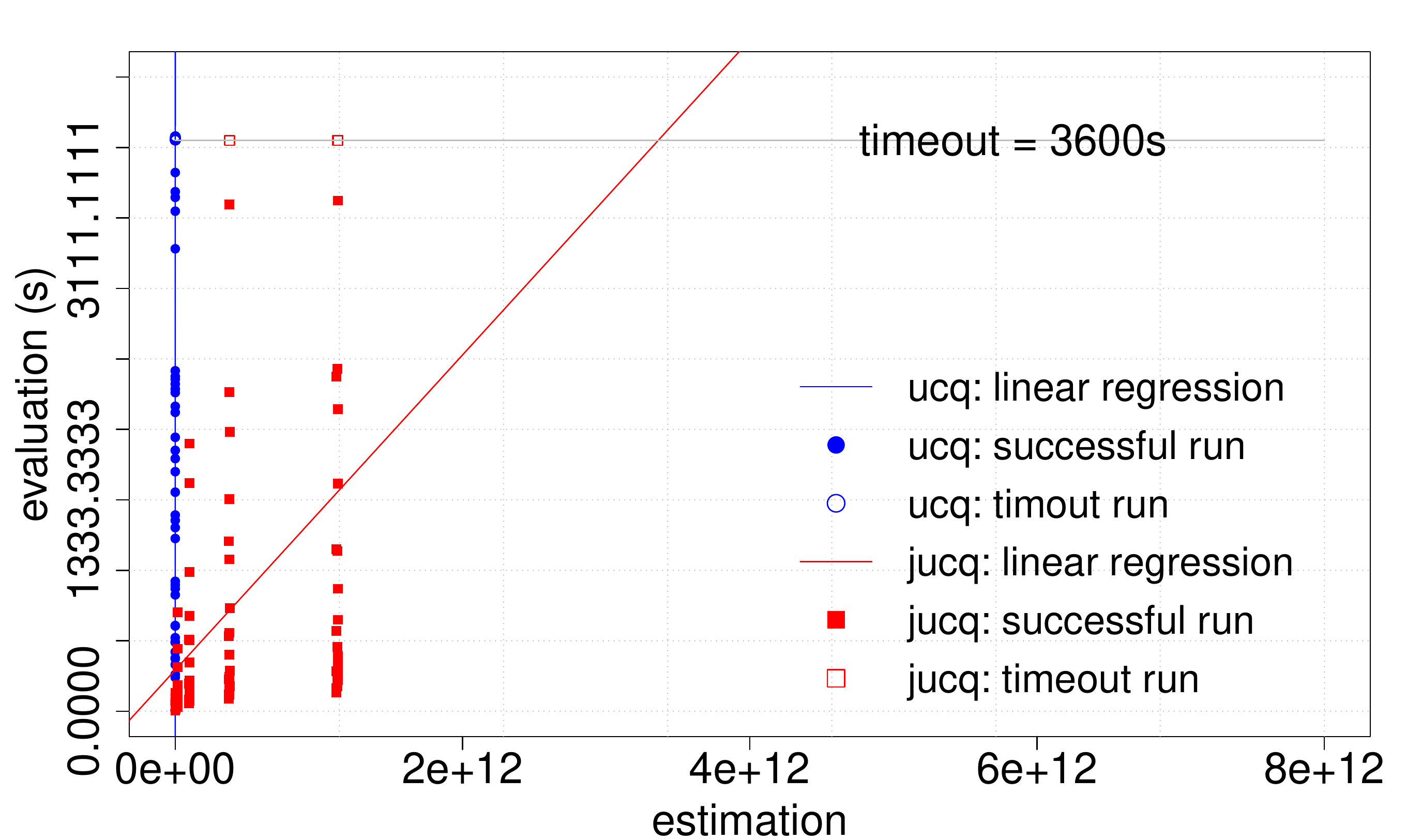}
\caption{PostgreSQL cost model vs. evaluation}
\label{fig:est-vs-eval:postgres-4}
\end{subfigure}

\begin{subfigure}[b]{0.9\textwidth}
  \begin{tabular}{ccc@{\extracolsep{4pt}}c@{~}c@{~}c@{~}c@{~}c@{~}c@{~}c@{~}c}
    \toprule
    & \multicolumn{2}{c}{queries} & \multicolumn{2}{c}{linear regression of our cost} & \multicolumn{2}{c}{linear regression of PostgreSQL cost} \\ \cline{2-3}\cline{4-5}\cline{6-7}
    & succ. & time out & $b_0$ & $b_1$  &  $b_0$ & $b_1$ \\ \midrule
    UCQ & 50 & 34 & 1.22e+01 & 6.20e-05 & 8.78e+01 & 2.21e-05 \\
    JUCQ & 79 & 5 & -1.30+e02 & 8.95e-05 & 2.74e+02 & 5.84e-10 \\
    \bottomrule
  \end{tabular}
  \caption{Results of linear regressions  ($ \mathit{evaluation} = b_0 + b_1 \times \mathit{estimation} $)}
  \label{fig:est-vs-eval:lm-4}
\end{subfigure}
\caption{Cost estimations vs evaluation running times}
\label{fig:est-vs-eval-4}
\end{figure}

\begin{figure}[t]
  \centering
  \includegraphics[width=.9\textwidth]{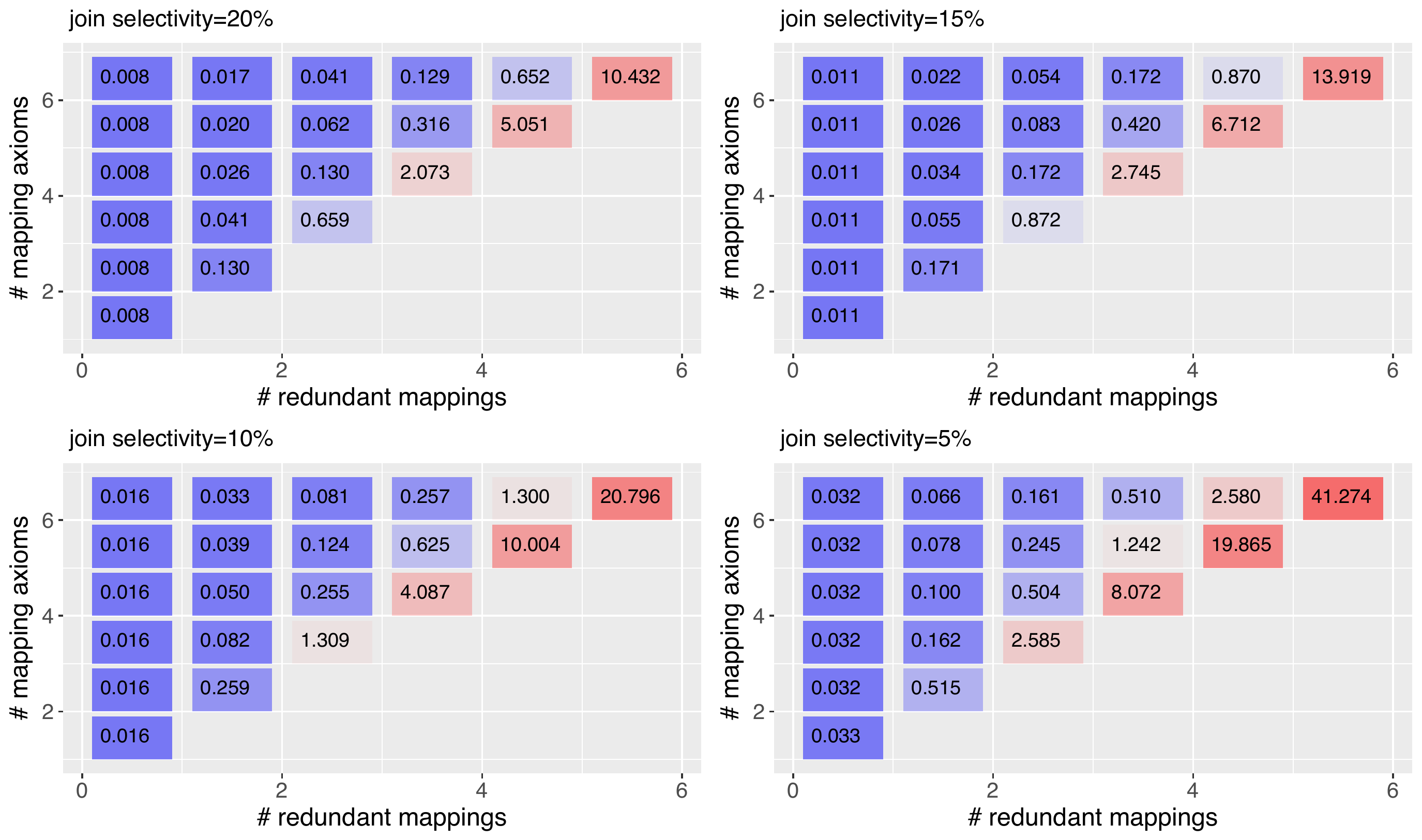}
  \caption{UCQs: (PostgreSQL estimated cardinality) / (real cardinality)}
  \label{fig:postgres-ucq-est-vs-real-4}
% \end{figure}

% \begin{figure}[t]
  \centering
  \includegraphics[width=.9\textwidth]{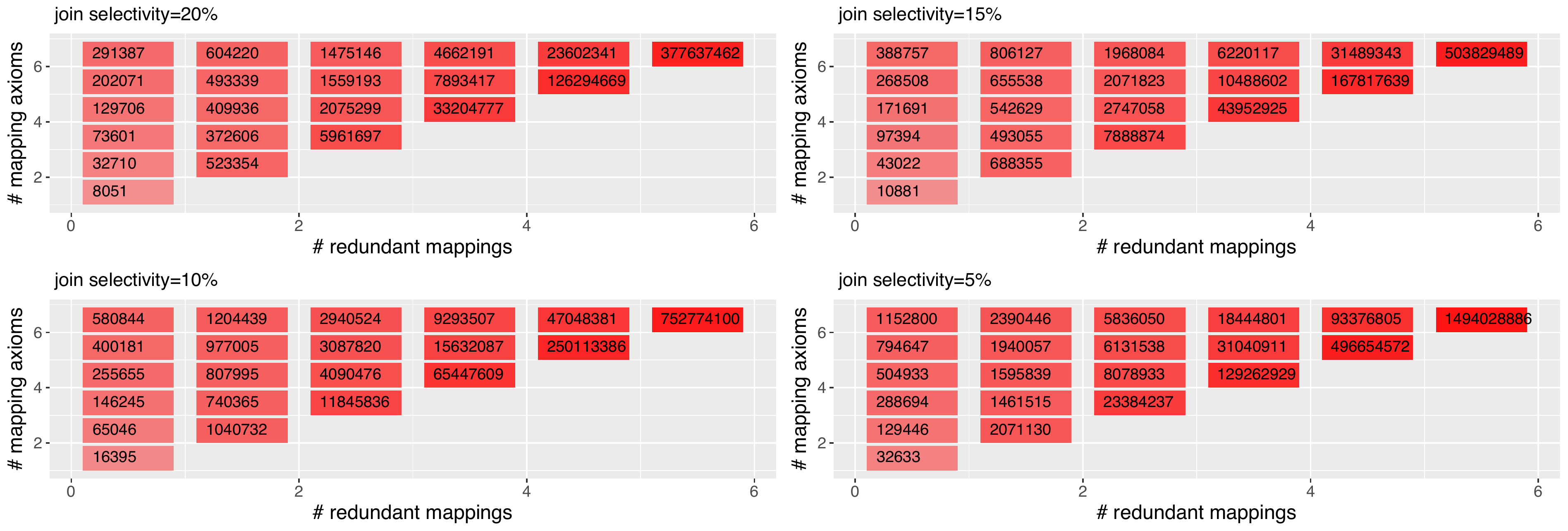}
  \caption{JUCQs: (PostgreSQL estimated cardinality) / (real cardinality)}
  \label{fig:postgres-jucq-est-vs-real-4}
\end{figure}

\subsection{NPD Experiment}

The goal of this experiment is to verify that cost-based techniques
can improve the performance of query answering over real-world queries
and instances. This test is carried on the original real-world
instance (as opposed to the scaled data instances) of the NPD
benchmark~\cite{LRXC15} for OBDA systems. We pick the three most
challenging UCQ queries (namely $q_6$, $q_{11}$, and $q_{12}$) from
the query catalog, and create another even more difficult query
(called $q_{31}$) by combining $q_6$ and $q_9$. Query
$q_{31}$, in the listing below, retrieves information regarding
wellbores (from $q_6$) and their related facilities (from $q_9$).

\noindent \begin{minipage}{\textwidth}
\begin{lstlisting}
SELECT DISTINCT ?fn ?c ?id ?wlb ?year ?comp WHERE {
  ## Fragment 1
  ?f a npdv:Facility ; npdv:name ?fn ; npdv:regInCountry ?c; npdv:idNPD ?id .
  ?w npdv:productionFacility ?f.

  ## Fragment 2
  ?w rdf:type npdv:Wellbore ; npdv:name ?wlb ; npdv:wellbCYear ?year;
     npdv:drillOpComp  [ npdv:name ?comp ]
FILTER (?id > "0"^^xsd:integer) }
\end{lstlisting}
\end{minipage}
In Table~\ref{tab:eval-npd}, we show the evaluation results over the
NPD benchmark for UCQs and JUCQs.
The unfoldings of JUCQs are constructed using cover queries of 2 fragments, each guided by our cost model. We observe that the sizes of the
unfoldings of JUCQs, measured in number of CQs, are sensibly
smaller than the size of the unfoldings of UCQs. Finally, we observe that the unfoldings of the JUCQ version of the considered queries improve the running times up to a factor of 34.

\begin{table}[t!]
%    \footnotesize
\scriptsize
    \centering
    \caption{Evaluation over the NPD benchmark}
\label{tab:eval-npd}
\begin{tabular}{cc  @{\extracolsep{8pt}} cc  @{\extracolsep{8pt}} ccc}
  \toprule
  \multicolumn{2}{c}{SPARQL Query} & \multicolumn{2}{c}{Unfolding of UCQs}  &  \multicolumn{3}{c}{Unfolding of JUCQs} \\
  \cline{1-2}\cline{3-4}\cline{5-7}
  name & \# triple patterns & time (s) & \# CQs & time (s) & \# Frags & \# CQs \\
  \midrule
  $q_6$ & 7 & 2.18 & 48 & 1.20 & 2 & 14 \\
  $q_{11}$ & 8 & 3.39 & 24 & 0.40 & 2 & 12 \\
  $q_{12}$ &  10 & 6.67 & 48 & 0.47 & 2 & 14  \\
  $q_{31}$ & 10 & 54.27 & 3840 &  1.58 & 2 & 327 \\
%  $q_{32}$ & 12 & 64.27 & 8000 &  0.58 & 3 & 90 \\
  \bottomrule
\end{tabular}
\end{table}

\section{Conclusion and Future Work}\label{c:planning:s:conclusion}

In this paper, we have studied the problem of finding efficient alternative translations of a user query in OBDA.
Specifically, we introduced a translation of JUCQ queries that preserves the JUCQ structure while maintaining the possibility of performing joins over database values, rather than URIs constructed by applying mappings definitions.
We devised a cost model based on a novel cardinality estimation, for estimating the cost of evaluating a translation of a UCQ or JUCQ over the database. We compared different translations on both a synthetic and fully customizable scenario based on the Wisconsin Benchmark and on a real-world scenario from the NPD Benchmark. In these experiments we have observed that (i)~our approach based on JUCQ queries can produce translations that are orders of magnitude more efficient than traditional translations into UCQs, and that (ii)~the cost model we devised is faithful to the actual query evaluation cost, and hence is well suited to select the best translation.

As future work, we plan to implement our techniques in the state-of-the-art OBDA system \ontop and to integrate them with existing optimization strategies. This will allow us to test our approach in more and diversified settings. We also plan to explore alternatives beyond JUCQs. An interesting line of research would be on the problem of relaxing the uniformity assumption made in our cost estimator, by integrating our model with existing techniques based on histograms. Finally, we plan to improve our approach for dealing with projection in an unfolding. In fact, the approach proposed here relies on a formula that has size exponential in the number of queries in the unfolding. To overcome this problem, we plan to remove such formula from our estimation by relying on additional statistics provided by the mapping and the ontology, such as the cardinality of the projection over one of the two arguments of a property defined in the wrap of a T-mapping.

%\newpage

%\FloatBarrier

\bibliographystyle{plain}

% \bibliography{main}

\begin{thebibliography}{10}

\bibitem{AbHV95}
Serge Abiteboul, Richard Hull, and Victor Vianu.
\newblock {\em Foundations of Databases}.
\newblock Addison Wesley Publ.\ Co., 1995.

\bibitem{BOSX13}
Meghyn Bienvenu, Magdalena Ortiz, Mantas Simkus, and Guohui Xiao.
\newblock Tractable queries for lightweight description logics.
\newblock In {\em Proc.\ of the 23rd Int.\ Joint Conf.\ on Artificial
  Intelligence (IJCAI)}. IJCAI/AAAI, 2013.

\bibitem{BursztynGM15b}
Damian Bursztyn, Fran{\c c}ois Goasdou{\'e}, and Ioana Manolescu.
\newblock Efficient query answering in {DL-Lite} through {FOL} reformulation
  (extended abstract).
\newblock In {\em Proc.\ of the Int.\ Workshop on Description Logic (DL)},
  volume 1350 of {\em CEUR Workshop Proceedings,
  {\upshape\protect\url{http://ceur-ws.org/}}}. CEUR-WS.org, 2015.

\bibitem{BursztynGM15}
Damian Bursztyn, Fran{\c c}ois Goasdou{\'e}, and Ioana Manolescu.
\newblock Reformulation-based query answering in {RDF:} alternatives and
  performance.
\newblock {\em Proc.\ of the VLDB Endowment}, 8(12):1888--1891, 2015.

\bibitem{DBLP:journals/pvldb/BursztynGM16}
Damian Bursztyn, Fran{\c{c}}ois Goasdou{\'{e}}, and Ioana Manolescu.
\newblock Teaching an {RDBMS} about ontological constraints.
\newblock {\em Proc.\ of the VLDB Endowment}, 9(12):1161--1172, 2016.

\bibitem{ontop-system}
Diego Calvanese, Benjamin Cogrel, Sarah Komla{-}Ebri, Roman Kontchakov, Davide
  Lanti, Martin Rezk, Mariano Rodriguez{-}Muro, and Guohui Xiao.
\newblock Ontop: Answering {SPARQL} queries over relational databases.
\newblock {\em Semantic Web J.}, 8(3):471--487, 2017.

\bibitem{CDLLR07}
Diego Calvanese, Giuseppe De~Giacomo, Domenico Lembo, Maurizio Lenzerini, and
  Riccardo Rosati.
\newblock Tractable reasoning and efficient query answering in description
  logics: The \textit{DL-Lite} family.
\newblock {\em J.\ of Automated Reasoning}, 39(3):385--429, 2007.

\bibitem{W3Crec-R2RML}
Souripriya Das, Seema Sundara, and Richard Cyganiak.
\newblock {R2RML}: {RDB} to {RDF} mapping language.
\newblock {W3C} {R}ecommendation, World Wide Web Consortium, September 2012.
\newblock Available at \protect\url{http://www.w3.org/TR/r2rml/}.

\bibitem{DeWitt93}
David~J. DeWitt.
\newblock The {W}isconsin benchmark: Past, present, and future.
\newblock In {\em The Benchmark Handbook for Database and Transaction Systems}.
  Morgan Kaufmann, 2 edition, 1993.

\bibitem{DLLM*13}
Floriana Di~Pinto, Domenico Lembo, Maurizio Lenzerini, Riccardo Mancini,
  Antonella Poggi, Riccardo Rosati, Marco Ruzzi, and Domenico~Fabio Savo.
\newblock Optimizing query rewriting in ontology-based data access.
\newblock In {\em Proc.\ of the 16th Int.\ Conf.\ on Extending Database
  Technology (EDBT)}, pages 561--572. ACM Press, 2013.

\bibitem{DBLP:conf/vldb/GardarinST96}
Georges Gardarin, Fei Sha, and Zhao{-}Hui Tang.
\newblock Calibrating the query optimizer cost model of iro-db, an
  object-oriented federated database system.
\newblock In {\em VLDB'96, Proceedings of 22th International Conference on Very
  Large Data Bases, September 3-6, 1996, Mumbai (Bombay), India}, pages
  378--389, 1996.

\bibitem{GKKP*14}
Georg Gottlob, Stanislav Kikot, Roman Kontchakov, Vladimir~V. Podolskii, Thomas
  Schwentick, and Michael Zakharyaschev.
\newblock The price of query rewriting in ontology-based data access.
\newblock {\em Artificial Intelligence}, 213:42--59, 2014.

\bibitem{W3Crec-SPARQL-1.1-query}
Steve Harris and Andy Seaborne.
\newblock {SPARQL}~1.1 query language.
\newblock {W3C} {R}ecommendation, World Wide Web Consortium, March 2013.
\newblock Available at \protect\url{http://www.w3.org/TR/sparql11-query}.

\bibitem{KKPZ12b}
Stanislav Kikot, Roman Kontchakov, Vladimir~V. Podolskii, and Michael
  Zakharyaschev.
\newblock Exponential lower bounds and separation for query rewriting.
\newblock In {\em Proc.\ of the 39th Int.\ Coll.\ on Automata, Languages and
  Programming (ICALP)}, volume 7392 of {\em Lecture Notes in Computer Science},
  pages 263--274. Springer, 2012.

\bibitem{KiKZ12}
Stanislav Kikot, Roman Kontchakov, and Michael Zakharyaschev.
\newblock Conjunctive query answering with {OWL~2~QL}.
\newblock In {\em Proc.\ of the 13th Int.\ Conf.\ on the Principles of
  Knowledge Representation and Reasoning (KR)}, pages 275--285, 2012.

\bibitem{W3Crec-RDF-Concepts}
Graham Klyne and Jeremy~J. Carroll.
\newblock {R}esource {D}escription {F}ramework ({RDF}): Concepts and abstract
  syntax.
\newblock {W3C} {R}ecommendation, World Wide Web Consortium, February 2004.
\newblock Available at \protect\url{http://www.w3.org/TR/rdf-concepts/}.

\bibitem{LRXC15}
Davide Lanti, Martin Rezk, Guohui Xiao, and Diego Calvanese.
\newblock The {NPD} benchmark: Reality check for {OBDA} systems.
\newblock In {\em Proc.\ of the 18th Int.\ Conf.\ on Extending Database
  Technology (EDBT)}, pages 617--628. OpenProceedings.org, 2015.

\bibitem{W3Crec-OWL2-Profiles}
Boris Motik, Bernardo Cuenca~Grau, Ian Horrocks, Zhe Wu, Achille Fokoue, and
  Carsten Lutz.
\newblock {OWL~2} {W}eb {O}ntology {L}anguage profiles (second edition).
\newblock {W3C} {R}ecommendation, World Wide Web Consortium, December 2012.
\newblock Available at \protect\url{http://www.w3.org/TR/owl2-profiles/}.

\bibitem{PLCD*08}
Antonella Poggi, Domenico Lembo, Diego Calvanese, Giuseppe De~Giacomo, Maurizio
  Lenzerini, and Riccardo Rosati.
\newblock Linking data to ontologies.
\newblock {\em J.\ on Data Semantics}, X:133--173, 2008.

\bibitem{RoCa12}
Mariano Rodriguez-Muro and Diego Calvanese.
\newblock High performance query answering over \textit{DL-Lite} ontologies.
\newblock In {\em Proc.\ of the 13th Int.\ Conf.\ on the Principles of
  Knowledge Representation and Reasoning (KR)}, pages 308--318, 2012.

\bibitem{RoKZ13}
Mariano Rodriguez-Muro, Roman Kontchakov, and Michael Zakharyaschev.
\newblock Ontology-based data access: {O}ntop of databases.
\newblock In {\em Proc.\ of the 12th Int.\ Semantic Web Conf.\ (ISWC)}, volume
  8218 of {\em Lecture Notes in Computer Science}, pages 558--573. Springer,
  2013.

\bibitem{RoRe15}
Mariano Rodriguez-Muro and Martin Rezk.
\newblock Efficient {SPARQL-to-SQL} with {R2RML} mappings.
\newblock {\em J.\ of Web Semantics}, 33:141--169, 2015.

\bibitem{ultrawrap}
Juan~F. Sequeda and Daniel~P. Miranker.
\newblock {U}ltrawrap: {SPARQL} execution on relational data.
\newblock {\em J.\ of Web Semantics}, 22:19--39, 2013.

\bibitem{SiKS05}
Abraham Silberschatz, Henry~F. Korth, and S.~Sudarshan.
\newblock {\em Database System Concepts, 5th Edition}.
\newblock McGraw-Hill Book Company, 2005.

\bibitem{SwSc94}
Arun Swami and K.~Bernhard Schiefer.
\newblock On the estimation of join result sizes.
\newblock In {\em Proc.\ of the 4th Int.\ Conf.\ on Extending Database
  Technology (EDBT)}, volume 779 of {\em Lecture Notes in Computer Science},
  pages 287--300. Springer, 1994.

\end{thebibliography}

%%%%%%%% END %%%%%%

\end{document}